\documentclass[letterpaper,11pt]{article}

\usepackage{fullpage}  
\usepackage[T1]{fontenc}
\usepackage{lmodern}
\usepackage{authblk}
\usepackage{amsfonts}
\usepackage{amsmath}
\usepackage{amssymb}
\usepackage{amsthm}
\usepackage{algorithm}
\usepackage{algpseudocode}
\usepackage{graphicx}
\usepackage[colorlinks=true,linkcolor=teal,citecolor=teal,urlcolor=teal,plainpages=false,pdfpagelabels]{hyperref}
\usepackage[capitalize,noabbrev]{cleveref}
\usepackage{color}
\usepackage[dvipsnames]{xcolor}
\usepackage{mathtools}
\usepackage{bbm}
\usepackage[square,numbers,sort&compress]{natbib}
\usepackage{float}
\usepackage{array}
\usepackage{verbatim}
\usepackage{enumitem}
\usepackage{abstract}
\usepackage{tocloft}
\usepackage[nottoc]{tocbibind}
\usepackage[font=small,labelfont=sc,margin=1cm]{caption}
\usepackage{dsfont}
\usepackage[normalem]{ulem}
\usepackage{booktabs}

\usepackage{anyfontsize}  
\usepackage{microtype} 
\usepackage{pifont}

\appto{\bibsetup}{\sloppy} 

\allowdisplaybreaks

\makeatletter
\g@addto@macro\bfseries{\boldmath}
\makeatother

\newcommand{\bra}[1]{\langle #1|}
\newcommand{\ket}[1]{|#1\rangle}

\newcommand{\braket}[2]{\langle #1|#2\rangle}
\newcommand{\ketbra}[2]{\ket{#1}\!\bra{#2}}

\newcommand{\inner}[2]{\langle #1,#2\rangle}

\newcommand{\abs}[1]{|#1|}
\newcommand{\Abs}[1]{\left|#1\right|}
\newcommand{\e}{\mathrm{e}}
\newcommand{\I}{\mathrm{i}}
\newcommand{\Tr}{\mathrm{Tr}}
\renewcommand{\t}{{\scriptscriptstyle\mathsf{T}}}
\newcommand{\id}{\operatorname{id}}
\newcommand{\conj}[1]{\overline{#1}}
\newcommand{\norm}[1]{\lVert#1\rVert}

\renewcommand{\Vec}[1]{\boldsymbol{#1}}

\newcommand{\Lin}{\mathrm{L}}

\DeclareMathOperator*{\argmin}{arg\,min}

\numberwithin{equation}{section}  

\makeatletter
\newcommand{\pushright}[1]{\ifmeasuring@#1\else\omit\hfill$\displaystyle#1$\fi\ignorespaces}
\newcommand{\pushleft}[1]{\ifmeasuring@#1\else\omit$\displaystyle#1$\hfill\fi\ignorespaces}
\makeatother

\theoremstyle{definition}
\newtheorem{theorem}{Theorem}
\newtheorem{lemma}[theorem]{Lemma}
\newtheorem{corollary}[theorem]{Corollary}
\newtheorem{definition}[theorem]{Definition}

\newtheorem{proposition}[theorem]{Proposition}
\newtheorem{remark}[theorem]{Remark}
\newtheorem{problem}{Problem}

\renewcommand{\qedsymbol}{$\blacksquare$}

\renewcommand{\qedsymbol}{\unskip\nobreak\quad\qedsymbol}
\renewcommand{\qedsymbol}{$\blacksquare$}
\newcommand{\qedgen}{$\blacktriangleleft$}

\definecolor{dred}  {RGB}{164,12,52}

\title{Online learning of quantum processes}

\author[1]{Asad Raza}
\author[1]{Matthias C. Caro}
\author[1,2]{Jens Eisert}
\author[1]{Sumeet Khatri}

\affil[1]{Dahlem Center for Complex Quantum Systems, Freie Universit\"at Berlin, Berlin, Germany}
\affil[2]{Helmholtz-Zentrum Berlin für Materialien und Energie, Berlin, Germany}

\date{\today}

\begin{document}

\maketitle

\begin{abstract}
    Among recent insights into learning quantum states, online learning and shadow tomography procedures are notable for their ability to accurately predict expectation values even of adaptively chosen observables.
    In contrast to the state case, quantum process learning tasks with a similarly adaptive nature have received little attention.
    In this work, we investigate online learning tasks for quantum processes.
    Whereas online learning is infeasible for general quantum channels, we show that channels of bounded gate complexity as well as Pauli channels can be online learned in the regret and mistake-bounded models of online learning. 
    In fact, we can online learn probabilistic mixtures of any exponentially large set of known channels.
    We also provide a provably sample-efficient shadow tomography procedure for Pauli channels. 
    Our results extend beyond quantum channels to non-Markovian multi-time processes, with favorable regret and mistake bounds, as well as a shadow tomography procedure.
    We complement our online learning upper bounds with mistake as well as computational lower bounds.
    On the technical side, we make use of the multiplicative weights update algorithm, classical adaptive data analysis, and Bell sampling, as well as tools from the theory of quantum combs for multi-time quantum processes.
    Our work initiates a study of online learning for classes of quantum channels and, more generally, non-Markovian quantum processes. Given the importance of online learning for state shadow tomography, this may serve as a step towards quantum channel variants of adaptive shadow tomography.
\end{abstract}

\section{Introduction}\label{sec-introduction}

Learning about quantum systems and their evolution over time is a task of fundamental importance in quantum physics. ``Learning'', broadly speaking, refers to the extraction of useful classical information from quantum-mechanical systems and their evolution through experiments. Such learning tasks first appeared in quantum information in the form of quantum state and process tomography, in which the aim is to extract all classical information about the system, in terms of the density matrix of the system in the case of quantum state tomography~\cite{BenchmarkingReview,hradil1997quantumstateestimation,MPS03,gross2010tomographycompressedsensing,blumeKohout2010reliableestimation,banaszek2013quantumtomography}, and the transfer matrix in the case of quantum process tomography~\cite{chuang1997processtomography,dariano2001processtomography,mohseni2008processtomography}. These tomographic tasks, with their strict measures of performance in terms of worst-case distance measures such as the trace and diamond norm, require resources scaling exponentially with the system size~\cite{haah2016sample, o2016efficient, chen2022tight, haah2023query, oufkir2023sample}. Consequently, recent years have seen a growing interest in less strict variants of state and process learning, defined by relaxing the requirement of extracting full classical information about the objects of interest and instead requiring that we learn only the values of certain observables of our state or process. In the case of quantum states, notable such learning tasks include ``pretty good tomography'' in the spirit of \textit{probably approximately correct} (PAC) learning~\cite{Aaronson07}, shadow tomography~\cite{Aar17, badescu2021improved, HKP21, king2024triply}, online learning~\cite{ACH+19, chen2022adaptive}, and classical shadows~\cite{HKP20, elben2022randomized,ONE2013}. Inspired by this progress in understanding state learning, also new perspectives on quantum channel learning have been explored~\cite{HBC+21, LLC21, KTCT21, huang2022learning, caro2022learning, angrisani2023learning, wadhwa2023learning, zhao2023learning}.

\begin{figure}
    \centering
    \includegraphics[scale=0.87]{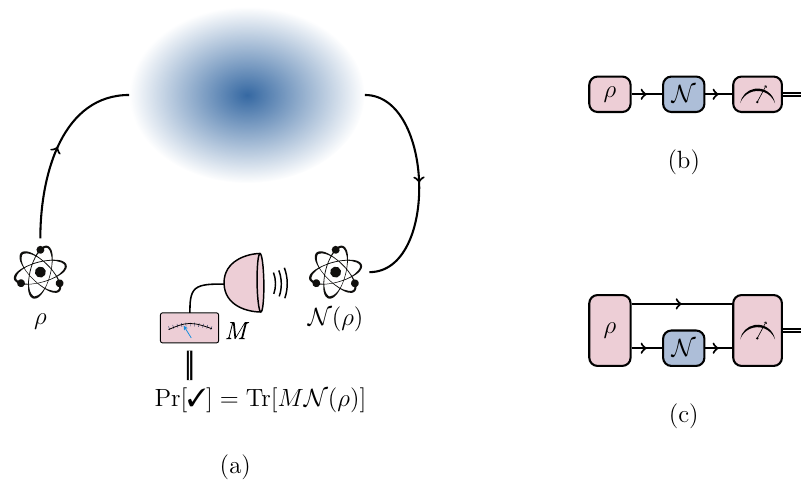}
    \caption{\textbf{Learning of quantum processes.} (a) To learn about the unknown evolution of a quantum system (symbolized by the blue shaded region and represented mathematically by the quantum channel $\mathcal{N}$), we prepare a probe quantum state $\rho$, let it evolve, and then measure it according to the POVM $\{M,\mathbbm{1}-M\}$. We encapsulate this process in the circuit diagram shown in (b). (c) More generally, we can prepare an entangled probe state of two systems, let only one of them evolve, and then jointly measure both systems. Our results apply to this more general class of tests, and also more generally to classes of multi-time quantum processes, in which the unknown evolution could be non-Markovian.}
    \label{fig:channel_tests}
\end{figure}

We consider the following channel learning task, see also Figure~\ref{fig:channel_tests}. Let $\mathcal{N}$ be an unknown quantum channel, describing the evolution of a quantum system. In order to learn the behavior of the channel, we can prepare our system in a state of our choice, let it evolve according to the channel $\mathcal{N}$, and then measure the output. The input state and measurement constitute a ``test'' for the quantum channel, and the probability of ``passing'' the test is given by $\Tr[M\mathcal{N}(\rho)]$, where $\rho$ is the input state and $M$ is the measurement operator corresponding to passing the test. Tests of this type are ubiquitous in real-world physical setups used to learn the dynamics of quantum systems, see, e.g., Refs.~\cite{chuang1997processtomography,mohseni2008processtomography,PhysRevLett.97.220407,Martinis,RandomSequences}. Our task is to accurately predict values of quantities of the form $f(x)=\Tr[M\mathcal{N}(\rho)]$ for test pairs $x=(\rho,M)$, and in this sense ``learn'' the evolution of the system.

In the typical setting of supervised learning used to achieve our task of interest, there is a \emph{training phase}, in which a set $\{(x,f(x))\}_x$ of pairs of tests and their passing probabilities is given to the learner, and the learner uses these tests to form a hypothesis for the unknown channel. This hypothesis should accurately predict passing probabilities on new, unseen tests, typically drawn from the same distribution that generated the tests during training. Learning algorithms in this setting are analyzed within the framework of PAC learning~\cite{valiant1984theory}.
In this work, we go beyond the setting of PAC learning. Instead of the data pairs $(x,f(x))$ being given to the learner in a batch, we suppose that they are given to the learner sequentially, one by one, and perhaps even in an adaptive and adversarial manner. The learner must produce a hypothesis for the unknown channel at every step, using which they estimate the passing probability of the test pair. Upon learning the true passing probability, they can update their hypothesis. The goal now is for the learner to devise a sequence of hypotheses such that, over time, they make few mistakes in their estimates; see Section~\ref{sec-framework} for a more formal description of this setting.

The framework of \textit{online learning}~\cite{littlestone1988learning} (see, e.g., Refs.~\cite[Chapter~21]{SBD14} and \cite[Chapter~8]{MRT18} for pedagogical introductions) has been developed precisely to address this arguably more realistic learning setting. After all, data will often be processed sequentially, and it only makes sense to update the hypothesis step by step. Indeed, the importance of online learning derives from its ability to describe scenarios in which data is presented to the learner sequentially and adaptively. Thereby, it removes assumptions on the data-generating process, such as the i.i.d.~assumption typical in PAC learning~\cite{valiant1984theory}. As such, it provides a more stringent type of learning compared to PAC learning. In fact, broadly speaking, it has been shown that online learning implies PAC learning~\cite{kearns1987recent, littlestone1989online, gretta2023improved,cesabianchi2001generalization, cesabianchi2004generalization}. Furthermore, if the learning algorithm is allowed unbounded computational resources, then something stronger holds: any concept class of Boolean functions is learnable in the online model if and only if is also learnable in the (distribution-free) PAC model~\cite{blum1994separating}.

While online learning of quantum states has been considered already~\cite{ACH+19,chen2022adaptive}, and it has also been lifted to shadow tomography of quantum states with adaptively chosen observables~\cite{Aar17,badescu2021improved}, the overwhelming majority of results on quantum process learning so far do not allow for accurate predictions based on an adaptive choice of the state-measurement pairs. To fill this gap, we initiate the study of online learning for quantum processes. 
We first show that general quantum channels cannot be online learned with subexponential regret or number of mistakes. However, a priori knowledge about the complexity or the structure of the unknown channel can make online learning feasible. Indeed, we identify two physically relevant classes of channels---efficiently implementable channels and Pauli channels---that can be online learned with regret and mistake bounds scaling polynomially in the system size. We extend these results to classes of more general multi-time processes, in particular quantum processes that are non-Markovian, establishing that they, too, can be online learned with regret and mistake bounds scaling polynomially in the system size.

\subsection{Statement of the problem}\label{sec-framework}
Consider an interaction between a learner and an adversary. At time step $t\in\mathbb{N}$, the adversary picks a state-measurement pair, $(\rho^{(t)}, M^{(t)})$, where $\rho^{(t)}$ is an input state and $M^{(t)}$ is an effect operator of a two-outcome POVM. More generally, we can allow for states and measurements with arbitrary auxiliary systems, as in Figure~\ref{fig:channel_tests}(c) (see \Cref{subsection:prob-statement-online-learning-qchannels}). The task for the learner is to predict $\Tr[M^{(t)}\mathcal{N}(\rho^{(t)})]$ for an unknown channel $\mathcal{N}$. 
To do so, the learner produces their own channel hypothesis, $\mathcal{N}^{(t)}$, and outputs $\Tr[M^{(t)}\mathcal{N}^{(t)}(\rho^{(t)})]$ as their prediction. The adversary then provides the learner with feedback on what would have been the correct expectation value\footnote{In general, the adversary may only provide an approximation to the true expectation value (see \Cref{sec-online_learning_basics}), but here we restrict our attention to the simpler version for ease of exposition in the introduction.}.
The goal of the learner is to ensure that their output values are not too far from correct in most rounds of the interaction. We can quantify the loss suffered by the learner at time step $t$ by the absolute difference between learner’s estimate of the expectation value and the correct expectation value, 
\begin{equation}\label{eq-loss_function_intro}
    \ell(\mathcal{N}^{(t)},\rho^{(t)},M^{(t)})\coloneqq \Abs{\Tr[M^{(t)}\mathcal{N}(\rho^{(t)})]-\Tr[M^{(t)}\mathcal{N}^{(t)}(\rho^{(t)})]}.
\end{equation}
The learner-adversary interaction proceeds for a total of $T\in\mathbb{N}$ rounds.

Now, how do we measure the learner's performance over the course of $T$ rounds of its interaction with the adversary? Note that due to the adversarial nature of the problem, statistical extrapolations are of little use. Indeed, as soon as the learner models the interaction using some probability distribution, the adversary can immediately change their strategy to make the learner fail. One way to gauge the learner’s performance is to compare their total loss at the end of the $T$ rounds of interaction with the loss that they would have incurred if they were allowed to make all predictions at the end of the $T$ rounds, after having seen all of the state-measurement pairs. We do this by considering the quantity
\begin{equation}
    R_T \coloneqq \sum_{t=1}^T  \ell (\mathcal{N}^{(t)}, \rho^{(t)}, M^{(t)}) -  \min_{\mathcal{N}}  \sum_{t=1}^T \ell(\mathcal{N}, \rho^{(t)}, M^{(t)}),
\end{equation}
where the minimization is over channels $\mathcal{N}$ from some class of interest. The larger this quantity, the more the learner would lament their choice of hypotheses $\mathcal{N}^{(t)}$ at the end of the $T$ rounds; hence, this quantity is called \textit{regret}. It is well-known~\cite{OCO_hazan_book} that any online learner suffers $\Omega(\sqrt{T})$ regret in general. We aim for online learners that saturate this lower bound and achieve a regret scaling as $\mathcal{O}(\sqrt{T \mathrm{poly}(\log D)})$, where $D$ is the dimension of the quantum system acted on by the channel.

Another intuitive way to evaluate whether the learner's performance is ``good'' is in terms of the number of rounds $t \in [T]$ in which the learner makes a mistake. By a ``mistake'', we mean that the learner's estimate of the expectation value is more than a given accuracy $\varepsilon$ away (in absolute-value distance) from the correct expectation value revealed by the adversary. In other words, the learner should minimize the number of rounds in which $\ell(\mathcal{N}^{(t)},\rho^{(t)},M^{(t)}) > \varepsilon$. More formally, we say that the learner makes an \textit{$\varepsilon$-mistake} in round $t$ if $\ell(\mathcal{N}^{(t)},\rho^{(t)},M^{(t)}) > \varepsilon$.
Viewed this way, the goal in online learning is to upper bound the number of $\varepsilon$-mistakes for any number $T$ of rounds and any adversarial/adaptive choice of state-measurement pairs presented to the learner. This is the so-called mistake-bounded model of Littlestone \cite{littlestone1988learning}. In our work, we more specifically are interested in online learners that incur a mistake bound scaling logarithmically with $D$ and inverse polynomially in the accuracy $\varepsilon$ for, ideally, a low-degree polynomial.

\subsection{Overview of the main results}\label{sec:overview_main_results}

For the task of online learning arbitrary $n$-qubit quantum states, it was shown in Ref.~\cite{ACH+19} that there exist procedures that make at most linearly-in-$n$ many mistakes. (Here, for ease of presentation, we consider $\varepsilon$ to be a constant, say $1/3$, and then simply speak of a mistake instead of an $\varepsilon$-mistake.) In terms of channel learning, this implies the same mistake bound for online learning arbitrary $n$-qubit state-preparation channels. However, as pointed out in Ref.~\cite[Footnote 18]{Aaronson07} and as we formalize further in \Cref{subsection:mistake-lower-bounds-general}, if the underlying concept class consists of all $n$-qubit unitary channels, then any (even computationally unbounded) online channel learner can be forced to make exponentially-in-$n$ many mistakes. 
Thus, in contrast to the case of quantum states, we have to consider restricted classes of channels to achieve online channel learning with a polynomial-in-$n$ number of $\varepsilon$-mistakes. 
In fact, it was left open in Ref.~\cite{Aaronson07} to find restricted classes of quantum channels for which this goal, which is an online version of ``pretty good process tomography'', can be realized. While recent years have seen some progress on the batch version of this task, see, e.g., Refs.~\cite{CHY16, caro2020pseudo, popescu2021learning, cai2022sample}, the online case remains open.

In this work, we answer the question of pretty good process tomography within the online learning framework for the following two concrete classes of channels (see Table~\ref{tab:summary-of-results} for a summary of our results):
\begin{enumerate}
    \item Channels that can be implemented by dissipative quantum circuits with a limited number of local gates. Channels in this class can be regarded as having limited complexity.
    \item Pauli channels or, more generally, mixtures of a fixed set of (potentially exponentially many) \textit{known} channels, each of which could have arbitrarily high gate complexity. This is a structural assumption on the channel.
\end{enumerate}

\begin{table}
\begin{center}
\begingroup
\setlength{\tabcolsep}{10pt} 
\renewcommand{\arraystretch}{1.5} 
\begin{tabular}{|c|c|c|c|}
\hline
Online learning
& $G$-gate channels 
& Pauli channels \\ 
\hline\hline
Regret bounds
& $\begin{matrix}
    \mathcal{O}\left(L\sqrt{T G\log(Gn)}\right)\\ \text{(\Cref{theorem:bounded-complexity-regret})}
\end{matrix}$
& $ \begin{matrix}
    \mathcal{O}\left(L \sqrt{T n} \right)\\ \text{(\Cref{theorem:pauli-regret})}
\end{matrix} $
\\
\hline
Mistake bounds
& $\begin{matrix}
    \mathcal{O}\left(\frac{L^2 G\log(Gn)}{\varepsilon^2}\right)\\
    \text{(\Cref{corollary:bounded-complexity-mistake})}
\end{matrix}$
& $\begin{matrix}
    \mathcal{O}\left(\frac{L^2 n}{\varepsilon^2}\right)\\
    \text{(\Cref{corollary:pauli-mistake-bound})}
\end{matrix}$
\\
& $\begin{matrix}
    \Omega(\min\{2^n, G\})\\
    \text{(\Cref{corollary:mistake-lower-bound-bounded-complexity,corollary:mistake-lower-bound-pauli-channels})}
\end{matrix}$
& $\begin{matrix}
    \Omega(n)\\
    \text{(\Cref{corollary:mistake-lower-bound-pauli-channels})}
\end{matrix}$ 
\\
\hline
$\begin{matrix}
    \textrm{Computational} \\
    \textrm{complexity}
\end{matrix}$
& $\begin{matrix}
    \omega(\mathrm{poly}(n)) \textrm{ already for } G=\mathcal{O}(n\mathrm{polylog}(n))\\
    \text{(\Cref{corollary:computational-hardness-bounded-complexity})}
\end{matrix}$
& $\begin{matrix}
    \tilde{\Omega}(4^n)\\
    \text{(\Cref{theorem:computational-hardness-pauli-channel-online-learning})}
\end{matrix}$\\
\hline
\end{tabular}
\endgroup
\caption{\textbf{Overview of our main results.} For channels of gate complexity $G$ as well as for $n$-qubit Pauli channels, we give regret and mistake upper bounds that scale favorably with $G$ and $n$, respectively. Our bounds hold for general loss functions with Lipschitz constant $L$. Additionally, complementary mistake lower bounds show that the dependencies on $G$ and $n$ are almost optimal. Finally, we give computational complexity lower bounds for online learning either of the two classes of channels with polynomially many mistakes.}
\label{tab:summary-of-results}
\end{center}
\end{table}

First, we show that there exists an online learner whose regret and number of mistakes can be controlled in terms of the gate complexity of the class of channels to be learned.

\begin{theorem}[Online learning channels of bounded complexity---informal]\label{informal-theorem:bounded-complexity-online-learning}
    The class of $n$-qubit channels that can be implemented by circuits consisting of $G$ arbitrary two-qubit channels can be online learned with regret bound $ \mathcal{O}\left(\sqrt{T G\log(Gn)}\right)$ and with $\varepsilon$-mistake bound $\mathcal{O}\left(\frac{G\log(Gn)}{\varepsilon^2}\right)$.
\end{theorem}

\Cref{informal-theorem:bounded-complexity-online-learning} extends beyond the absolute-value loss function from \Cref{eq-loss_function_intro} to more general loss functions, in which case our bounds depend also on their Lipschitz constant $L$.
General channels require exponentially many gates to be implemented by a $2$-local circuit, so \Cref{informal-theorem:bounded-complexity-online-learning} does not provide useful guarantees in this case. 
However, if we focus on the physically relevant class of channels with a polynomial gate complexity, then \Cref{informal-theorem:bounded-complexity-online-learning} gives polynomial regret and mistake bounds.
Additionally, we show with an $\mathcal{O}(1)$-mistake lower bound of $\Omega(\min\{2^n, \sqrt{G}\})$ for general $G$ and of $\Omega(G)$ for $G\leq n$ that the $G$-dependence in the online learning guarantees of \Cref{informal-theorem:bounded-complexity-online-learning} cannot be significantly improved (\Cref{corollary:mistake-lower-bound-bounded-complexity}).

Second, we prove regret and mistake bounds for Pauli channel online learning that scale efficiently in the number of qubits. Pauli channels play an important role in quantum information theory, specifically in the field of quantum computation, where Pauli channel noise either naturally emerges or is achievable via group twirls~\cite{wallman2016randomizedcompiling}. Thus, the following result establishes adaptive/online learnability of an important class of quantum noise channels.

\begin{theorem}[Online learning Pauli channels---informal]\label{informal-theorem:pauli-channel-online-learning}
    The class of $n$-qubit Pauli channels can be online learned with regret bound $\mathcal{O}\left(\sqrt{T n} \right)$ and with $\varepsilon$-mistake bound $\mathcal{O}\left(\frac{n}{\varepsilon^2}\right)$.
\end{theorem}

The linear-in-$n$ scaling in the $\varepsilon$-mistake bound is optimal for constant $\varepsilon$ (\Cref{corollary:mistake-lower-bound-pauli-channels}). Moreover, establishing a connection to important notions from classical learning theory, we demonstrate that \Cref{informal-theorem:pauli-channel-online-learning} yields bounds on the (sequential) fat-shattering dimension of Pauli channels, and gives rise to a sample compression scheme for Pauli channels (\Cref{section:learning-theory-implications}).

\Cref{informal-theorem:bounded-complexity-online-learning,informal-theorem:pauli-channel-online-learning} give favorably scaling regret and mistake bounds. 
However, the respective online learning procedures are computationally inefficient. 
We show that, under reasonable cryptographic assumptions, this is unavoidable when aiming for good regret and mistake bounds in these online learning problems:

\begin{theorem}[Computational lower bounds for online learning---informal]\label{informal-theorem:computational-hardness}
    On the one hand, any online learner that makes at most $\mathcal{O}(\mathrm{poly}(n))$ many $(1/3)$-mistakes in online learning $n$-qubit Pauli channels has to use runtime exponential in $n$.
    On the other hand, assuming that \textsf{RingLWE} cannot be solved by classical polynomial-time algorithms, then already for $G=\mathcal{O}(n\mathrm{polylog}(n))$ there is no polynomial-time online learner that makes at most $\mathcal{O}(\mathrm{poly}(n))$ many $(1/3)$-mistakes in online learning $G$-gate $n$-qubit quantum channels.
\end{theorem}

While the computational inefficiency of our online learning procedures is of course undesirable, \Cref{informal-theorem:computational-hardness} shows that this is not a flaw of our specific procedures. Rather, it is simply not possible to computationally efficiently learn Pauli channels with a good mistake bound. And given that \textsf{RingLWE} is widely believed to be hard \cite{regev2009lattices, ananth2023revocable, aggarwal2023lattice}, then as soon as $G$ scales only slightly superlinearly in $n$, we also do not expect there to be any computationally efficient online learners that achieve a good regret for the class of $G$-gate channels. 

A problem related to online learning quantum processes is \textit{shadow tomography} of quantum processes. By analogy with shadow tomography of quantum states~\cite{Aar17,badescu2021improved}, shadow tomography of quantum processes is a stricter form of learning than mistake-bounded online learning: while in the online learning task the number of mistakes should be bounded, in shadow tomography one has to correctly estimate (i.e., with error $\varepsilon$) the expectation values of all the $M$ observables provided, with probability at least $1-\delta$. (See \Cref{problem:shadow-tomography} for a formal statement of the problem.) Furthermore, in shadow tomography, only quantum access to the channel is provided, while classical descriptions of tests and their passing probabilities (with respect to the unknown channel) are provided in online learning. The main observation underlying our proof of \Cref{informal-theorem:pauli-channel-online-learning} implies that techniques from classical adaptive data analysis~\cite{dwork2015adaptivedataanalysis, bassily2016adaptivedataanalysis} directly carry over to Pauli channel shadow tomography. In particular, we obtain the following result.

\begin{theorem}[Pauli channel shadow tomography---informal]\label{informal-corollary:pauli-channel-shadow-tomography}
    Shadow tomography of an arbitrary $n$-qubit Pauli channel can be solved using 
    \begin{equation}\label{eq-shadow_tomography_pauli_main}
        k = \mathcal{O}\left(\frac{\sqrt{n}\log (M) \log^{3/2}((\varepsilon \delta)^{-1}) }{\varepsilon^3}\right)
    \end{equation}
    copies of the channel. The strategy runs in time $\mathrm{poly}(4^n, k)$ per channel use.
\end{theorem}

We leverage \Cref{informal-corollary:pauli-channel-shadow-tomography} to make a more general statement about shadow tomography of arbitrary channels. In particular, we show in \Cref{corollary:shadow_tomography_channels} that we can solve the shadow tomography problem for an arbitrary quantum channel $\mathcal{N}$ with a number of copies scaling as in \eqref{eq-shadow_tomography_pauli_main}, with $\varepsilon$ therein replaced by $\varepsilon-\frac{1}{2}\norm{\mathcal{N}-\mathcal{N}^{\mathsf{P}}}_{\diamond}$, for all $\varepsilon>\frac{1}{2}\norm{\mathcal{N}-\mathcal{N}^{\mathsf{P}}}_{\diamond}$. Here $\mathcal{N}^{\mathsf{P}}$ is the Pauli-twirled version of $\mathcal{N}$.

\subsection{Extensions of our results}

\begin{figure}
    \centering
    \includegraphics[scale=1]{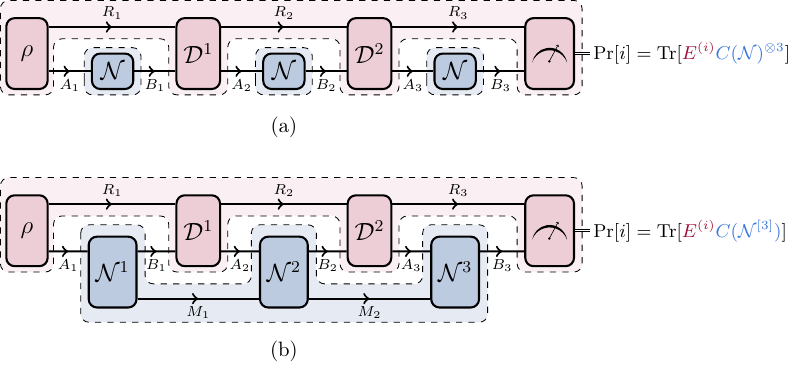}
    \caption{\textbf{Extensions of our results to general quantum processes.} (a) Going beyond one use of a channel $\mathcal{N}$, as shown in Figure~\ref{fig:channel_tests}, we may want to learn the value of the channel on tests that make multiple, adaptive uses of the channel. Shown are three independent uses of $\mathcal{N}$, whose Choi representation is $C(\mathcal{N})$. (b) We can similarly perform adaptive tests of a non-Markovian process $\mathcal{N}^{[3]}$, characterized by the blue quantum comb, with Choi representation $C(\mathcal{N}^{[3]})$. The generalized Born rule~\cite{CDP09} tells us that the outcome probabilities of measurements, or ``tests'', of quantum channels and multi-time quantum processes can be determined by an analogue of the usual Born rule for quantum states, in which the Choi representation takes the place of the quantum state, and the test is characterized by operators $E^{(i)}$ that are generalizations of effect operators for quantum states. (See Section~\ref{sec-basics_quantum_info} for details.)
    }\label{fig-channel_test_gen}
\end{figure}

Many of the techniques underlying our main results in Section~\ref{sec:overview_main_results} can be applied to more general settings, going beyond $G$-gate channels and Pauli channels. We now summarize these extensions.

\paragraph*{Convex mixtures of known channels.} While \Cref{informal-theorem:pauli-channel-online-learning} is phrased for Pauli channels, we in fact show the following more general statement: If we consider a class of channels that can be written as probabilistic mixtures of a fixed set of $K$ known quantum channels, then we can achieve online learning for this class with regret and number of mistakes bounded in terms of $\log (K)$. Notably, these bounds apply even if channels with high circuit complexity occur in the mixtures.

\begin{theorem}[Online learning convex mixtures of known channels---informal]\label{informal-theorem:convex-mixture-channel-online-learning}
    Given an arbitrary set of $K>0$ known and fixed quantum channels, any convex mixture of these channels can be online learned with regret bound $\mathcal{O}\left(\sqrt{T \log K}\right)$ and $\varepsilon$-mistake bound $\mathcal{O}(\frac{\log(K)}{\varepsilon^2})$.
\end{theorem}

\paragraph*{Adaptive tests of channels.} We may also want to learn the passing probabilities of more general channel tests that make use of the channel multiple times, perhaps adaptively. Such tests have the form shown in Figure~\ref{fig-channel_test_gen}(a). We can directly import results about convex mixtures of known channels to this setting. Indeed, given a channel $\mathcal{N}=\sum_{j=1}^K p_j\mathcal{N}_j$ that is a convex mixture of known channels $\mathcal{N}_j$, it holds that
\begin{equation}
    \mathcal{N}^{\otimes k}=\sum_{j_1,j_2,\dotsc,j_k=1}^K p_{j_1}p_{j_2}\dotsb p_{j_k}\mathcal{N}_{j_1}\otimes\mathcal{N}_{j_2}\otimes\dotsb\otimes\mathcal{N}_{j_k}\, .
\end{equation}
This is itself a convex mixture of $K^k$ known channels, so Theorem~\ref{informal-theorem:convex-mixture-channel-online-learning} applies and yields regret and mistake bounds scaling with $k\log(K)$. We note, however, that Theorem~\ref{informal-theorem:convex-mixture-channel-online-learning} in this scenario will generally not give a \emph{proper} online learner. Namely, the learner's hypotheses will be $(nk)$-qubit channels given by convex combinations of $\{\mathcal{N}_{j_1}\otimes\mathcal{N}_{j_2}\otimes\dotsb\otimes\mathcal{N}_{j_k}\}_{j_1,\ldots, j_k=1}^K$, which in general cannot be factorized into $k$ copies of a single $n$-qubit channel.

\paragraph*{Non-Markovian quantum processes.} At the heart of the quantities that we aim to estimate is Born's rule~\cite{born1926BornRule}, which tells us that the expected value of an observable (Hermitian operator) $H$ for a quantum state $\rho$ is given by $\Tr[H\rho]$. The Born rule generalizes not only to quantum channels but also to multi-time/non-Markovian quantum processes~\cite{Ziman08,CDP09}, where we model multi-time processes mathematically as ``quantum combs''~\cite{CDP09}, also called ``quantum strategies''~\cite{GW07}; see also Refs.~\cite{pollock2018nonmarkovian,MM20,berk2021multitimeprocesses}. Within the framework of quantum combs, the Choi representation of the process takes the place of the state $\rho$ and the observable $H$ is replaced by a generalized ``process observable'' $O$; see Figure~\ref{fig-channel_test_gen}(b). We provide formal definitions of quantum combs and process observables in Section~\ref{sec-basics_quantum_info}. Consequently, we can readily translate our results on bounded complexity channels and convex mixtures of known channels to bounded complexity non-Markovian processes and convex mixtures of known non-Markovian processes. While full tomography of multi-time/non-Markovian processes has been considered~\cite{white2022nonMarkoviantomography}, to the best of our knowledge, the restricted tomographic setting that we consider here has so far not been considered for multi-time quantum processes.

We start by considering multi-time processes of bounded complexity. We formally define these processes by analogy with quantum channels of bounded gate complexity in Section~\ref{sec-online_learning_multi_time_processes}. We can then extend Theorem~\ref{informal-theorem:bounded-complexity-online-learning} as follows.

\begin{theorem}[Online learning multi-time processes of bounded complexity---informal]\label{informal-theorem:online_learning_multi_time_bounded_complexity}
    The class of $n$-qubit multi-time processes with complexity parameter $G$ can be online learned with regret bound $\mathcal{O}\left(\sqrt{T G\log(Gn)}\right)$ and with $\varepsilon$-mistake bound $\mathcal{O}\!\left(\frac{G\log(Gn)}{\varepsilon^2}\right)$.
\end{theorem}

We also extend Theorem~\ref{informal-theorem:convex-mixture-channel-online-learning} to an online learning result for convex mixtures of arbitrary known multi-time processes.

\begin{theorem}[Online learning of convex mixtures of known multi-time processes---informal]\label{informal-theorem:convex-mixture-multi_time-online-learning}
    Given an arbitrary set of $K>0$ known and fixed quantum multi-time processes, any convex mixture of these processes can be online learned with regret bound $\mathcal{O}\left(\sqrt{T \log K}\right)$ and $\varepsilon$-mistake bound $\mathcal{O}(\frac{\log(K)}{\varepsilon^2})$.
\end{theorem}

Finally, we extend the shadow tomography result of \Cref{informal-corollary:pauli-channel-shadow-tomography} to arbitrary multi-time processes. Here, the shadow tomography problem for multi-time processes is defined analogously as for channels; we refer to Problem~\ref{problem:shadow_tomography_multi_time_processes} for the formal problem statement.

\begin{theorem}[Shadow tomography of multi-time processes---informal]
    Shadow tomography of an arbitrary $n$-qubit multi-time process with $r\in\{1,2,\dotsc\}$ time steps can be solved to accuracy $\varepsilon > \frac{1}{2}\norm{N-N^{\mathsf{P}}}_{\diamond r}$ using
    \begin{equation}
        k = \mathcal{O}\left(\frac{\sqrt{nr}\log(M) \log^{3/2}(((\varepsilon - \frac{1}{2}\norm{N-N^{\mathsf{P}}}_{\diamond r}) \delta)^{-1})}{(\varepsilon - \frac{1}{2}\norm{N-N^{\mathsf{P}}}_{\diamond r})^3}\right)
    \end{equation}
    copies of the process, where $N$ is the Choi representation of the process, $N^{\mathsf{P}}$ is the Choi representation of the Pauli-twirled version of the process, and $\norm{\cdot}_{\diamond r}$ is the strategy $r$-norm.
\end{theorem}

\subsection{Related work}

\paragraph{Online learning of quantum states.} Ref.~\cite{ACH+19} introduced the problem of online learning quantum states and proposed three conceptually different approaches that all achieve mistake bounds scaling linearly in the system size. Recently, Ref.~\cite{chen2022adaptive} investigated an adaptive variant of this online learning problem, in which the underlying state may change over time. A notable application of online state learning is its use as a subroutine in recent shadow tomography protocols~\cite{Aar17, badescu2021improved, abbas2023quantum}.

A natural question to ask is whether we can simply apply known results on online learning quantum states to the Choi state of the unknown channel to accomplish the channel online learning task laid out above. This is certainly possible; however, we aim to predict quantities of the form $\Tr[M_B\mathcal{N}_{A\to B}(\rho_A)]$. And while Choi states are efficiently (online) learnable, the so-called ``Choi-to-channel'' translation incurs a dimension factor:
\begin{equation}\label{eq-choi-channel-translation}
    \mathcal{N}_{A\to B}(X_A)=\Tr_A[(X_A^{\t}\otimes\mathbbm{1}_B)C_{A,B}^{\mathcal{N}}]=d_A\Tr_A[(X_A^{\t}\otimes\mathbbm{1}_B)\Phi_{A,B}^{\mathcal{N}}],
\end{equation}
where $C_{A,B}^{\mathcal{N}}$ is the Choi \emph{matrix} of $\mathcal{N}$, $\Phi_{A,B}^{\mathcal{N}}=\frac{1}{d_A}C_{A,B}^{\mathcal{N}}$ is the Choi \emph{state} of $\mathcal{N}$, and $A$ and $B$ are the input and output systems, respectively, of $\mathcal{N}$; see Section~\ref{sec-basics_quantum_info} for formal definitions. 
Because of the dimension factor in the right-most equality above, we have
\begin{equation}\label{eq-channel_expectation_main}
    \Tr[M_B\mathcal{N}_{A\to B}(\rho_A)]=\Tr[(\rho_A^{\t}\otimes M_B)C_{A,B}^{\mathcal{N}}]=d_A\Tr[(\rho_A^{\t}\otimes M_B)\Phi_{A,B}^{\mathcal{N}}].
\end{equation}
Consequently, good regret and mistake bounds for online learning the Choi state do not give rise to good regret and mistake bounds for online learning the channel. In \Cref{sec-online_learning_Choi_state}, we discuss this and further issues arising when merely applying quantum state learning algorithms from Ref.~\cite{ACH+19} (such as the \emph{matrix multiplicative weights} (MMW) algorithm) to the Choi state for channel online learning. In particular, noting that the MMW algorithm for online learning quantum states cannot be directly applied to the learning of Choi states 
(due to the fact that Choi states have an additional partial trace requirement), we provide a projected MMW algorithm that can be used to learn the Choi state, along with the associated regret bound analysis.

\paragraph{Learning channels of bounded complexity.}  With general quantum channels being impossible to learn in many scenarios, recent work has investigated the learnability of channels with bounded gate complexity in different settings.
In variational quantum machine learning, a variety of works derived bounds on learning-theoretic complexity measures, and hence the sample complexity sufficient for good PAC generalization bounds, in terms of the number of gates~\cite{caro2020pseudo, popescu2021learning, chen2021expressibility, du2022efficient, caro2022generalization, cai2022sample, caro2023out-of-distribution}.
For learning classical-to-quantum mappings, Refs.~\cite{chung2021sample, caro2021binary, fanizza2022learning} gave similar-in-spirit sample complexity bounds derived from gate complexity assumptions. Finally, Refs.~\cite{HBC+21, zhao2023learning} considered different scenarios of state and process learning under assumptions of limited gate complexity.

\paragraph{Learning Pauli channels.}  Pauli channel learning has been considered in different (mostly non-online) scenarios.
Ref.~\cite{flammia2020efficient} gave procedures for approximating the Pauli error rates of a general unknown Pauli channel in different $\ell_p$ norms, with a recent improvement for the $\ell_\infty$ norm in Ref.~\cite{flammia2021pauli}.
Ref.~\cite{fawzi2023lower} proved that the query complexities achieved by these procedures are optimal among non-adaptive incoherent strategies, and also gave lower bounds for adaptive incoherent strategies.
If the Pauli noise is known to have a local structure, the query complexity can be improved beyond the results of Refs.~\cite{flammia2020efficient, harper2020efficient}, even if the conditional independence structure is not known in advance~\cite{rouze2023efficient}.
Refs.~\cite{chen2022quantum, chen2023tight, chen2023futility} highlight the importance of auxiliary systems and entanglement in learning the eigenvalues of an unknown Pauli channel. In Ref.~\cite{chen2023efficientPauli}, the authors provide an online algorithm for learning the eigenvalues of a Pauli channel\footnote{We note that Ref.~\cite{chen2023efficientPauli} is an updated version of Ref.~\cite{chen2023futility}.}, while in this work we consider the task of learning the error rates of Pauli channels.
Finally, Ref.\ \cite{caro2022learning} investigates the more general task of learning the Pauli transfer matrix of a general channel, giving both non-adaptive and adaptive procedures.

\subsection{Techniques and proof overview}

\paragraph*{Sequential covering numbers.} In our proof of \Cref{informal-theorem:bounded-complexity-online-learning}, we combine tools from two recent lines of work in classical online learning and quantum machine learning.
On the one hand, we rely on regret bounds for online learning in terms of sequential complexity measures for the underlying hypothesis class. In particular, we use regret bounds via sequential covering numbers \cite{rakhlin2015online, rakhlin2015sequential}, which can be viewed as a sequential version of Dudley's theorem. Namely, as we recall in \Cref{theorem:regret-bound-sequential-covering},
\begin{equation}
    \sum_{t=1}^T \ell_t(f_t(x_t)) - \min_{f\in\mathcal{F}}\sum_{t=1}^T \ell_t(f(x_t))
    \leq 2LT \inf_{\alpha>0} \left\{4\alpha + \frac{12}{\sqrt{T}}\int_{\alpha}^{1} \sqrt{\log N_{T}(\mathcal{F},\beta,2)}~\mathrm{d}\beta\right\}\, ,
\end{equation}
where $N_{T}(\mathcal{F},\beta,2)$ denotes the worst-case sequential $2$-norm $\beta$-covering number of $\mathcal{F}\subseteq [0,1]^{\mathsf{X}}$ over all complete binary trees $\mathbf{x}$ of depth $T$ whose nodes are labeled by elements of $\mathsf{X}$.
On the other hand, for the class of channels with a given gate complexity $G$, we invoke covering number bounds with respect to the diamond norm distance \cite{caro2022generalization, HBC+21, zhao2023learning}, and we demonstrate that these also control sequential covering numbers.
Concretely, in \Cref{corollary:sequential-covering-number-gate-complexity}, we show that 
\begin{equation}
    N_{T}(\mathsf{CPTP}_{n, G},\varepsilon,p)
    \leq \binom{n}{2}^G\left(\frac{6G}{\varepsilon}\right)^{512 G} \, , 
\end{equation}
where $\mathsf{CPTP}_{n, G}$ denotes the class of $n$-qubit channels with gate complexity $G$.
Together, these results imply the regret bound from \Cref{informal-theorem:bounded-complexity-online-learning}. 
This in turn leads to the claimed mistake bound via a standard argument (see \Cref{lemma:regret-to-mistake-template}).

\paragraph*{Online convex optimization.} Theorem \ref{informal-theorem:pauli-channel-online-learning} is based on efficient mistake-bounded learning of a convex mixture of exponentially many \textit{known} channels. A simple but crucial observation is that the only unknown about such channels is a classical probability distribution $\Vec{p}$ (for example, in the special case of Pauli channels, the Pauli error rate distribution). Consequently, online learning of these channels corresponds to online learning of $\Vec{p}$ when the input states and the measurements (on the output state evolved by the unknown channel) are adversarially revealed to the learner. We show that this task can be achieved via an alternative online learning scenario, in which any
adversarially chosen state and measurement can be encoded in a `channel observable' $E_{A,B}^{(t)}$, which when revealed to the learner is associated to a ``challenge'' vector $\Vec{m}^{(t)}$ given by
\begin{equation}\label{eq-challenge-vectors-propor}
    \Vec{m}^{(t)}=(m_{\Vec{z},\Vec{x}}^{(t)})_{\Vec{z},\Vec{x}\in\{0,1\}^n} \text{ with } m_{\Vec{z},\Vec{x}}^{(t)}
    \propto \Tr[E_{A,B}^{(t)}\Gamma_{A,B}^{\Vec{z},\Vec{x}}],
\end{equation}
where $E_{A,B}^{(t)}=(\rho_A^{(t)})^{\t}\otimes M_B^{(t)}$, and $\Gamma_{A,B}^{\Vec{z},\Vec{x}}$ is the Choi representation of the Pauli unitary channel $\rho\mapsto P^{\Vec{z},\Vec{x}}\rho P^{\Vec{z},\Vec{x}\dagger}$. The learner produces hypotheses $\Vec{p}^{(t)}$ of the unknown probability distribution at times $t\in\{1,2,\dotsc,T\}$. The learner's hypotheses should be such that, for $T\in\{1,2,\dotsc\}$, $\sum_{t=1}^T\Vec{m}^{(t)}\cdot \Vec{p}^{(t)}$ is not too different from $\sum_{t=1}^T\Vec{m}^{(t)}\cdot \Vec{p}$. 
Using known guarantees~\cite{arora2012multiplicative}, the difference between these two sums can be shown to scale only logarithmically with the size of $\Vec{p}$'s support for the multiplicative weights update method. This implies regret bounds, and therefore mistake bounds, that scale only linearly in the number of qubits.

\paragraph*{Mistake lower bounds.} We prove our lower bounds by embedding classical online learning problems into the quantum tasks of interest, thereby inheriting classical lower bounds.
First, it is easy to embed the task of online learning a general function $f:\{0,1\}^{n-1}\to\{0,1\}$ into that of online learning a general $n$-qubit unitary. (Alternatively, when allowing for non-unitary channels, we give such an embedding into $(n-1)$-qubit channels.) Thus, the folklore mistake lower bound of $\Omega(2^n)$ for the former problem immediately carries over to the latter, showing that the class of all $n$-qubit unitaries cannot be online learned with a sub-exponential number of mistakes.

Second, to prove mistake lower bounds for learning channels of gate complexity $G$, we use the fact that $\mathcal{O}(2^k)$ many gates suffice to implement an arbitrary function $f:\{0,1\}^{k}\to\{0,1\}$ with a classical circuit. 
We can quantumly realize such a circuit either by first measuring in the computational basis and then applying the classical circuit, or by implementing irreversible AND and OR gates with reversible Toffoli gates, where the auxiliary system gets reset to a suitable value after every gate. 
Both constructions demonstrate that the class of $n$-qubit channels of gate complexity $G$ contains the class of all Boolean functions on the first $q=\min\{\log_2(\Theta(G)), n-1\}$ inputs (extended to the remaining subsystems by a trivial action).
So, the classical folklore lower bound gives an $\Omega(2^q)=\Omega(\min\{2^n, G\})$ mistake bound here.

Third, by associating a general function $f:\{1,\ldots,n\}\to\{0,1\}$ with the $n$-qubit Pauli channel $\mathcal{N}_f$ given by $\mathcal{N}_f (\rho) = \left(\bigotimes_{i=1}^n Z_i^{f(i)}\right)\rho \left(\bigotimes_{i=1}^n Z_i^{f(i)}\right)$, we prove that online learning $n$-qubit Pauli channels is at least as hard as online learning a general $\{0,1\}$-valued function on $\lfloor\log(n)\rfloor$ bits. Hence, the classical mistake lower bound for online learning arbitrary functions becomes a $\Omega(n)$ mistake lower bound for Pauli channel online learning.

\paragraph*{Computational complexity lower bounds.} Our exponential computational complexity lower bound for online learning Pauli channels with a polynomial mistake bound is a simple formalization of the following intuition: The channel observables posed as challenges by the adversary are exponentially-sized objects, and a successful online learner has to process all the exponentially many entries. We give a simple adversary demonstrating that this intuition is correct, and applies even if the adversary provides the channel observables already in the basis expansion that is most natural for Pauli channels, namely the (unnormalized) $n$-qubit Bell basis.

To prove our computational complexity lower bound for online learning bounded-complexity channels, we require a hardness assumption. Here, we consider the so-called \emph{ring learning with errors} (LWE) problem~\cite{lyubashevsky2010ideal} (\textsf{RingLWE}), which underlies much of lattice-based cryptography. To establish that hardness of \textsf{RingLWE} implies hardness of online learning slightly superlinear-sized quantum circuits, we first show, via a known construction, that this class of quantum circuits can implement a pseudorandom function class. This implies hardness of online learning, because pseudorandom function classes are hard to learn, an intuition that we formalize specifically for mistake-bounded online learning. Importantly, as the hardness arises from an underlying classical pseudorandom function class, it persists even when all challenges consists of input states and projective effect operators that are computational basis elements. In particular, online learning remains hard here, even though the learner can efficiently read all challenges. 
Additionally, as our pseudorandom function class is secure against adversaries that have query access to the function, the computational hardness of online learning holds even if the learner can actively choose the (pairwise distinct) challenges, rather than those being chosen by the adversary.

\paragraph*{Classical adaptive data analysis.} In the problem of shadow tomography of quantum channels, the learner's task is to output values $b_t\in\mathbb{R}$ such that $\abs{b_t - \Tr[M_B^{(t)}\mathcal{N}_{A\to B}(\rho_A^{(t)})]} \leq \varepsilon$ for all $t\in\{1,2,\dotsc,T\}$, where $\{(\rho_A^{(t)}, M_B^{(t)})\}_{t=1}^T$ is a set of state-measurement pairs. They should do so while minimizing the number $k$ of times that they access the channel. The shadow tomography problem for multi-time processes is formulated in a similar manner.

Our shadow tomography results are based on Bell sampling of the Choi state of the channel, i.e., sending one-half of the maximally-entangled state through the channel and then performing a joint Bell-basis measurement on the output system and the second entangled copy of the input system. If the unknown channel is a Pauli channel, then this procedure directly gives us samples from the Pauli error rate distribution. Consequently, the desired expectation values $\Tr[M_B^{(t)}\mathcal{N}_{A\to B}(\rho_A^{(t)})]$ can be interpreted as (classical) statistical queries with respect to the (unknown) Pauli error rate vector. Specifically, we have that $\Tr[M_B^{(t)}\mathcal{N}_{A\to B}(\rho_A^{(t)})]=\Vec{p}\cdot\Vec{e}^{(t)}$, where $\Vec{p}$ is the Pauli error rate vector and $\Vec{e}^{(t)}=(e_{\Vec{z},\Vec{x}}^{(t)})_{\Vec{z},\Vec{x}\in\{0,1\}^n}$, $e_{\Vec{z},\Vec{x}}=\Tr[((\rho_A^{(t)})^{\t}\otimes M_B^{(t)})\Gamma_{A,B}^{\Vec{z},\Vec{x}}]$. With this observation, we can make use of the classical adaptive data analysis algorithms presented in Ref.~\cite{bassily2016adaptivedataanalysis} in order to obtain our bound in \Cref{informal-corollary:pauli-channel-shadow-tomography} on the number $k$ of accesses to the channel. For an arbitrary unknown channel, we show that the same Bell sampling strategy enables us to sample from the error-rate vector of the \textit{Pauli-twirled} version of the channel, which we denote by $\mathcal{N}_{A\to B}^{\mathsf{P}}$. Then, as long as $\varepsilon>\frac{1}{2}\norm{\mathcal{N}-\mathcal{N}^{\mathsf{P}}}_{\diamond}$, we obtain a similar guarantee as in \Cref{informal-corollary:pauli-channel-shadow-tomography}.

\paragraph*{Quantum combs.} Our results on multi-time quantum processes make use of theory of quantum combs~\cite{CDP09}, also known as ``quantum strategies''~\cite{GW07}. We provide formal definitions of quantum combs in Appendix~\ref{sec-multi_time_processes}. In particular, the analysis of multi-time processes entails use of the so-called ``strategy norm'' and its H\"{o}lder dual~\cite{Gut12}, which are the multi-time generalizations of the trace and spectral norm, respectively, used in the analysis of algorithms for quantum states. They also generalize the diamond norm and its H\"{o}lder dual in the case of quantum channels. A crucial ingredient in the proof of Theorem~\ref{informal-theorem:online_learning_multi_time_bounded_complexity} is submultiplicativity of the strategy norm under composition of multi-time processes. As the composition of multi-time processes is given by the link product~\cite{CDP09}, the technique for proving submultiplicativity of the diamond norm does not generalize straightforwardly to the strategy norm. To the best of our knowledge, a proof of submultiplicativity of the strategy norm under link product has not been provided before, and we provide such a proof in Appendix~\ref{sec-strategy_norms} based on semi-definite programming duality.

\subsection{Directions for future work}

Motivated by the well known impossibility of online learning general unitaries and channels with a subexponential number of mistakes, our work initiates a study of subclasses of channels that allow for good regret and mistake bounds in online learning. We have identified two such classes: Channels of bounded gate complexity and mixtures of arbitrary known channels, with Pauli channels as a notable special case.
However, we show that achieving favorably scaling regret and mistake bounds for these classes is not possible in a computationally efficient manner.
Our results open up several directions for future research; here we outline some of them.

First, while our computational complexity lower bounds put limitations on where we can hope for computationally efficient channel online learning, they still leave relevant regions to explore. 
On the one hand, by the same reasoning as in \Cref{sec:computational-lower-bounds-bounded-gate-comlexity}, the channel class under consideration must not be expressive enough to implement known pseudorandom function constructions. Here, inspired by the recent work~\cite{huang2024learningshallow}, one may investigate the online learnability of shallow quantum circuits.
On the other hand, as we argue in detail in \Cref{sec:computational-lower-bounds-pauli-channels}, a necessary condition for efficient online learning is that the channels of interest admit efficient descriptions.
Candidates for such channel classes may be Clifford circuits or channels represented by matrix product operators of low bond-dimension.

Second, we believe that there is room for ``onlinification'' of other quantum learning scenarios. For instance, recent work~\cite{huang2022learning, caro2022learning} has proposed to avoid the exponential bottleneck of general process tomography (compare, e.g., Refs.~\cite{haah2023query, oufkir2023sample}) by considering learning tasks with arbitrary channels but restricted or structured input states and output measurements. Similarly, one may attempt to circumvent exponential lower bounds in online learning arbitrarily complex channels by imposing restrictions on the behavior of the adversary, for instance, with respect to the challenges that they can pose.
Another line of quantum learning research that has recently seen significant progress is learning a Hamiltonian from access to the associated dynamics~\cite{Haah2021optimal, franca2022efficient, Gu2022Practical, wilde2022scalably, caro2022learning, Yu2023robustefficient,  huang2023learning, li2023heisenberglimited, mobus2023dissipation, castaneda2023hamiltonian, bakshi2023learning, haah2024learning, bluhm2024hamiltonian, bakshi2024structure}. Online learning variants of the standard Hamiltonian learning task could give insights into whether one can learn to fine-tune a Hamiltonian evolution according to adaptive feedback.

Finally, while general shadow tomography for quantum channels is not possible, we have demonstrated that it can become feasible for a suitable subclass, in our case Pauli channels. Given the important role of online learning quantum states in state tomography procedures, we envision positive results, both ours and future ones, on online learning restricted classes of channels to serve as a stepping stone towards shadow tomography procedures for such classes.
Achieving the latter would likely require analogues of threshold search~\cite{badescu2021improved} for these kinds of channels.

\subsection*{Acknowledgments} 
The authors thank Akshay Bansal, Ian George, Soumik Ghosh, Jamie Sikora, and Alice Zheng for sharing a draft of their independent and concurrent work on online learning quantum objects. The authors gratefully acknowledge support from the BMBF (QPIC-1, HYBRID), the ERC (DebuQC), the Munich Quantum Valley,
the Einstein Foundation, and Berlin Quantum. MCC was partially supported by a DAAD PRIME fellowship.

\newpage

\tableofcontents

\newpage

\section{Preliminaries}

\subsection{Basics of quantum information}\label{sec-basics_quantum_info}

Here we provide a brief review of fundamental quantum information concepts that we make use of throughout this work. We refer to, e.g., Ref.~\cite{Wat18_book} for further details on the concepts and definitions presented in this subsection.

\paragraph*{Pauli operators.} For a quantum system of $n\in\mathbb{N}$ qubits, we define the \textit{$n$-qubit Pauli operators} as
\begin{align}
    P^{\Vec{z},\Vec{x}}&\coloneqq (+\I)^{\Vec{z}\cdot\Vec{x}}Z^{\Vec{z}}X^{\Vec{x}},\quad \Vec{x},\Vec{z}\in\{0,1\}^n, \label{eq-n_qubit_Pauli_ops}\\ 
    Z^{\Vec{z}}&\coloneqq Z^{z_1}\otimes Z^{z_2}\otimes\dotsb\otimes Z^{z_n},\quad Z^{z}=\ketbra{0}{0}+(-1)^z\ketbra{1}{1},\\
    X^{\Vec{x}}&\coloneqq X^{x_1}\otimes X^{x_2}\otimes\dotsb\otimes X^{x_n},\quad X^{x}=\ketbra{x}{0}+\ketbra{x\oplus 1}{1},
\end{align}
where the operation ``$\oplus$'' denotes addition modulo two. From this, we can define the \textit{$n$-qubit Bell states} as 
\begin{equation}\label{eq-n_qubit_Bell_states}
    \Phi^{\Vec{z},\Vec{x}}=\ketbra{\Phi^{\Vec{z},\Vec{x}}}{\Phi^{\Vec{z},\Vec{x}}},\quad \ket{\Phi^{\Vec{z},\Vec{x}}}\coloneqq (\mathbbm{1}\otimes P^{\Vec{z},\Vec{x}})\ket{\Phi},\quad \ket{\Phi}=\frac{1}{\sqrt{2^n}}\sum_{\Vec{x}\in\{0,1\}^n}\ket{\Vec{x},\Vec{x}},
\end{equation}
for all $\Vec{z},\Vec{x}\in\{0,1\}^n$. The Bell state vectors $\ket{\Phi^{\Vec{z},\Vec{x}}}$ form an orthonormal basis for $(\mathbb{C}^{2})^{\otimes n}\otimes(\mathbb{C}^2)^{\otimes n}$, and the set $\{\Phi^{\Vec{z},\Vec{x}}\}_{\Vec{z},\Vec{x}\in\{0,1\}^n}$ forms a \emph{positive operator-valued measure} (POVM).

\paragraph*{Quantum channels.} A quantum channel is a completely positive trace-preserving (CPTP) linear map $\mathcal{N}:\Lin(\mathbb{C}^d)\to\Lin(\mathbb{C}^{d'})$. We often write $\mathcal{N}_{A\to B}$ to refer to a quantum channel with input system $A$ and output system $B$, with corresponding Hilbert spaces $\mathcal{H}_A\cong\mathbb{C}^{d_A}$ and $\mathcal{H}_B\cong\mathbb{C}^{d_B}$, respectively. We let $\mathsf{CPTP}(A;B)$ denote the set of all quantum channels mapping system $A$ to system $B$. In this work, we mostly consider the case $d_A=d_B=d=2^n$, corresponding to a system of $n$ qubits, and we use the notation $\mathsf{CPTP}_n$ to refer to the set of all quantum channels mapping $n$ qubits to $n$ qubits.

The \textit{Choi representation} (or \textit{Choi matrix}) of a linear map $\mathcal{N}:\Lin(\mathbb{C}^d)\to\Lin(\mathbb{C}^{d'})$ is defined as
\begin{equation}\label{eq-choi_rep}
    C(\mathcal{N})\coloneqq(\id_{d}\otimes\mathcal{N})(\ketbra{\Gamma_d}{\Gamma_d})=\sum_{i,j=0}^{d-1}\ketbra{i}{j}\otimes\mathcal{N}(\ketbra{i}{j}),
\end{equation}
where $\id_d:\Lin(\mathbb{C}^d)\to\Lin(\mathbb{C}^d)$ is the identity superoperator, and
\begin{equation}
    \ket{\Gamma_d}\coloneqq\sum_{i=0}^{d-1}\ket{i,i}.
\end{equation}
The \textit{Choi state} of $\mathcal{N}$ is the normalized Choi matrix, defined as
\begin{equation}
    \Phi(\mathcal{N})\coloneqq\frac{1}{d}(\id_d\otimes\mathcal{N})(\ketbra{\Gamma_d}{\Gamma_d}) =  \frac{C(\mathcal{N})}{d}.
\end{equation}
We sometimes write $C_{A,B}^{\mathcal{N}}\equiv C(\mathcal{N})$ and $\Phi_{A,B}^{\mathcal{N}}\equiv\Phi(\mathcal{N})$ for the Choi matrix and Choi state, respectively, of a quantum channel $\mathcal{N}_{A\to B}$ when we want to indicate explicitly the input and output systems $A$ and $B$ of the channel. For a quantum channel $\mathcal{N}_{A\to B}$, its Choi representation $C_{A,B}^{\mathcal{N}}$ is positive semi-definite and satisfies $\Tr_B[C_{A,B}^{\mathcal{N}}]=\mathbbm{1}_A$.

The function $C$ defined in \eqref{eq-choi_rep} has an inverse, such that we can identify every Hermitian operator $H_{A,B}\in\Lin(\mathcal{H}_A\otimes\mathcal{H}_B)$ with a Hermiticity-preserving map $C^{-1}(H_{A,B}):\Lin(\mathcal{H}_A)\to\Lin(\mathcal{H}_B)$ as
\begin{equation}\label{eq-Choi_iso_inv}
    \mathcal{N}_{A\to B}^{H}\coloneqq C^{-1}(H_{A,B}),\quad \mathcal{N}_{A\to B}^H(X_A)=\Tr_A[(X_A^{\t}\otimes\mathbbm{1}_B)H_{A,B}].
\end{equation}
In particular, if $H_{A,B}$ is positive semi-definite, then $\mathcal{N}^H$ is completely positive. If in addition $\Tr_B[H_{A,B}]=\mathbbm{1}_A$, then $\mathcal{N}^H$ is a quantum channel. Consequently, the set $\mathsf{CPTP}(A;B)$ of quantum channels is in one-to-one correspondence with the set $\mathsf{CPTP}^{\prime}(A;B)\coloneqq\{N_{A,B}\in\Lin(\mathcal{H}_A\otimes\mathcal{H}_B):N_{A,B}\geq 0,\,\Tr_B[N_{A,B}]=\mathbbm{1}_A\}$. We let
\begin{equation}\label{eq-CPTP_choi_matrices}
    \mathsf{CPTP}_n^{\prime}\coloneqq\Big\{N_{A,B}\in\Lin(\mathcal{H}_A\otimes\mathcal{H}_B):\mathcal{H}_A\cong\mathcal{H}_B\cong(\mathbb{C}^2)^{\otimes n},\,N_{A,B}\geq 0,\,\Tr_B[N_{A,B}]=\mathbbm{1}_A\Big\}
\end{equation}
denote the set of Choi matrices of quantum channels mapping $n$ qubits to $n$ qubits.

Quantum channels also have a Kraus representation, such that
\begin{equation}
    \mathcal{N}(\rho)=\sum_{\ell=1}^r K_{\ell}\rho K_{\ell}^{\dagger},\quad\sum_{\ell=1}^r K_{\ell}^{\dagger}K_{\ell}=\mathbbm{1}_A,
\end{equation}
where $r\in\mathbb{N}$ and $K_{\ell}:\mathcal{H}_A\to\mathcal{H}_B$ is a linear operator for every $\ell\in\{1,2,\dotsc,r\}$.

\paragraph*{Pauli channels.} A Pauli channel is a quantum channel whose Kraus operators are proportional to the Pauli operators defined in \Cref{eq-n_qubit_Pauli_ops}. Specifically, an $n$-qubit Pauli channel (by definition) has the form
\begin{equation}\label{eq-pauli_channel}
    \mathcal{P}(\rho)=\sum_{\Vec{z},\Vec{x}\in\{0,1\}^n} p_{\Vec{z},\Vec{x}} P^{\Vec{z},\Vec{x}}\rho P^{\Vec{z},\Vec{x}\dagger},
\end{equation}
where the \textit{Pauli error rates} $p_{\Vec{z},\Vec{x}}$ form a probability distribution, i.e., $p_{\Vec{z},\Vec{x}}\in[0,1]$ for all $\Vec{z},\Vec{x}\in\{0,1\}^n$ and $\sum_{\Vec{z},\Vec{x}\in\{0,1\}^n}p_{\Vec{z},\Vec{x}}=1$. The Choi representation of a Pauli channel $\mathcal{P}$ is
\begin{equation}\label{eq-Pauli_channel_Choi_rep}
    C(\mathcal{P})=\sum_{\Vec{z},\Vec{x}\in\{0,1\}^n}p_{\Vec{z},\Vec{x}} \Gamma^{\Vec{z},\Vec{x}},
\end{equation}
where
\begin{equation}
    \Gamma^{\Vec{z},\Vec{x}}\coloneqq\ketbra{\Gamma^{\Vec{z},\Vec{x}}}{\Gamma^{\Vec{z},\Vec{x}}},\quad \ket{\Gamma^{\Vec{z},\Vec{x}}}=(\mathbbm{1}\otimes P^{\Vec{z},\Vec{x}})\ket{\Gamma}=\sqrt{2^n}\ket{\Phi^{\Vec{z},\Vec{x}}},
\end{equation}
are the unnormalized versions of the $n$-qubit Bell states defined in \Cref{eq-n_qubit_Bell_states}. We let $\Vec{p}=(p_{\Vec{z},\Vec{x}}:\Vec{z},\Vec{x}\in\{0,1\}^n)$ denote the $4^n$-dimensional probability vector of error rates.

We let $\mathsf{PAULI}_n$ be the set of all Pauli channels acting on $n$-qubit systems, and analogously to \eqref{eq-CPTP_choi_matrices}, we let
\begin{equation}
    \mathsf{PAULI}_n^{\prime}\coloneqq\left\{N_{A,B}:N_{A,B}=\sum_{\Vec{z},\Vec{x}\in\{0,1\}^n}p_{\Vec{z},\Vec{x}}\Gamma_{A,B}^{\Vec{z},\Vec{x}},\, \Vec{p}=(p_{\Vec{z},\Vec{x}})_{\Vec{z},\Vec{x}}\in\Delta_{4^n}\right\}
\end{equation}
be the set of all Choi matrices of Pauli channels, where
\begin{equation}
    \Delta_m\coloneqq\left\{\Vec{p}=(p_1,p_2,\dotsc,p_m)\in\mathbb{R}^m:p_k\in[0,1]~\forall k\in\{1,2,\dotsc,m\},\,\sum_{k=1}^m p_k=1\right\}
\end{equation}
denotes the probability simplex of $m\in\mathbb{N}$ elements. Depending on the context, we refer to a vectors $\Vec{p}=(p_1,p_2,\dotsc,p_m)$ with values $p_1,p_2,\dotsc,p_m\in[0,1]$ and $\sum_{k=1}^m p_k=1$ as both a \textit{probability vector} and a \textit{probability distribution}.

Associated to every quantum channel $\mathcal{N}$ is its Pauli-twirled version, defined as~\cite{DHCB05}
\begin{equation}
    \mathcal{N}^{\mathsf{P}}(\rho)\coloneqq\frac{1}{4^n}\sum_{\Vec{z},\Vec{x}\in\{0,1\}^n}P^{\Vec{z},\Vec{x}\dagger}\mathcal{N}(P^{\Vec{z},\Vec{x}}\rho P^{\Vec{z},\Vec{x}\dagger})P^{\Vec{z},\Vec{x}},
\end{equation}
where the superscript ``$\mathsf{P}$'' in $\mathcal{N}^{\mathsf{P}}$ refers to the set $\mathsf{P}\coloneqq\{P^{\Vec{z},\Vec{x}}:\Vec{z},\Vec{x}\in\{0,1\}^n\}$ of $n$-qubit Pauli operators. In other words, the Pauli-twirled version of a channel is given by applying a Pauli operator, chosen uniformly at random, and its inverse at the input and output of the channel. As we recall in Appendix~\ref{sec-Pauli_twirl_quantum_channel}, the Pauli-twirled channel $\mathcal{N}^{\mathsf{P}}$ is indeed a Pauli channel, and the error rates are given by $p_{\Vec{z},\Vec{x}}=\frac{1}{d}\Tr[\Phi^{\Vec{z},\Vec{x}}C(\mathcal{N})]$ for all $\Vec{z},\Vec{x}\in\{0,1\}^n$. In particular, we show that the Choi representation of the Pauli-twirled channel can be obtained via the \textit{pinching channel} $\mathcal{S}_{\mathsf{P}}$ in the Bell basis, defined as
\begin{equation}
    \mathcal{S}_{\mathsf{P}}(X)\coloneqq\sum_{\Vec{z},\Vec{x}\in\{0,1\}^n}\ketbra{\Phi^{\Vec{z},\Vec{x}}}{\Phi^{\Vec{z},\Vec{x}}}X\ketbra{\Phi^{\Vec{z},\Vec{x}}}{\Phi^{\Vec{z},\Vec{x}}}=\sum_{\Vec{x},\Vec{z}\in\{0,1\}^n}\Tr[\Phi^{\Vec{z},\Vec{x}}X]\Phi^{\Vec{z},\Vec{x}},
\end{equation}
for every linear operator $X\in\Lin((\mathbb{C}^2)^{\otimes n})$. In particular, then, the Choi representation of the Pauli-twirled version $\mathcal{N}^{\mathsf{P}}$ of a channel $\mathcal{N}$ is given by
\begin{equation}
     C(\mathcal{N}^{\mathsf{P}})
     =\sum_{\Vec{z},\Vec{x}\in\{0,1\}^n}\frac{1}{d}\Tr[\Phi^{\Vec{z},\Vec{x}}C(\mathcal{N})]\Gamma^{\Vec{z},\Vec{x}}
     =\sum_{\Vec{z},\Vec{x}\in\{0,1\}^n}\Tr[\Phi^{\Vec{z},\Vec{x}}C(\mathcal{N})]\Phi^{\Vec{z},\Vec{x}}
     =\mathcal{S}_{\mathsf{P}}(C(\mathcal{N})).\label{eq-Pauli_twirl_pinching}
\end{equation}

\paragraph*{Channel measurements and observables.} An $m$-outcome measurement of a quantum state is given by a \emph{positive operator-valued measure} (POVM), i.e., a set $\{M^{(i)}\}_{i=1}^m$ of $m$ operators satisfying $0\leq M^{(i)}\leq\mathbbm{1}$ such that $\sum_{i=1}^m M^{(i)}=\mathbbm{1}$. Observables for states are simply Hermitian operators $H$, and the expected value of the observable $H$ when measured on a state $\rho$ is $\Tr[H\rho]$.

Now, a measurement, or a \textit{test}, for a quantum channel $\mathcal{N}_{A\to B}$ is given by a pair consisting of a bipartite state $\rho_{R,A}$ and a POVM $\{M_{R,B}^{(i)}\}_{i=1}^m$, where $R$ is an arbitrary memory/reference system; see Figure~\ref{fig-channel_test_gen}(a). Using \Cref{eq-choi-channel-translation}, the probability of obtaining a particular outcome $i\in\{1,2,\dotsc,m\}$ of the test is given by
\begin{equation}\label{eq-choi-matrix-to-channel-prob}
    \Tr[M_{R,B}^{(i)}\mathcal{N}_{A\to B}(\rho_{R,A})]=\Tr[E_{A,B}^{(i)}C_{A,B}^{\mathcal{N}}],
\end{equation}
where $E_{A,B}^{(i)}=\Tr_R[(\mathbbm{1}_B\otimes\rho_{R,A}^{\t_A})(\mathbbm{1}_A\otimes M_{R,B}^{(i)})]$ for all $i\in\{1,2,\dotsc,m\}$. The right-hand side of this equation is the generalized Born rule for quantum channels; see Figure~\ref{fig:channel_tests}(c) for a depiction. The channel test operators $E_{A,B}^{(i)}$ satisfy $E_{A,B}^{(i)}\geq 0$ for all $i\in\{1,2,\dotsc,m\}$ and $\sum_{i=1}^m E_{A,B}^{(i)}=\rho_A^{\t}\otimes\mathbbm{1}_B$, where $\rho_A\equiv\Tr_R[\rho_{R,A}]$. The converse is also true~\cite{Ziman08,BGS20}, meaning that every channel test can be characterized by a set $\{E_{A,B}^{(i)}\}_{i=1}^m$ such that $0\leq E_{A,B}^{(i)}\leq\sigma_A\otimes\mathbbm{1}_B$ for all $i\in\{1,2,\dotsc,m\}$, for some density operator $\sigma_A$, and $\sum_{i=1}^m E_{A,B}^{(i)}=\sigma_A\otimes\mathbbm{1}_B$. The corresponding ``physical realization'' of the channel test is given by $\rho_{R,A}=\ketbra{\psi^{\sigma}}{\psi^{\sigma}}_{R,A}$, where $\mathcal{H}_R\cong\mathcal{H}_A$, $\ket{\psi^{\sigma}}$ is a purification of $\sigma$, and 
\begin{equation}    
    M_{R,B}^{(i)}=\sigma_R^{-\frac{1}{2}}E_{R,B}^{(i)}\sigma_R^{-\frac{1}{2}}.
\end{equation}
An especially simple example of a channel test is one without memory, involving an input state $\rho_A$ at the input of the channel and a measurement $\{M_B^{(i)}\}_{i=1}^m$ at the output of a channel; see Figure~\ref{fig:channel_tests}(b). In this case, the channel test operators are in tensor-product form, given by $E_{A,B}^{(i)}=\rho_A^{\t}\otimes M_B^{(i)}$ for all $i\in\{1,2,\dotsc,m\}$.
    
The statements above for measurements readily generalize to statements about observables (Hermitian operators), due to the fact that every Hermitian operator has a spectral decomposition, and the spectral projections form a POVM. Consequently, the expected value of an observable $H_{R,B}$, measured according to the general scenario depicted in Figure~\ref{fig:channel_tests}(c), is equal to
\begin{equation}
    \Tr[H_{R,B}\mathcal{N}_{A\to B}(\rho_{R,A})]=\Tr[O_{A,B}C_{A,B}^{\mathcal{N}}],
\end{equation}
where the ``channel observable'' is $O_{A,B}=\Tr_R[(\mathbbm{1}_B\otimes\rho_{R,A}^{\t_A})(\mathbbm{1}_A\otimes H_{R,B})]$. If the measurement scheme does not contain a memory, then $O_{A,B}=\rho_A^{\t}\otimes H_B$.

Let us now consider the case that the memory system $R$ has the same dimension as the input system $A$ of the channel, so let us make the relabeling $R\equiv A'$. Let us also suppose that $\rho_{R,A}\equiv\psi_{A'A}$ is a pure state. Then, for every bipartite pure state $\psi_{A'A}$, there exists a state $\rho_A$ such that $\psi_{A'A}=\rho_A^{\frac{1}{2}}\Gamma_{A'A}\rho_A^{\frac{1}{2}}$. The overall observable is then
\begin{align}
\nonumber
    O_{A,B}&=\Tr_{A'}[(\mathbbm{1}_A\otimes H_{A'B})((\rho_A^{\frac{1}{2}}\Gamma_{A'A}\rho_A^{\frac{1}{2}})^{\t_A}\otimes\mathbbm{1}_B)]\\
    \nonumber
    &=\Tr_{A'}[(\mathbbm{1}_A\otimes H_{A'B})((\rho_{A'}^{\t})^{\frac{1}{2}}F_{A'A}(\rho_{A'}^{\t})^{\frac{1}{2}}\otimes\mathbbm{1}_B)]\\
    \nonumber
    &=\sum_{i,j=0}^{d_A-1}\Tr_{A'}\!\left[(\mathbbm{1}_A\otimes H_{A'B})((\rho_{A'}^{\t})^{\frac{1}{2}}\ketbra{j,i}{i,j}_{A'A}(\rho_{A'}^{\t})^{\frac{1}{2}}\otimes\mathbbm{1}_B)\right]\\
    \nonumber
    &=\sum_{i,j=0}^{d_A-1}\ketbra{i}{j}_A\otimes \Tr_{A'}\!\left[H_{A'B}(\rho_{A'}^{\t})^{\frac{1}{2}}\ketbra{j}{i}_{A'}(\rho_{A'}^{\t})^{\frac{1}{2}}\right]\nonumber\\
    \nonumber
    &=\sum_{i,j=0}^{d_A-1}\ketbra{i}{j}_A\otimes\bra{i}_{A'}(\rho_{A'}^{\t})^{\frac{1}{2}}H_{A'B}(\rho_{A'}^{\t})^{\frac{1}{2}}\ket{j}_{A'}\\
    &=(\rho_A^{\t})^{\frac{1}{2}}H_{A,B}(\rho_A^{\t})^{\frac{1}{2}}.\label{eq-channel_observable_with_memory}
\end{align}
Note that a special case of the observable in \Cref{eq-channel_observable_with_memory} is when $\rho_A=\frac{\mathbbm{1}_A}{d_A}$, which corresponds to measuring $H_{A,B}$ on the Choi state of the channel. In this case
\begin{equation}
    O_{A,B}=\frac{1}{d_A}H_{A,B}\quad (\rho_A=\mathbbm{1}_A/d_A),
\end{equation}
which means that
\begin{equation}\label{eq-channel_expectation_Choi_state}
    \Tr[O_{A,B}C_{A,B}^{\mathcal{N}}]=\Tr[H_{A,B}\Phi_{A,B}^{\mathcal{N}}].
\end{equation}

Now, the operator/spectral norm $\norm{\cdot}_{\infty}$ is used to characterize measurement operators and observables for quantum states, because it is the (H\"{o}lder) dual to the trace norm $\norm{\cdot}_1$. For quantum channels, the relevant norm is the diamond norm~\cite{kitaev1997}. Let $\mathcal{N}:\Lin(\mathcal{H}_A)\to\Lin(\mathcal{H}_B)$ be a Hermiticity-preserving linear map. The diamond norm of $\mathcal{N}_{A\to B}$ can be expressed as
\begin{align}
    \nonumber \norm{\mathcal{N}}_{\diamond}&=\norm{C_{A,B}^{\mathcal{N}}}_{\diamond 1}\\
    \nonumber &\coloneqq\sup_{\substack{S_{A,B},T_{A,B}\geq 0\\\sigma_A\geq 0}}\left\{\Tr[C_{A,B}^{\mathcal{N}}(S_{A,B}-T_{A,B})]:S_{A,B}+T_{A,B}=\sigma_A\otimes\mathbbm{1}_B,\,\Tr[\sigma_A]=1\right\}\\
    \nonumber 
    &=\inf_{\substack{t\geq 0\\Y_{A,B}\geq 0}}\left\{t:-Y_{A,B}\leq C_{A,B}^{\mathcal{N}}\leq Y_{A,B},\,\Tr_B[Y_{A,B}]=t\mathbbm{1}_A\right\}
    \nonumber \\
    &=\inf_{\substack{t\geq 0\\N_{A,B}\geq 0}}\left\{t:-tN_{A,B}\leq C_{A,B}^{\mathcal{N}}\leq tN_{A,B},\,\Tr_B[N_{A,B}]=\mathbbm{1}_A\right\}.\label{eq-strategy_1norm}
\end{align}
The norm $\norm{\cdot}_{\diamond 1}$ is referred to as the \textit{strategy 1-norm}, and we define it formally in Appendix~\ref{sec-multi_time_processes}. The relevant norm for channel observables is thus the H\"{o}lder dual of the strategy 1-norm, which is given by~\cite{gutoski2010thesis,Gut12}
\begin{align}
    \norm{O_{A,B}}_{\diamond 1}^{\ast}&\coloneqq\sup_{\substack{X_{A,B}\geq 0,\\Y_{A,B}\geq 0}}\Big\{\Tr[O_{A,B}(X_{A,B}-Y_{A,B})]:\Tr_B[X_{A,B}+Y_{A,B}]\leq\mathbbm{1}_A\Big\} \nonumber \\
    \nonumber&=\inf_{Y_A\geq 0}\Big\{\Tr[Y_A]:-Y_A\otimes\mathbbm{1}_B\leq O_{A,B}\leq Y_A\otimes\mathbbm{1}_B\Big\}\\
    &=\inf_{\substack{t\geq 0,\\\sigma_A\geq 0}}\Big\{t:-t\sigma_A\otimes\mathbbm{1}_B\leq O_{A,B}\leq t\sigma_A\otimes\mathbbm{1}_B,\,\Tr[\sigma_A]=1\Big\},\label{eq-diamond_norm_dual_3}
\end{align}
for every Hermitian operator $O_{A,B}$ acting on $\mathcal{H}_A\otimes\mathcal{H}_B$. We note that channel test operators satisfy $\norm{E_{A,B}}_{\diamond 1}^{\ast}\leq 1$. Also, for channel observables without memory, it holds that 
\begin{equation}
    \norm{\rho_A^{\t}\otimes H_B}_{\diamond 1}^{\ast}=\norm{\rho_A^{\t}}_1\norm{H_B}_{\infty}=\norm{H_B}_{\infty},
\end{equation}
as we might expect, and which is a property interesting in its own right; we refer to Appendix~\ref{sec-multi_time_processes} for a proof. In general, we have that $\norm{O_{A,B}}_{\diamond 1}^{\ast}\geq\norm{O_{A,B}}_{\infty}$ for all Hermitian $O_{A,B}\in\Lin(\mathcal{H}_A\otimes\mathcal{H}_B)$. This follows immediately from \eqref{eq-diamond_norm_dual_3}, on account of the fact that $\norm{\sigma_A}_{\infty}\leq\norm{\sigma_A}_1=\Tr[\sigma_A]=1\Rightarrow-\mathbbm{1}_A\leq\sigma_A\leq\mathbbm{1}_A$ for every $\sigma_A$ in the optimization in \eqref{eq-diamond_norm_dual_3}, and the fact that $\norm{H}_{\infty}=\inf\{t:-t\mathbbm{1}_d\leq H\leq t\mathbbm{1}_d\}$ for all Hermitian $H\in\Lin(\mathbb{C}^d)$.

\paragraph*{Multi-time quantum processes.} Multi-time quantum processes are those that occur over multiple time steps, as opposed simply one time step for a quantum channel. A simple example of a multi-time quantum process is the one depicted in Figure~\ref{fig-channel_test_gen}(a), in which the process (in blue) consists of three (independent) uses of a quantum channel $\mathcal{N}$. In Figure~\ref{fig-channel_test_gen}(b), we depict in blue a quantum multi-time process with memory, which can model non-Markovian dynamics~\cite{pollock2018nonmarkovian}. Measurements, or \textit{testers}, for multi-time processes consist of an input state to the process, several (possibly adaptive) interactions with the process over multiple time steps, and finally a measurement. The testers are depicted as the red processes in Figure~\ref{fig-channel_test_gen}.

The mathematical objects describing multi-time processes and testers are known as \textit{quantum combs}~\cite{CDP08,CDP09} and \textit{quantum strategies}~\cite{GW07}, and we provide formal definitions of these objects in Appendix~\ref{sec-multi_time_processes}. Briefly, quantum combs are multipartite operators that are defined as the Choi representations of multi-time quantum processes. For $k\in\mathbb{N}$, let $\mathcal{H}_{A,B}^{(k)}\equiv\mathcal{H}_{A_1}\otimes\mathcal{H}_{B_1}\otimes\mathcal{H}_{A_2}\otimes\mathcal{H}_{B_2}\otimes\dotsb\otimes\mathcal{H}_{A_k}\otimes\mathcal{H}_{B_k}$, such that the systems $A_1,A_2,\dotsc,A_k$ are the input systems and $B_1,B_2,\dotsc,B_k$ are the output systems; see Figure~\ref{fig-channel_test_gen}. Then, the set of quantum combs describing $r$-step multi-time quantum processes is defined as follows:
\begin{multline}
    \mathsf{COMB}_r(A_1,\dotsc,A_r;B_1,\dotsc,B_r)\\\coloneqq\Big\{N_r\in\Lin_{+}(\mathcal{H}_{A,B}^{(r)}):\exists\,N_k\in\Lin_{+}(\mathcal{H}_{A,B}^{(k)}),\,\Tr_{B_k}[N_k]=N_{k-1}\otimes\mathbbm{1}_{A_k},\\k\in\{2,3,\dotsc,r\},\,\Tr_{B_1}[N_1]=\mathbbm{1}_{A_1}\Big\}.\label{eq-strategies}
\end{multline}
We suppress the system labels and simply write $\mathsf{COMB}_r$ when the systems are unimportant or understood in the context being considered. For the process in blue in Figure~\ref{fig-channel_test_gen}(b), its Choi representation belongs to the set $\mathsf{COMB}_3$. We also note that $\mathsf{COMB}_1=\mathsf{CPTP}^{\prime}$. The set of inputs to multi-time processes, sometimes called \textit{co-strategies}, are defined as the Choi representations of multi-time processes in which the first input system is trivial. Specifically, let $\widetilde{\mathcal{H}}_{A,B}^{(k)}\equiv\mathbb{C}\otimes\mathcal{H}_{A_1}\otimes\dotsb\otimes\mathcal{H}_{B_{k-1}}\otimes\mathcal{H}_{A_k}$, for $k\in\mathbb{N}$. Then, the set of inputs to multi-time processes is defined as
\begin{align}
    &\mathsf{COMB}_r^{\ast}(A_1,\dotsc,A_r;B_1,\dotsc,B_{r-1})\nonumber\\
    \nonumber &\qquad\coloneqq\mathsf{COMB}_r(\varnothing,B_1,\dotsc,B_{r-1};A_1,\dotsc,A_r)\\
    &\qquad=\Big\{D_r\in\Lin_{+}(\widetilde{\mathcal{H}}_{A,B}^{(r)}):\exists\,D_k\in\Lin_{+}(\widetilde{\mathcal{H}}_{A,B}^{(k)}),\Tr_{A_k}[D_k]=D_{k-1}\otimes\mathbbm{1}_{B_{k-1}},\nonumber\\
    &\qquad\qquad\qquad k\in\{2,3,\dotsc,r\},\,\Tr_{A_1}[D_1]=1\Big\}.\label{eq-costrategies}
\end{align}
An example is the red process in Figure~\ref{fig-channel_test_gen}(b) (excluding the measurement), which belongs to the set $\mathsf{COMB}^{\ast}_3$. Again, we suppress the system labels and simply write $\mathsf{COMB}^{\ast}_r$ when the systems are unimportant or understood in the context being considered. Observe that every element of $\mathsf{COMB}^{\ast}_r$ is a positive semi-definite operator with unit trace. 
In particular, the elements of $\mathsf{COMB}^{\ast}_r$ are density operators, and we can think of them as multi-time analogues of quantum states.

The analogue of a POVM for multi-time processes, and thus the multi-time generalization of a channel test as we defined them above, is a \textit{mult-time tester}: a set $\{E^{(i)}\}_{i=1}^m$ of positive semi-definite \textit{test operators} $E^{(i)}\in\Lin(\mathcal{H}_{A,B}^{(r)})$ such that $0\leq E^{(i)}\leq S\otimes\mathbbm{1}_{B_r}$ for all $i\in\{1,2,\dotsc,m\}$, for some $S\in\mathsf{COMB}^{\ast}_r$, and $\sum_{i=1}^m E_{r}^{(i)}=S\otimes\mathbbm{1}_{B_r}$.

The multi-time analogues of the norms in \eqref{eq-strategy_1norm} and \eqref{eq-diamond_norm_dual_3} are the \textit{strategy $r$-norm} and its H\"{o}lder dual, which can be expressed as~\cite{Gut12}
\begin{align}
    \norm{O}_{\diamond r}&=\inf\Big\{t:t\geq 0,\,-tN\leq O\leq tN,\,N\in\mathsf{COMB}_r\Big\},\\
    \norm{O}_{\diamond r}^{\ast}&=\inf\Big\{t:t\geq 0,\,-tS\otimes\mathbbm{1}_{B_r}\leq O\leq tS\otimes\mathbbm{1}_{B_r},\,S\in\mathsf{COMB}^{\ast}_r\Big\},\label{eq-strategy_norm_dual_main}
\end{align}
for every Hermitian operator $O\in\Lin(\mathcal{H}_{A,B}^{(r)})$. It holds that $\norm{O}_{\diamond r}^{\ast}\geq\norm{O}_{\infty}$ for every Hermitian operator $O\in\Lin(\mathcal{H}_{A,B}^{(r)})$. For a multi-time test $\{E^{(i)}\}_{i=1}^m$ with $r$ time steps, it holds that every test operator $E^{(i)}$, $i\in\{1,2,\dotsc,m\}$, satisfies $\norm{E^{(i)}}_{\diamond r}^{\ast}\leq 1$. We also have the H\"{o}lder inequality
\begin{equation}\label{eq-Holder_inequality_strategy_norm}
    \Abs{\Tr[ON]}\leq\norm{O}_{\diamond r}^{\ast}\norm{N}_{\diamond r},
\end{equation}
for all Hermitian $O,N\in\Lin(\mathcal{H}_{A,B}^{(r)})$.

\subsection{Basics of online learning}\label{sec-online_learning_basics}

We view online learning a hypothesis class $\mathcal{F}\subseteq\mathsf{Y}^\mathsf{X}$ of functions as an interactive game between a learner and an adversary. 
In round $t$ of the interaction, the adversary challenges the learner with an input $x_t\in\mathsf{X}$.
Then, the learner predicts the output $f_t (x_t)\in\mathsf{Y}$ based on their current hypothesis $f_t$. (We focus on \textit{proper} online learning, where $f_t\in\mathcal{F}$ for all $t$.)
To conclude the round, the adversary provides the learner with feedback in the form of a loss $\ell_t (f_t (x_t))$, and the learner uses this piece of information to update\footnote{We assume the updating procedure to be deterministic. One can think of randomized learners as deterministically updating a probability distribution over hypotheses, whose performance is measured by its expected loss.} their hypothesis to $f^{(t+1)}$.
For our purposes, the target space is $\mathsf{Y}=[0,1]$ and every $\ell_t:[0,1]\to \mathbb{R}$ is convex and Lipschitz. 
Specific losses of interest are often of the form $\ell_t(y) = \ell(y-b_t)$ for some convex and Lipschitz $\ell:[0,1]\to [0,\infty)$ and $b_t\in [0,1]$; for example, this gives rise to the $L_p$-losses with $\ell(\cdot)=|\cdot|^p$. 
In these cases, we assume that the adversary provides the loss $\ell_t(f_t(x_t)) = \ell (f_t(x_t) - b_t)$ by explicitly revealing the value of $b_t$, and that $\ell$ is known in advance.
Note that in general, the $b_t$ here can be arbitrary. If there is an ``approximately true'' underlying concept $f^\ast\in\mathcal{F}$ such that $|b_t - f^\ast(x_t)|\leq\varepsilon/3$ holds for all $t$, then we speak of a realizable scenario, otherwise we call the setting non-realizable. 

\paragraph{Regret-bounded online learning.}

One way of phrasing desiderata in online learning is in terms of bounds on the difference between the incurred and the in hindsight optimal incurred loss, the so-called regret. Namely, after $T$ rounds of interaction, we say that the learner has incurred a regret of\footnote{Here, we implicitly assume that the minimum exists, which will be the case for our purposes because of compactness and continuity. More generally, one may replace the minimum by an infimum.} 
\begin{equation}\label{eq-regret_general}
    R_T = \sum_{t=1}^T \ell_t(f_t(x_t)) - \min_{f\in\mathcal{F}}\sum_{t=1}^T \ell_t(f(x_t))\, .
\end{equation}
The goal of an online learner now is to achieve a small regret, in particular scaling sub-linearly with $T$. 
Note that this model makes sense even if no restrictions on how the adversary chooses the $\ell_t$ (or, in the more concrete case of $L_p$-losses, the $b_t$ above) are imposed.

We make use of the following lemma throughout this work.

\begin{lemma}[Regret bound]\label{lem-regret_via_pseudoregret}
    Consider the online learning scenario described above, in which the functions $\ell_t$ are convex and differentiable. Then, for every $f\in\mathcal{F}$, it holds that
    \begin{equation}\label{eq-regret_via_pseduoregret_general}
        \sum_{t=1}^T\Big(\ell_t(f_t(x_t))-\ell_t(f(x_t))\Big)\leq\sum_{t=1}^T\Big(\ell_t'(f_t(x_t))(f_t(x_t)-f(x_t))\Big).
    \end{equation}
    The regret in \eqref{eq-regret_general} is therefore bounded from above as 
    \begin{equation}\label{eq-regret_via_pseduoregret_general_2}
        R_T\leq \sum_{t=1}^T \ell'_t(f_t(x_t))f_t(x_t)-\min_{f\in\mathcal{F}}\sum_{t=1}^T\ell'_t(f_t(x_t))f(x_t).
    \end{equation}
\end{lemma}

\begin{proof}
    Since $\ell_t$ is convex and differentiable, it readily follows that
    \begin{equation}
        \ell_t(y_1-y_2)\leq \ell_t'(y_1)(y_1-y_2),
    \end{equation}
    for all $y_1,y_2\in[0,1]$. (See, e.g., Ref.~\cite[Chapter~2]{OCO_hazan_book}.) Consequently, for every $f\in\mathcal{F}$,
    \begin{equation}
        \ell_t(f_t(x_t))-\ell_t(f(x_t))\leq \ell_t'(f_t(x_t))(f_t(x_t)-f(x_t)),
    \end{equation}
    for all $t\in\{1,2,\dotsc,T\}$. Therefore,
    \begin{align}
        \sum_{t=1}^T\Big(\ell_t(f_f(x_t))-\ell_t(f(x_t))\Big)&\leq\sum_{t=1}^T\Big(\ell_t'(f_t(x_t))(f_t(x_t)-f(x_t))\Big)\\
        \nonumber
        &=\sum_{t=1}^T \ell_t'(f_t(x_t))f_t(x_t)-\sum_{t=1}^T\ell'_t(f_t(x_t))f(x_t)\\
        \nonumber
        &\leq\sum_{t=1}^T\ell'_t(f_t(x_t))-\min_{f\in\mathcal{F}}\sum_{t=1}^T\ell'_t(f_t(x_t))f(x_t).
        \nonumber
    \end{align}
    The first inequality is \eqref{eq-regret_via_pseduoregret_general}. The inequality in \eqref{eq-regret_via_pseduoregret_general_2} follows because the function $f\in\mathcal{F}$ is arbitrary.
\end{proof}

Next, we recall results from 
Refs.\ \cite{rakhlin2015online, rakhlin2015sequential}, which demonstrate how to obtain regret bounds in terms of sequential complexity measures of the hypothesis class.
To formulate these results, we introduce the following pieces of notation: We use $\mathbf{z}=(\mathbf{z}_t)_{t=1}^T$ with labeling functions $\mathbf{z}_t:\{\pm 1\}^{t-1}\to \mathsf{Z}$ to describe a complete rooted binary tree of depth $T$ with nodes labeled by elements of $\mathsf{Z}$.
That is, if we arrive at a node by following a path $\pi=(\pi_s)_{s=1}^{t-1}\in\{\pm 1\}^{t-1}$ of length $t-1$ from the root, then $\mathbf{z}$ assigns the label $\mathbf{z}_t(\pi)$ to that node. 
For notational convenience, if $\pi\in\{\pm 1\}^T$ is a path of length $T>t-1$, we identify $\mathbf{z}_t(\pi) = \mathbf{z}_t(\pi_1,\ldots,\pi_{t-1})$.
We can now introduce a notion of sequential covering to capture the effective size of a function class.

\begin{definition}[Sequential covering numbers \cite{rakhlin2015online}]\label{definition:sequential-covering}
    Let $\mathcal{G}\subseteq\mathbb{R}^{\mathsf{Z}}$, let $\mathbf{z}=(\mathbf{z}_t)_{t=1}^T$ be a complete rooted binary tree of depth $T$, and let $\varepsilon>0$, $p\geq 1$.
    We call a set $V$ of $\mathbb{R}$-valued complete binary trees of depth $T$ a \emph{sequential $p$-norm $\varepsilon$-cover of $\mathcal{G}$ on $\mathbf{z}$} if the following holds:
    \begin{equation}
        \forall g \in \mathcal{G}~\forall\pi\in\{\pm 1\}^T~\exists \mathbf{v}\in V:~\left(\frac{1}{T}\sum_{t=1}^T \lvert \mathbf{v}_t(\pi) - g(\mathbf{z}_t(\pi))\rvert^p\right)^{1/p}\leq\varepsilon\, .
    \end{equation}
    The \emph{sequential $p$-norm $\varepsilon$-covering number of $\mathcal{G}$ on $\mathbf{z}$} is defined to be
    \begin{equation}
        N_\mathbf{z} (\mathcal{G}, \varepsilon,p)
        \coloneqq\inf\left\{\lvert V\rvert : V\textrm{ is a sequential }p\textrm{-norm }\varepsilon\textrm{-cover of }\mathcal{G}\textrm{ w.r.t.\  }\mathbf{z}\right\}\, .
    \end{equation}
    We write $N_T (\mathcal{G}, \varepsilon,p)=\sup_{\mathbf{z}} N_\mathbf{z} (\mathcal{G}, \varepsilon,p)$, where the supremum is over trees of depth $T$.
\end{definition}

Note that \Cref{definition:sequential-covering} does not require a cover to consist of $\mathbb{R}$-valued trees that can be realized within $\mathcal{G}$. In that sense, the above notion is one of exterior covering.
Also, it will be useful to observe that sequential covering numbers satisfy the monotonicity relation $N_{\mathbf{z}}(\mathcal{G},\varepsilon,p)\leq N_{\mathbf{z}}(\mathcal{G},\varepsilon,q)$ for $1\leq p\leq q\leq\infty$.

Similarly to how empirical metric entropies, defined as the logarithm of covering numbers, control the generalization error in probably approximately correct learning, sequential metric entropies, the logarithm of sequential covering numbers, can be used to bound the regret in online learning. This result goes back to 
Refs.\ \cite{rakhlin2015sequential, rakhlin2015online}, we state it in a form similar to 
Ref.\ \cite[Theorem 9]{ACH+19}.

\begin{theorem}[Regret bound from sequential covering {\cite[Theorem 3]{rakhlin2015sequential} and \cite[Theorem 7]{rakhlin2015online}}]\label{theorem:regret-bound-sequential-covering}
    Let $\mathcal{F}\subseteq [0,1]^{\mathsf{X}}$. 
    For every $t\in\mathbb{N}_{\geq 1}$, let $\ell_t:[0,1]\to \mathbb{R}$ be convex and $L$-Lipschitz.
    Then, there exists an online learning strategy that, when presented sequentially with $x_1,\ldots,x_T\in\mathsf{X}$ and associated loss functions $\ell_1,\ldots,\ell_T$, outputs a sequence $f_1,\ldots,f_T\in\mathcal{F}$ of hypotheses whose regret is bounded as
    \begin{align}
        \sum_{t=1}^T \ell_t(f_t(x_t)) - \min_{f\in\mathcal{F}}\sum_{t=1}^T \ell_t(f(x_t))
        &\leq 2LT \sup_{\mathbf{x}}\inf_{\alpha>0} \left\{4\alpha + \frac{12}{\sqrt{T}}\int_{\alpha}^{1} \sqrt{\log N_{\mathbf{x}}(\mathcal{F},\beta,2)}~\mathrm{d}\beta\right\}\\
        &\leq 2LT \inf_{\alpha>0} \left\{4\alpha + \frac{12}{\sqrt{T}}\int_{\alpha}^{1} \sqrt{\log N_{T}(\mathcal{F},\beta,2)}~\mathrm{d}\beta\right\}\, .
        \nonumber
    \end{align}
    Here, the $\sup_{\mathbf{x}}$ is a supremum over all complete binary trees of depth $T$ with nodes labeled by elements of $\mathsf{X}$.
\end{theorem}

\paragraph{Mistake-bounded online learning.}

As an alternative to measuring the performance of an online learner in terms of regret, we can count the number of rounds in which the learner incurs a loss that exceeds a certain threshold. More formally, given $\varepsilon\in (0,1)$, we say that the learner makes an $\varepsilon$-mistake in round $t$ if $\ell_t (f_t(x_t)) > \varepsilon$. The goal of the online learner then becomes to make only a small number of $\varepsilon$-mistakes. Note that this mistake-bounded model of online learning only makes sense in the realizable scenario, since the number of $\varepsilon$-mistakes is in general infinite in the non-realizable scenario.

To conclude our introductory discussion of online learning, we note a well known connection between the models of regret- and mistake-bounded online learning. Informally, we can say that good regret bounds lead to good mistake bounds, and the next result makes this formal.

\begin{lemma}[From regret to mistake bounds (compare, e.g., 
Ref.\ {\cite[Section 3.3]{ACH+19}} or {\cite[Corollary 2.1.4]{lowe2021learning}})]\label{lemma:regret-to-mistake-template}
	Let $\mathcal{F}\subseteq\mathbb{R}^{\mathsf{X}}$.
    Suppose we have an online learner for $\mathcal{F}$ that is sequentially presented with $x_1,\ldots,x_T\in\mathsf{X}$ and losses evaluated according to $\ell_t(\cdot) = \lvert (\cdot ) - b_t \rvert$, where there exists $f_\ast\in\mathcal{F}$ such that $|b_t - f_\ast(x_t)|\leq\varepsilon/3$ holds for all $t$. For an update rule that results in a sequence $f_1,\ldots,f_T\in\mathcal{F}$ of outputs from the learner, suppose that the regret is bounded as $R_T\leq h_1(\varepsilon,T) + h_2(\varepsilon)$, where $h_1:(0,1)\times \mathbb{N}_{\geq 1}\to\mathbb{R}_{\geq 0}$ and $h_2:(0,1)\to \mathbb{R}_{\geq 0}$.
    Assume that $h_1(\varepsilon,T)\in o(T)$ for every fixed $\varepsilon$, and that $h_2(\varepsilon)<2\varepsilon/3$.
    Let $T^\ast = T^\ast (h_1,h_2,\varepsilon)$ be the smallest natural number such that 
    \begin{equation}
        T'> T^\ast \textrm{ implies } T' > \frac{h_1(\varepsilon, T')}{(2\varepsilon/3) - h_2(\varepsilon)}\, .
    \end{equation}
    Then, applying the update rule only in rounds in which the learner makes an $\varepsilon$-mistake, and outputting the previous hypothesis otherwise, leads to a total number of $\varepsilon$-mistakes bounded by $T^\ast$, independently of the overall number of rounds.
\end{lemma}

Notice that the online learning procedure achieving the mistake bound $T^\ast$ in \Cref{lemma:regret-to-mistake-template} is \emph{mistake-driven}: It updates the hypothesis only when an $\varepsilon$-mistake is made. After a round without an $\varepsilon$-mistake, the learner just proceeds without changing their hypothesis.

\begin{proof}
    Let $T'$ denote the number of $\varepsilon$-mistakes (i.e., the number of updates) and let us focus on the subsequence of rounds in which the learner makes a mistake.
	First note that, by assumption, there is a function $f_\ast\in\mathcal{F}$ such that the prediction $f_t(x_t)$ achieves loss $\ell(f_\ast(x_t))\leq \varepsilon/3$ for all $t$. Hence, the optimal accumulated loss after $T'$ rounds is $\leq T'\varepsilon/3$.    
    As the described procedure applies the update rule only when a loss $>\varepsilon$ is incurred, this mistake-driven online learning procedure incurs an accumulated loss of $> T'\varepsilon$.
    Thus, its regret is $> 2T'\varepsilon/3$.
    Comparing this with the assumed regret upper bound of $R_{T'}\leq h_1(\varepsilon,T') + h_2(\varepsilon)T'$ and rearranging, we see that $T'\leq \tfrac{h_1(\varepsilon, T')}{(2\varepsilon/3) - h_2(\varepsilon)}$. Thus, by the definition of $T^\ast$, we conclude $T'\leq T^\ast$.
\end{proof}

\subsection{The multiplicative weights framework and corresponding guarantees}

In this section, we provide a brief overview of the multiplicative weights framework and refer to Ref.~\cite{arora2012multiplicative} for a more comprehensive review.
This toolkit has a variety of applications in areas such as game theory and economics~\cite{freund1999adaptive}, machine learning~\cite{littlestone1988learning, littlestone1994weighted}, and semidefinite programming~\cite{arora2016MMW, arora2005fast}. In quantum computing, the matrix multiplicative weights algorithm was used to prove that the complexity classes QIP (problems with a quantum interactive proof system) and PSPACE (problems solvable in polynomial amount of space) coincide~\cite{jain2011qip}.

Consider an interactive game, in which a learner is tasked with picking the best option from a set of $d\in\mathbb{N}$ decisions over $T\in\mathbb{N}$ rounds of interaction. In each round $1\leq t\leq T$, every decision $i\in\{1,2,\dotsc,d\}$ is associated with a cost $m_i^{(t)}\in[-1,1]$. Upon making their decision, the learner is informed of the associated cost, which they can use to make their decisions in subsequent rounds. The learner's goal is to output a sequence of decisions over the $T$ rounds such that their total accumulated cost after $T$ rounds is minimized.

In the multiplicative weights framework, the decision in round $1\leq t\leq T$ is made according to a $d$-dimensional probability vector $\Vec{p}^{(t)}=(p^{(t)}_1,p^{(t)}_2,\dotsc,p^{(t)}_d)$. Explicitly, $p^{(t)}_i$ is the probability of making the decision $i\in\{1,2,\dotsc,d\}$. 
The expected cost of the distribution $\Vec{p}^{(t)}$ is then $\Vec{m}^{(t)}\cdot\Vec{p}^{(t)}$, where $\Vec{m}^{(t)}=(m^{(t)}_1,\dotsc,m^{(t)}_d)$ is the vector of costs. 
The expected accumulated cost after $T$ rounds is thus given by $\sum_{t=1}^{T}\Vec{m}^{(t)}\cdot\Vec{p}^{(t)}$.
The \textit{multiplicative weights update} (MWU) algorithm, presented in Algorithm~\ref{alg:mwu_general}, is a method for obtaining a sequence of probability distributions over decisions based on the costs incurred in the previous rounds. 
 
\begin{algorithm} 
    \begin{algorithmic}[1]
    \Require{$\eta\in(0,\frac{1}{2}]$; $T\in\mathbb{N}$ rounds; $d\in\mathbb{N}$ decisions; initial weights $\Vec{w}^{(1)} = (1,1,\dotsc,1)$.}
    \For{\texttt{$t = 1, 2, \ldots T$}}
        \State Output the decision/estimate $\Vec{p}^{(t)}=\frac{\Vec{w}^{(t)}}{\phi^{(t)}}$, $\phi^{(t)}=\sum_{i=1}^d w_i^{(t)}$.
        \State Receive the cost vector $\Vec{m}^{(t)}=(m_1^{(t)},m_2^{(t)},\dotsc,m_d^{(t)})$, $m_i^{(t)}\in[-1,1]$, $i\in\{1,2,\dotsc,d\}$.
        \State Update the weights as $w_i^{(t+1)}=w_i^{(t)}(1-\eta m_i^{(t)})$, $i\in\{1,2,\dotsc,d\}$.
        \State Set $\Vec{w}^{(t+1)}=(w_1^{(t+1)},w_2^{(t+1)},\dotsc,w_d^{(t+1)})$.
    \EndFor
    \end{algorithmic}
    \caption{Multiplicative weights update (MWU) algorithm~\cite{arora2012multiplicative}}\label{alg:mwu_general}
\end{algorithm}

\begin{algorithm}
     \begin{algorithmic}[1]
     \Require {$\eta \in (0, 1]$; Initialize $W^{(1)} = \mathbbm{1}_d$}
     \For{\texttt{$t = 1, 2, \ldots, T$}}
        \State Output the decision/estimate $\omega^{(t)} = \frac{W^{(t)}}{\Tr[W^{(t)}]}$.
        \State Receive the cost matrix $L^{(t)}$, $ -\mathbbm{1}_d \leq L^{(t)} \leq \mathbbm{1}_d$.
        \State Update the weight matrix: $W^{(t+1)} = \exp\!\left[-\eta \sum_{t' = 1}^t L^{(t')}\right]=\exp[\log(W^{(t)})-\eta L^{(t)}]$.
        \EndFor
     \end{algorithmic}
     \caption{Matrix multiplicative weights (MMW) algorithm~\cite{tsuda2005MMW,arora2016MMW}}\label{alg:mmw}
\end{algorithm}

Arora, Hazan and Kale~\cite{arora2012multiplicative} have shown that if a learner makes their decisions according to the MWU algorithm, then their expected accumulated cost only grows logarithmically in $d$, the number of possible decisions.

\begin{theorem}[{\cite[Theorem~2.1 \& Corollary~2.2]{arora2012multiplicative}}]\label{thm:MWU}
    Consider the setting of an interactive game, as described above, with cost vectors $\Vec{m}^{(t)}$ satisfying $m_i^{(t)}\in[-1,1]$ for all $i\in\{1,2,\dotsc,d\}$ and for all $t\in\{1,2, \dotsc, T\}$, and let $\Vec{q}$ be an arbitrary probability distribution over the $d$ decisions. Using the MWU algorithm (Algorithm~\ref{alg:mwu_general}), the expected accumulated cost over $T\in\mathbb{N}$ rounds is bounded from above as
    \begin{equation}
        \sum_{t=1}^T \Vec{m}^{(t)} \cdot \Vec{p}^{(t)} 
        \leq \sum_{t=1}^T (\Vec{m}^{(t)}+\eta|\Vec{m}^{(t)}|)\cdot\Vec{q}+\frac{\log d}{\eta}
        \leq \sum_{t=1}^T \Vec{m}^{(t)}\cdot\Vec{q} + \eta T + \frac{\log d}{\eta}. \label{Eq-mwu-arora}
    \end{equation}
\end{theorem}

The \textit{matrix multiplicative weights} (MMW) algorithm~\cite{tsuda2005MMW, arora2016MMW} is a generalization of the MWU algorithm to costs specified $d\times d$ Hermitian matrices $L^{(t)}$ that satisfy $-\mathbbm{1}_d\leq L^{(t)}\leq\mathbbm{1}_d$; equivalently, $\norm{L^{(t)}}_{\infty}\leq 1$. Here, the decisions are described by density operators $\omega^{(t)}$, and the expected cost in round $t$ is equal to $\Tr[L^{(t)}\omega^{(t)}]$. The MMW algorithm, presented in Algorithm~\ref{alg:mmw}, is a method for obtaining a sequence of density operators $\omega^{(1)},\omega^{(2)},\dotsc,\omega^{(T)}$, based on the costs incurred in the previous rounds.

A bound on the expected accumulated cost for the MMW algorithm has been shown in Ref.~\cite[Theorem~3.1]{arora2016MMW}. By modifying the arguments in Ref.~\cite{arora2016MMW} via use of the relative entropy, we provide a bound on the expected accumulated cost for the MMW algorithm that can in general be tighter than the one obtained in Ref.~\cite[Theorem~3.1]{arora2016MMW}.

\begin{proposition}[Bound on the expected accumulated cost for the MMW algorithm]\label{prop-MMW_amortization_regret_bound}
    Let $\rho$ be an arbitrary density operator.
    Let $T\in\mathbb{N}$ be the number of rounds of interaction, and consider a sequence $L^{(1)},L^{(2)},\dotsc,L^{(T)}$ of cost matrices in dimension $d\in\{2,3,\dotsc\}$ along with the updates $\omega^{(t)}$ provided by the MMW algorithm in Algorithm~\ref{alg:mmw}. Then, the expected accumulated cost over the $T$ rounds is bounded from above as 
    \begin{equation}\label{eq-MMW_amortization_regret_bound}
        \sum_{t=1}^T \Tr[L^{(t)}\omega^{(t)}]\leq \Tr\!\left[\rho\left(\sum_{t=1}^TL^{(t)}\right)\right]+\eta\sum_{t=1}^T\Tr[(L^{(t)})^2\omega^{(t)}]+\frac{\log d-H(\rho)}{\eta},
    \end{equation}
    where $H(\rho)\coloneqq-\Tr[\rho\log\rho]$ is the von Neumann entropy of $\rho$.
\end{proposition}
\begin{proof}
    See \Cref{sec-MMW_entropic}, where we also describe how the bound in Ref.~\cite[Theorem~3.1]{arora2016MMW} arises as a special case of our bound in \Cref{eq-MMW_amortization_regret_bound}.
\end{proof}

\begin{remark}[The Hedge algorithm]
    If the cost matrices $L^{(t)}$ in the MMW algorithm are all diagonal in the same basis, then Algorithm~\ref{alg:mmw} reduces to the so-called \textit{Hedge} algorithm, introduced by Freund and Schapire~\cite{freund1997decision}. We state this algorithm in \Cref{sec-MMW_entropic}, and in \Cref{cor-Hedge_amortization_regret_bound} we state the Hedge algorithm counterpart to \Cref{prop-MMW_amortization_regret_bound} above.~\qedgen
\end{remark}

\subsection{Problem statement: Online learning classes of quantum channels}\label{subsection:prob-statement-online-learning-qchannels}

Our task is to online learn a class $\mathsf{C}\subseteq\mathsf{CPTP}_n$ of $n$-qubit quantum channels, in the sense of predicting the quantities $\Tr[M_{R,B}^{(t)}\mathcal{N}_{A\to B}(\rho_{R,A}^{(t)})]$, with the state-measurement pairs $(\rho_{R,A}^{(t)},M_{R,B}^{(t)})$ provided by an adversary, where $R$ is an arbitrary finite-dimensional reference system; see Figure~\ref{fig:channel_tests}. Let us now cast this problem in terms of the general framework of online learning laid out in Section~\ref{sec-online_learning_basics}.

We fix finite-dimensional Hilbert spaces $\mathcal{H}_A$ and $\mathcal{H}_B$. Then, the input set/domain $\mathsf{X}$ comprises state-measurement pairs $(\rho_{R,A},M_{R,B})$, where $R$ is an arbitrary finite reference system. Equivalently, due to \eqref{eq-choi-matrix-to-channel-prob}, $\mathsf{X}$ comprises channel test operators $E_{A,B}$. Precisely,
\begin{align}
    \mathsf{X}&=\Big\{E_{A,B}\in\Lin(\mathcal{H}_A\otimes\mathcal{H}_B): \exists\,\sigma_A\textrm{ s.t. } 0\leq E_{A,B}\leq\sigma_A\otimes\mathbbm{1}_B,\,\sigma_A\geq 0,\,\Tr[\sigma_A]=1\Big\}\\
    &=\Big\{E_{A,B}\in\Lin(\mathcal{H}_A\otimes\mathcal{H}_B):E_{A,B}\geq 0,\,\norm{E_{A,B}}_{\diamond 1}^{\ast}\leq 1\Big\}\, ,\nonumber
\end{align}
where in the second equality we have used \Cref{eq-diamond_norm_dual_3}.
The output set is $\mathsf{Y}=[0,1]$, and $\mathcal{F}$ is defined by the subclass $\mathsf{C}\subseteq\mathsf{CPTP}_n$ of $n$-qubit channels (in which case $\mathcal{H}_A\cong\mathcal{H}_B\cong(\mathbb{C}^2)^{\otimes n}$), such that for every $\mathcal{N}\in\mathsf{C}$ we define the function $f_{\mathcal{N}}$ as $f_{\mathcal{N}}(E)=\Tr[EC(\mathcal{N})]$ for all $E\in\mathsf{X}$. In other words, 
\begin{equation}
    \mathcal{F}=\Big\{f_{\mathcal{N}}:\mathsf{X}\to[0,1]:\mathcal{N}\in\mathsf{C},\,f_{\mathcal{N}}(E)=\Tr[EC(\mathcal{N})]~\forall\,E\in\mathsf{X}\Big\}.
\end{equation}
Notably, due to \Cref{eq-choi-matrix-to-channel-prob}, we can see that our online learning task is equivalently formulated as the task of learning a class $\mathsf{C}^{\prime}$ of \textit{Choi matrices} defined by the channels in $\mathsf{C}$, where
\begin{equation}
    \mathsf{C}^{\prime}\coloneqq\Big\{N:N=C(\mathcal{N}),\,\mathcal{N}\in\mathsf{C}\Big\}\subseteq\mathsf{CPTP}_n^{\prime}\, .
\end{equation}
We therefore view the adversary as providing channel test operators $E_{A,B}^{(t)}$ and having the learner predict the quantities $\Tr[E_{A,B}^{(t)}C_{A,B}^{\mathcal{N}}]$ based on hypotheses for the \textit{Choi matrix} of the unknown channel. We denote the learner's hypothesis Choi matrices by $N_{A,B}^{(t)}$ for $t\in\{1,2,\dotsc,T\}$.

As outlined in \Cref{sec-online_learning_basics}, we will evaluate the performance of an online learner in terms of either the regret or the number of mistakes. 
That is, on the one hand, the learner aims to achieve a small regret
\begin{equation}\label{eq-regret-channels}
    R_T = \sum_{t=1}^T \ell_t(\Tr[E_{A,B}^{(t)}N_{A,B}^{(t)}]) - \min_{\mathcal{N}\in\mathsf{C}}\sum_{t=1}^T \ell_t(\Tr[E_{A,B}^{(t)}C_{A,B}^{\mathcal{N}}])\, ,
\end{equation}
where the losses $\ell_t$ are revealed by the adversary.
More precisely, we will aim for regret bounds scaling as $\mathcal{O}(\sqrt{T\cdot\text{poly}(n)})$.
On the other hand, in the realizable scenario, where there exists a (to the learner unknown) channel $\mathcal{N}_{A\to B}\in\mathsf{C}$ such that all losses take the form
\begin{equation}\label{eq-loss-channels}
    \ell_t(\cdot )= \ell( (\cdot) - b_t)\, ,
\end{equation}
where $\ell$ is some fixed function and each $b_t$ satisfies $\lvert b_t - \Tr[E_{A,B}^{(t)}C_{A,B}^{\mathcal{N}}] \rvert \leq \varepsilon/3$, the learner aims to achieve a small number of $\varepsilon$-mistakes.
Here, our goal will be to guarantee mistake bounds scaling as $\mathcal{O}(\text{poly}(n, \varepsilon^{-1}))$.
We can summarize the formulation of our channel online learning task as follows:

\begin{problem}[Online learning classes of quantum channels]\label{problem:online-learning-classes-of-channels}
    Consider a subset $\mathsf{C} \subseteq \mathsf{CPTP}_n$ of $n$-qubit quantum channels, and let $\mathcal{N}_{A\to B}\in\mathsf{C}$, with Choi representation $C_{A,B}^{\mathcal{N}}\in\mathsf{C}^{\prime}$, be unknown. Given a sequence of $T\in\mathbb{N}$ interactive rounds, in which two-outcome channel test operators $E_{A,B}^{(1)}, E_{A,B}^{(2)}, \cdots,E_{A,B}^{(T)}\in\mathsf{X}$ are presented sequentially by an adversary, the problem is to output a sequence of Choi matrices $N_{A,B}^{(1)},N_{A,B}^{(2)},\dotsc,N_{A,B}^{(T)}\in\mathsf{C}^{\prime}$, such that for losses $\ell_t$ as defined in \eqref{eq-loss-channels}, the regret in \eqref{eq-regret-channels} scales as $\mathcal{O}(\sqrt{T\cdot\text{poly}(n)})$ and the number of $\varepsilon$-mistakes scales as $\mathcal{O}(\text{poly}(n, \varepsilon^{-1}))$.~\qedgen
\end{problem}

\begin{remark}\label{rem-online_state_learning}
    In this work, we primarily consider the case that the input system dimension $d_A$ and the output system dimension $d_B$ are equal and satisfy $d_A=d_B=d=2^n$, although Problem~\ref{problem:online-learning-classes-of-channels} applies also to quantum channels with different input and output system dimensions. In particular, if $d_A=1$ and $d_B=d=2^n$, then every channel is a state-preparation channel for some $n$-qubit quantum state, and Problem~\ref{problem:online-learning-classes-of-channels} reduces to online learning of quantum states, as considered in Ref.~\cite{ACH+19}.~\qedgen
\end{remark}

\begin{remark}
    Successfully solving Problem~\ref{problem:online-learning-classes-of-channels} does not imply learning of the unknown channel with respect to the diamond norm, i.e., the learner's hypotheses $N_{A,B}^{(t)}$ could be very far from the true Choi matrix $C_{A,B}^{\mathcal{N}}$ with respect to the strategy 1-norm in \eqref{eq-strategy_1norm}. In our scenario, the learner's only goal is to ensure that their hypotheses are such that $\Tr[E_{A,B}^{(t)}N_{A,B}^{(t)}]$ well approximates $\Tr[E_{A,B}^{(t)}C_{A,B}^{\mathcal{N}}]$ in most rounds, and this can be achieved by hypotheses that are not necessarily close to the true Choi matrix with respect to the strategy 1-norm.~\qedgen
\end{remark}

\subsection{Obstacles to online learning via the Choi state}\label{sec-online_learning_Choi_state}

As pointed out above, the problem of online learning classes of quantum channels (Problem~\ref{problem:online-learning-classes-of-channels}) is equivalent to the problem of online learning classes of \textit{Choi matrices}. A natural first strategy for solving this problem might then be to simply online learn the Choi \textit{state} of the unknown channel using the protocols presented in Ref.~\cite{ACH+19} for online learning of quantum states, such as the MMW algorithm (Algorithm~\ref{alg:mmw}). However, we immediately encounter two issues.
    
First, the MMW algorithm (Algorithm~\ref{alg:mmw}) cannot be applied out of the box: while Choi states $\Phi_{A,B}^{\mathcal{N}}$ of quantum channels $\mathcal{N}_{A\to B}$ have unit trace, they also have to satisfy $\Tr_B[\Phi_{A,B}^{\mathcal{N}}]=\mathbbm{1}_A/d_A$, which the iterates of Algorithm~\ref{alg:mmw} will generally not guarantee. Furthermore, the proof of Ref.\ \cite[Theorem~3.1]{arora2016MMW}, as well as the proof of Proposition~\ref{prop-MMW_amortization_regret_bound} above, relies crucially on the fact that the iterates of the MMW algorithm have unit trace. So, potential modifications of the update rule would have to simultaneously ensure the unit trace and the partial trace conditions.

We can modify the MMW algorithm by adding a \textit{projection step}: for every iterate $W_{A,B}^{(t)}$ of the MMW algorithm, we find the closest Choi state $\rho_{A,B}^{(t)}$ with respect to relative entropy and use that as our hypothesis for the unknown Choi state. In other words, we let
\begin{equation}
    \rho_{A,B}^{(t)}\coloneqq\argmin_{\rho_{A,B}}\Big\{D(\rho_{A,B}\Vert \omega_{A,B}^{(t)}):\rho_{A,B}\geq 0,\,\Tr_B[\rho_{A,B}]=\frac{\mathbbm{1}_A}{d_A}\Big\},
\end{equation}
where $\omega_{A,B}^{(t)}=W_{A,B}^{(t)}/\Tr[W_{A,B}^{(t)}]$ and the relative entropy $D(P\Vert Q)$ of two positive semi-definite operators $P$ and $Q$ is defined as~\cite{Wat18_book}
\begin{equation}\label{eq-relative_entropy}
    D(P\Vert Q)\coloneqq\left\{\begin{array}{l l} \Tr[P\log P-P\log Q] & \text{if }\text{supp}(P)\subseteq\text{supp}(Q), \\ +\infty & \text{otherwise}. \end{array}\right.
\end{equation}
We present the modified MMW algorithm in Algorithm~\ref{alg:projected-mmw-choi-state}. Then, because the relative entropy is a Bregman divergence~\cite{petz2007bregman}, we can regard Algorithm~\ref{alg:projected-mmw-choi-state} as a particular instance of the \textit{online mirror descent} algorithm~\cite[Section~5.3]{OCO_hazan_book}. Consequently, we obtain a regret bound for Choi states that is similar to the regret bound in Proposition~\ref{prop-MMW_amortization_regret_bound}, thus similar to the MMW regret bound in Ref.~\cite{ACH+19} for online learning of quantum states. We provide the details in Appendix~\ref{sec-projected_MMW}.

\begin{algorithm}
     \begin{algorithmic}[1]
     \Require {$\eta \in (0, 1)$; Initialize $W_{A,B}^{(1)} = \mathbbm{1}_A \otimes \mathbbm{1}_B$ and $\rho_{A,B}^{(1)}=\frac{1}{d_Ad_B}\mathbbm{1}_A\otimes\mathbbm{1}_B$}
     \For{\texttt{$t = 1, 2, \ldots, T$}}
        \State Output the decision/estimate $\rho_{A,B}^{(t)}$
        \State Receive the cost matrix $L_{A,B}^{(t)}$, $-\mathbbm{1}_{A,B}\leq L_{A,B}^{(t)}\leq\mathbbm{1}_{A,B}$
        \State Update the weight matrix: $W_{A,B}^{(t+1)}=\exp[\log(W_{A,B}^{(t)})-\eta L_{A,B}^{(t)}]$, $\omega_{A,B}^{(t+1)}=W_{A,B}^{(t+1)}/\Tr[W_{A,B}^{(t+1)}]$
        \State Project: $\rho_{A,B}^{(t+1)}\coloneqq\argmin_{\rho_{A,B}}\Big\{D(\rho_{A,B}\Vert\omega_{A,B}^{(t+1)}):\rho_{A,B}\geq 0,\,\Tr_B[\rho_{A,B}]=\mathbbm{1}_A/d_A\Big\}$
        \EndFor
     \end{algorithmic}
     \caption{Projected MMW algorithm for Choi states of quantum channels }\label{alg:projected-mmw-choi-state}
\end{algorithm}

The second issue is that the above strategy of learning the Choi state of the unknown channel leads to favorable regret and mistake bounds, but only as long as we modify Problem~\ref{problem:online-learning-classes-of-channels} as follows: instead of predicting expectation values of the form $\Tr[E_{A,B}C_{A,B}^{\mathcal{N}}]$, where $C_{A,B}^{\mathcal{N}}$ is the Choi matrix of the unknown channel $\mathcal{N}_{A\to B}$, we aim to predict values of the form $\Tr[E_{A,B}\Phi_{A,B}^{\mathcal{N}}]$, where $\Phi_{A,B}^{\mathcal{N}}$ is the Choi state of $\mathcal{N}_{A\to B}$. Now, because of \Cref{eq-choi-channel-translation}, we see that $\Tr[E_{A,B}C_{A,B}^{\mathcal{N}}]=d_A\Tr[E_{A,B}\Phi_{A,B}^{\mathcal{N}}]$. Thus, even though we show that the $\mathcal{O}(L\sqrt{Tn})$ regret bound of \cite{ACH+19} carries over to (properly) online learning the Choi state, this implies only a regret bound of $\mathcal{O}(dL\sqrt{Tn})=\mathcal{O}(2^n L\sqrt{Tn})$ when it comes to our actual task of learning the Choi matrix (and thus the channel), which has a favorable square root scaling in $T$ but an unfavorable exponential scaling in the number of qubits\footnote{Note that the dimension factor $d=d_A=2^n$ is dimension of the input system of the channel. As explained in Remark~\ref{rem-online_state_learning}, in the case $d_A=1$, our problem reduces to online learning of quantum states, and the dimension factor in the regret bound becomes $1$, as expected~\cite{ACH+19}.}.

\section{Online learning upper bounds}

\subsection{Regret bound for channels of bounded gate complexity}\label{section:online-learning-bounded-complexity}

In this section, we show online learnability for quantum channels of bounded gate complexity\footnote{While we focus on exact gate complexity here, an extension to approximate gate complexities, defined in terms of the number of two-qubit channel gates sufficient to achieve a desired degree of approximation in diamond norm distance, is straightforward based on the triangle inequality.}. We say that a quantum channel $\mathcal{N}$ has (exact) gate complexity (at most) $G\in\mathbb{N}$ \cite{BCH+19,HFK+21} if there exists a (not necessarily geometrically) two-local\footnote{The results naturally extend to circuits with $k$-local gates, compare Ref.\ \cite[Supplementary~Note~3, Remark~1]{caro2022generalization}.} quantum circuit with $G$ two-qubit channels as gates that implements $\mathcal{N}$.
Such quantum channels include, for example, noisy quantum circuits modeled in terms of perfect two-qubit unitary gates followed by single-qubit noise channels. The set of $n$-qubit channels of gate complexity (at most) $G$ is denoted by $\mathsf{CPTP}_{n, G}$. We abuse notation and also use $\mathsf{CPTP}_{n, G}$ to denote the class of $[0,1]$-valued functions on channel test operators that arises from $G$-gate channels as described in \Cref{subsection:prob-statement-online-learning-qchannels}.
We first bound the sequential covering numbers for (the $[0,1]$-valued function classes associated to) such channels.
\Cref{theorem:regret-bound-sequential-covering} then gives us a regret bound, from which we can derive a mistake bound via \Cref{lemma:regret-to-mistake-template}.

The (interior) covering number of the set $\mathsf{CPTP}_{n,G}$ of quantum channels with gate complexity at most $G$ is defined to be
\begin{equation}
    N(\mathsf{CPTP}_{n,G},\varepsilon,\norm{\cdot}_{\diamond})\coloneqq\inf\{\Abs{\mathsf{N}}:\mathsf{N}\subseteq\mathsf{CPTP}_{n,G},\;\forall\mathcal{N}\in\mathsf{CPTP}_{n,G}\;\exists\mathcal{M}\in\mathsf{N},\;\norm{\mathcal{N}-\mathcal{M}}_{\diamond}\leq\varepsilon\}.
\end{equation}
The fact that quantum channels of bounded gate complexity also have bounded complexity in the sense of covering numbers has previously been observed in Refs.~\cite{du2022efficient, caro2022generalization, HBC+21, zhao2023learning}. We recall this insight and its proof in the following lemma.

\begin{lemma}[Covering number bounds from gate complexity (see Refs.\ {\cite[Theorem C.2]{caro2022generalization}} and {\cite[Theorem 8]{zhao2023learning}})]\label{lemma:covering-number-gate-complexity}
    Let $\varepsilon\in (0,1)$. The (interior) $\varepsilon$-covering number of $\mathsf{CPTP}_{n, G}$ with respect to $\norm{\cdot}_\diamond$ is bounded as
    \begin{equation}
        N(\mathsf{CPTP}_{n, G},\varepsilon,\norm{\cdot}_\diamond)
        \leq \binom{n}{2}^G\left(\frac{6G}{\varepsilon}\right)^{512 G} \, .
    \end{equation}
\end{lemma}

\begin{proof}
    Let $\varepsilon'=\frac{\varepsilon}{G}$. From Ref.~\cite[Lemma~C.2]{caro2022generalization}, we have that the covering number for the set of two-qubit channels is
    \begin{equation}
        N(\mathsf{CPTP}_2,\varepsilon',\norm{\cdot}_{\diamond})\leq \left(\frac{6}{\varepsilon'}\right)^{512}=\left(\frac{6G}{\varepsilon}\right)^{512}.
    \end{equation}
    If we consider the set of all possible channels that act on two out of $n$ qubits, which we denote by $\mathsf{CPTP}_2^{(n)}$, then the covering number of this set is given by
    \begin{equation}
        N(\mathsf{CPTP}_2^{(n)},\varepsilon',\norm{\cdot}_{\diamond})\leq\binom{n}{2}\left(\frac{6G}{\varepsilon}\right)^{512},
    \end{equation}
    with the binomial factor $\binom{n}{2}$ coming from the fact that we allow the two-qubit channels to act on any pair of qubits.

    Let us now consider an arbitrary $\mathcal{N}\in\mathsf{CPTP}_{n,G}$. By definition, every such channel has the form $\mathcal{N}=\mathcal{N}_G\circ\mathcal{N}_{G-1}\circ\dotsb\circ\mathcal{N}_1$, where $\mathcal{N}_i\in\mathsf{CPTP}_2^{(n)}$ for all $i\in\{1,2,\dotsc,G\}$. By definition of $N(\mathsf{CPTP}_2^{(n)},\varepsilon',\norm{\cdot}_{\diamond})$, for every channel $\mathcal{N}_i$, we can find a corresponding $\widetilde{\mathcal{N}}_i$ in an $\varepsilon'$-covering net for $\mathsf{CPTP}_2^{(n)}$ such that $\norm{\mathcal{N}_i-\widetilde{\mathcal{N}_i}}_{\diamond}\leq\varepsilon'$. Then, the channel $\widetilde{\mathcal{N}}\coloneqq\widetilde{\mathcal{N}}_G\circ\widetilde{\mathcal{N}}_{G-1}\circ\dotsb\circ\widetilde{\mathcal{N}}_1\in\mathsf{CPTP}_{n,G}$ satisfies
    \begin{equation}
        \norm{\mathcal{N}-\widetilde{\mathcal{N}}}_{\diamond}\leq\sum_{i=1}^G \norm{\mathcal{N}_i-\widetilde{\mathcal{N}}_i}_{\diamond}\leq G\varepsilon'=\varepsilon,
    \end{equation}
    where we made use of the subadditivity-under-composition property of the diamond norm; see, e.g., Ref.~\cite[Proposition~3.48]{Wat18_book}. Therefore, if we let $\widetilde{\mathsf{N}}_{\varepsilon'}\subseteq\mathsf{CPTP}_2^{(n)}$ be an $\varepsilon'$-covering net for $\mathsf{CPTP}_2^{(n)}$, then the set
    \begin{equation}
        \mathsf{N}_{\varepsilon}\coloneqq\Big\{\mathcal{N}_G\circ\mathcal{N}_{G-1}\circ\dotsb\circ\mathcal{N}_1:\mathcal{N}_i\in\widetilde{\mathsf{N}}_{\varepsilon'}\Big\}
    \end{equation}
    is an $\varepsilon$-covering net for $\mathsf{CPTP}_{n,G}$. Noting that $|\mathsf{N}_{\varepsilon}|=|\widetilde{\mathsf{N}}_{\varepsilon'}|^G$, we obtain $N(\mathsf{CPTP}_{n,G},\varepsilon,\norm{\cdot}_{\diamond})\leq |\mathsf{N}_{\varepsilon}|\leq|\widetilde{\mathsf{N}}_{\varepsilon'}|^G$, for every $\varepsilon'$-covering net for $\mathsf{CPTP}_2^{(n)}$. As the covering number $N(\mathsf{CPTP}_{2}^{(n)},\varepsilon',\norm{\cdot}_{\diamond})$ is, by definition, the size of the smallest $\varepsilon'$-covering net for $\mathsf{CPTP}_2^{(n)}$, we can conclude that
    \begin{equation}
        N(\mathsf{CPTP}_{n,G},\varepsilon,\norm{\cdot}_{\diamond})
        \leq N(\mathsf{CPTP}_2^{(n)},\varepsilon',\norm{\cdot}_{\diamond})^G
        \leq\binom{n}{2}^G\left(\frac{6G}{\varepsilon}\right)^{512G},
    \end{equation}
    as required.
\end{proof}

We now observe that \Cref{lemma:covering-number-gate-complexity} immediately implies similar sequential covering number bounds. 
In fact, this can be seen by a reasoning analogous to how Ref.~\cite{caro2022generalization} went from covering w.r.t.~a norm on the level of the channel to empirical covering.

\begin{corollary}[Sequential covering number bounds from gate complexity]\label{corollary:sequential-covering-number-gate-complexity}
    Let $T\in\mathbb{N}_{\geq 1}$ and let $\varepsilon\in (0,1)$, $p\geq 1$.
    Then, 
    \begin{equation}
        N_{T}(\mathsf{CPTP}_{n, G},\varepsilon,p)
        \leq \binom{n}{2}^G\left(\frac{6G}{\varepsilon}\right)^{512 G} \, .
    \end{equation}
\end{corollary}
\begin{proof}
    Because of \Cref{lemma:covering-number-gate-complexity} and the monotonicity of sequential covering numbers w.r.t.~$p$, it suffices to show that $N_{\mathbf{z}}(\mathsf{CPTP}_{n, G},\varepsilon,\infty) \leq N(\mathsf{CPTP}_{n, G},\varepsilon,\norm{\cdot}_\diamond)$ holds for any complete rooted binary tree $\mathbf{z}$ of depth $T$. 
    This follows immediately from the fact that, if $\mathcal{N}$ and $\widetilde{\mathcal{N}}$ are two quantum channels, then
    \begin{align}
        \left\lvert \Tr[M_{R,B} \mathcal{N} (\rho_{R,A})] - \Tr[M_{R,B} \widetilde{\mathcal{N}} (\rho_{R,A})]\right\rvert
        &\leq \norm{M_{R,B}}_\infty \norm{(\mathcal{N} - \widetilde{\mathcal{N}})(\rho_{R,A})}_1\\
        \nonumber
        &\leq \norm{M_{R,B}}_\infty \norm{\mathcal{N} - \widetilde{\mathcal{N}}}_{\diamond} \norm{\rho_{R,A}}_1\\
        &\leq \norm{\mathcal{N} - \widetilde{\mathcal{N}}}_{\diamond}
        \nonumber
    \end{align}
    holds for any bipartite effect operator $M_{R,B}$ and for any bipartite state $\rho_{R,A}$.
\end{proof}

We can now plug this sequential covering number bound into \Cref{theorem:regret-bound-sequential-covering}. This leads to the following regret bound for online learning channels of bounded gate complexity:

\begin{theorem}[Regret bound for online learning channels of bounded complexity]\label{theorem:bounded-complexity-regret}
    Let $\ell:[0,1]\to\mathbb{R}$ be convex and $L$-Lipschitz.
    There exists an online learning strategy that, when presented sequentially with channel test operators $E_{A,B}^{(t)}$, $t\in\{1,2,\dotsc,T\}$, and associated loss functions $\ell_t (\cdot) = \ell ((\cdot) -b_t)$, outputs a sequence of hypothesis Choi matrices $N_{A,B}^{(t)}\in\mathsf{CPTP}_{n,G}^{\prime}$ whose regret is bounded as
    \begin{equation}
        \sum_{t=1}^T \ell_t(\Tr[E_{A,B}^{(t)}N_{A,B}^{(t)}]) - \min_{\mathcal{N}\in\mathsf{CPTP}_{n,G}}\sum_{t=1}^T \ell_t(\Tr[E_{A,B}^{(t)}C_{A,B}^{\mathcal{N}}])
        \leq \mathcal{O}\left(L\sqrt{T G\log(Gn)}\right)\, .
    \end{equation}
\end{theorem}
\begin{proof}
    Combining \Cref{theorem:regret-bound-sequential-covering,corollary:sequential-covering-number-gate-complexity}, we obtain the regret bounds from gate complexity
    as
    \begin{align}
    		&\sum_{t=1}^T \ell_t(\Tr[E_{A,B}^{(t)} N_{A,B}^{(t)}) - \min_{\mathcal{N}\in\mathsf{CPTP}_{n,G}}\sum_{t=1}^T \ell_t(\Tr[E_{A,B}^{(t)} C_{A,B}^{\mathcal{N}}])\\
            \nonumber
    		&\quad\leq 24 L \sqrt{512TG}\int_{0}^{1} \sqrt{\log\left(\frac{6G}{\beta}\right) + \log\binom{n}{2}}~\mathrm{d}\beta\\    	\nonumber	
    		&\quad\leq 24 L \sqrt{512TG}\left( \int_{0}^{1} \sqrt{\log\left(6G\right) + \log\left(\frac{1}{\beta}\right)}~\mathrm{d}\beta + \sqrt{2\log(n)}\right)\\
      \nonumber
    		&\quad\leq 24 L \sqrt{512TG}\left(  \sqrt{\log(6G)} + \int_{0}^{1} \sqrt{ \log\left(\frac{1}{\beta}\right)}~\mathrm{d}\beta + \sqrt{2\log(n)}\right)\\
      \nonumber
    		&\quad\leq 24 L \sqrt{512TG}\left( \sqrt{\log(6G)} + \frac{\sqrt{\pi}}{2} + \sqrt{2\log(n)}\right)\\
      \nonumber
    		&\quad\leq \mathcal{O}\left(L\sqrt{T G\log(Gn)}\right)\, ,
      \nonumber
    \end{align}
    where we have used the inequalities $\sqrt{a+b}\leq \sqrt{a} + \sqrt{b}\leq \sqrt{2(a+b)}$ for $a,b\geq 0$ and the integral identity $\int\sqrt{\log(1/x)}~\mathrm{d}x = x\sqrt{\log(1/x)} - (\sqrt{\pi}/2) \cdot \operatorname{erf}(\sqrt{\log(1/x)})$, with the error function given as $\operatorname{erf}(x)=\frac{2}{\sqrt{\pi}}\int_{0}^{x}\exp(-t^2)~\mathrm{d}t$.
\end{proof}

Finally, to conclude our discussion of online learning channels with gate complexity $G$, we can combine \Cref{theorem:bounded-complexity-regret} with \Cref{lemma:regret-to-mistake-template} to obtain the following mistake bound.

\begin{corollary}[Mistake bound for online learning channels of bounded complexity]\label{corollary:bounded-complexity-mistake}
    Let $\varepsilon\in (0,1)$.
    There exists an online learning strategy that, in a realizable setting for $\mathsf{CPTP}_{n,G}$, when presented sequentially with channel test operators $E_{A,B}^{(t)}$, $t\in\{1,2,\dotsc,T\}$, and associated loss functions $\ell_t (\cdot) = \ell ((\cdot) -b_t)$, for some $L$-Lipschitz $\ell$, outputs a sequence $N_{A,B}^{(t)}\in\mathsf{CPTP}_{n,G}^{\prime}$, of hypothesis Choi matrices that makes at most $\mathcal{O}\left(\frac{L^2 G\log(Gn)}{\varepsilon^2}\right)$ many $\varepsilon$-mistakes.
\end{corollary}
\begin{proof}
    Thanks to \Cref{theorem:bounded-complexity-regret}, we can apply \Cref{lemma:regret-to-mistake-template} with $h_1(\varepsilon,T)=CL\sqrt{T G\log(Gn)}$ for a suitable constant $C>0$ and with $h_2(\varepsilon)=0$. With these choices, 
    \begin{equation}
        T^\ast 
        = T^\ast (h_1,h_2,\varepsilon)
        =\left(\frac{3CL\sqrt{G\log(Gn)}}{2\varepsilon}\right)^2\leq\mathcal{O}\left(\frac{L^2 G\log(Gn)}{\varepsilon^2}\right). 
    \end{equation}
    So, the mistake bound obtained from \Cref{lemma:regret-to-mistake-template} in this case is exactly as claimed in the statement of the corollary.
\end{proof}

Together, \Cref{theorem:bounded-complexity-regret} and \Cref{corollary:bounded-complexity-mistake} establish \Cref{informal-theorem:bounded-complexity-online-learning}. 
In particular, this implies: For the physically relevant class of channels implementable with polynomial-size circuits, we can solve the online learning task with only polynomially many mistakes.

\subsection{Regret bound for mixtures of known channels}\label{section:online-learning-mixtures}

In the previous section, it was shown that online learning quantum channels of bounded gate complexity is possible with good regret and number of mistakes. Here we show that even if the channel has unbounded (exponentially many in the number of qubits) gates that act on the input state, regret- and mistake-bounded online learning is still possible if we know the gates but not the probability with which they act. Even more generally, we show that any channel composed of mixture of known channels is efficiently online learnable, even if the mixture is over exponentially many known channels, which could be arbitrary quantum channels. A notable example of such channels are mixed unitary channels (with known unitaries), for example Pauli channels. Since the Pauli channel framework is better understood and well-known, we first prove regret upper bounds in this setting for clearer exposition. Later we generalize this to mixtures of general known channels.

\begin{theorem}[Regret bound for online learning Pauli channels]\label{theorem:pauli-regret}
    Let $\ell:[0,1]\to\mathbb{R}$ 
    be convex and $L$-Lipschitz. There exists an online learning strategy that, when presented sequentially with channel test operators $E_{A,B}^{(t)}$, $t\in\{1,2,\dotsc,T\}$, and associated losses $\ell_t (\cdot) = \ell ((\cdot) -b_t)$, outputs Pauli channel Choi matrix hypotheses $N_{A,B}^{(t)} \in \mathsf{PAULI}_n^{\prime}$ whose regret is bounded as
    \begin{equation}\label{eq-regret-Pauli-channels}
        \sum_{t=1}^T \ell_t(\Tr[E_{A,B}^{(t)}N_{A,B}^{(t)}]) - \sum_{t=1}^T \ell_t(\Tr[E_{A,B}^{(t)}C_{A,B}^{\mathcal{P}}]) = \mathcal{O}(L\sqrt{nT}),
    \end{equation}
   for every Pauli channel $\mathcal{P}_{A\to B}\in\mathsf{PAULI}_n$.
\end{theorem}

\begin{proof}
The result follows by applying Lemma~\ref{lem-regret_via_pseudoregret} and applying the multiplicative weights update (MWU) algorithm (Algorithm~\ref{alg:mwu_general}) with a particular choice of the loss vectors $\Vec{m}^{(t)}$. Specifically, we let
\begin{align}
    \Vec{m}^{(t)}&\coloneqq(m_{\Vec{z},\Vec{x}}^{(t)})_{\Vec{z},\Vec{x}\in\{0,1\}^n},\\
    m_{\Vec{z},\Vec{x}}^{(t)}&\coloneqq\frac{1}{L}\ell_t'(\Tr[E_{A,B}^{(t)}N_{A,B}^{(t)}])\Tr[E_{A,B}^{(t)}\Gamma_{A,B}^{\Vec{z},\Vec{x}}],\quad\forall~\Vec{z},\Vec{x}\in\{0,1\}^n,
\end{align}
where
\begin{equation}
    N_{A,B}^{(t)}\coloneqq\sum_{\Vec{z},\Vec{x}\in\{0,1\}^n}p_{\Vec{z},\Vec{x}}^{(t)}\Gamma_{A,B}^{\Vec{z},\Vec{x}},
\end{equation}
with the probability vectors $\Vec{p}^{(t)}\coloneqq(p_{\Vec{z},\Vec{x}}^{(t)})_{\Vec{z},\Vec{x}\in\{0,1\}^n}$ defined according to the MWU algorithm. Let us verify that $m_{\Vec{z},\Vec{x}}^{(t)}\in[-1,1]$ for all $\Vec{z},\Vec{x}\in\{0,1\}^n$, as required by the MWU algorithm. Indeed, we readily have that $\frac{1}{L}\ell_t'(\Tr[E_{A,B}^{(t)}N_{A,B}^{(t)}])\in[-1,1]$ because $\ell$ is assumed to be $L$-Lipschitz. In addition, we have that
\begin{equation}
    0
    \leq \Tr[E_{A,B}^{(t)}\Gamma_{A,B}^{\Vec{z}, \Vec{x}}] 
    \leq \Tr[(\sigma_A \otimes \mathbbm{1}_B)\Gamma_{A,B}^{\Vec{z}, \Vec{x}}] 
    = \Tr_A[\sigma_A \Tr_{B}[\Gamma_{A,B}^{\Vec{z}, \Vec{x}}]]
    = \Tr[\sigma_A]
    = 1 ,
\end{equation}
where we have used the fact that, because $E_{A,B}^{(t)}$ is a channel test operator, there exists a density operator $\sigma_A$ such that $E_{A,B}^{(t)}\leq\sigma_A\otimes\mathbbm{1}_B$. Therefore, $m_{\Vec{z},\Vec{x}}^{(t)}\in[-1,1]$ for all $\Vec{z},\Vec{x}\in\{0,1\}^n$.

Now, for every Pauli channel $\mathcal{P}_{A\to B}$ with associated error-rate vector $\Vec{q}=(q_{\Vec{z},\Vec{x}})_{\Vec{z},\Vec{x}\in\{0,1\}^n}$, we can use \Cref{eq-Pauli_channel_Choi_rep} to see that
\begin{align}
    \frac{1}{L}\ell'_t(\Tr[E_{A,B}^{(t)}N_{A,B}^{(t)}])\Tr[E_{A,B}^{(t)}C_{A,B}^{\mathcal{P}}]&=\sum_{\Vec{z},\Vec{x}\in\{0,1\}^n}q_{\Vec{z},\Vec{x}}\frac{1}{L}\ell_t'(\Tr[E_{A,B}N_{A,B}^{(t)}])\Tr[E_{A,B}^{(t)}\Gamma_{A,B}^{\Vec{z},\Vec{x}}]\\
    \nonumber
    &=\sum_{\Vec{z},\Vec{x}\in\{0,1\}^n} m_{\Vec{z},\Vec{x}}^{(t)}q_{\Vec{z},\Vec{x}}\\
    &=\Vec{m}^{(t)}\cdot\Vec{q},
    \nonumber
\end{align}
and similarly, we find that
\begin{equation}
    \frac{1}{L}\ell'_t(\Tr[E_{A,B}^{(t)}N_{A,B}^{(t)}])\Tr[E_{A,B}^{(t)}N_{A,B}^{(t)}]=\Vec{m}^{(t)}\cdot\Vec{p}^{(t)}.    
\end{equation}
Therefore, it follows from the known regret bound in Theorem~\ref{thm:MWU} on the MWU algorithm that
\begin{align}
    &\sum_{t=1}^T \frac{1}{L}\ell'_t(\Tr[E_{A,B}^{(t)}N_{A,B}^{(t)}])\Tr[E_{A,B}^{(t)}N_{A,B}^{(t)}] - \sum_{t=1}^T \frac{1}{L}\ell'_t(\Tr[E_{A,B}^{(t)}N_{A,B}^{(t)}])\Tr[E_{A,B}^{(t)}C_{A,B}^{\mathcal{P}}]\nonumber\\
    &\qquad=\sum_{t=1}^T\Vec{m}^{(t)}\cdot\Vec{p}^{(t)}-\sum_{t=1}^T\Vec{m}^{(t)}\cdot\Vec{q}\\
    &\qquad\leq\eta T+\frac{n\log 4}{\eta}.
\end{align}
Finally, combining this inequality with the result of Lemma~\ref{lem-regret_via_pseudoregret}, we obtain 
\begin{align}
  &\sum_{t=1}^T \left(\ell_t(\Tr[E_{A,B}^{(t)}N_{A,B}^{(t)}]) - \ell_t(\Tr[E_{A,B}^{(t)}C_{A,B}^{\mathcal{P}}])\right) \nonumber \\
  &\qquad\leq  \sum_{t=1}^T \left(\ell'_t(\Tr[E_{A,B}^{(t)}N_{A,B}^{(t)}]) (\Tr[E_{A,B}^{(t)}N_{A,B}^{(t)}] - \Tr[E_{A,B}^{(t)}C_{A,B}^{\mathcal{P}}]) \right) \label{eq:regret-loss-1} \\
  &\qquad\leq L \left(\eta T + \frac{n\log 4}{\eta}\right), \label{eq:regret-loss-3}
\end{align}
for every Pauli channel $\mathcal{P}_{A\to B}\in\mathsf{PAULI}_n$. The claimed bound then follows by setting $\eta = \sqrt{n/T}$.
\end{proof}

While Theorem~\ref{theorem:pauli-regret} focused solely on Pauli channels, as we show below, it readily translates to \emph{any} convex combination of \emph{known} channels---even exponentially many known channels that may have arbitrarily large gate complexity. This generalization captures a wide class of channels of interest, such as mixed unitary channels, in which each known channel is a unitary. Pauli channels are themselves a special case of such channels, because they are convex combinations of the $4^n$ unitary channels defined by the Pauli strings.

\begin{corollary}[Regret bound for convex combinations of known channels]\label{corollary:regret-mixture}
    Let $K \in \mathbb Z_{> 0}$, and let $\{\mathcal{N}_j\}_{j=1}^K$ be a set of $K$ known $n$-qubit channels $\mathcal{N}_j\in\mathsf{CPTP}_n$. Let $\ell:[0,1]\to\mathbb{R}$ be convex and $L$-Lipschitz. There exists an online learning strategy that, when presented sequentially with channel test operators $E_{A,B}^{(t)}$, $t\in\{1,2,\dotsc,T\}$, and associated losses $\ell_t (\cdot) = \ell ((\cdot) -b_t)$, outputs Choi matrix hypotheses $N_{A,B}^{(t)}$ of the form
    \begin{equation}
        N_{A,B}^{(t)}=\sum_{j=1}^K p_j^{(t)}C(\mathcal{N}_j),
    \end{equation}
    with $\Vec{p}^{(t)}\coloneqq(p_1^{(t)},p_2^{(t)},\dotsc,p_K^{(t)})\in\Delta_K$, whose regret is bounded as
    \begin{equation}
        \sum_{t=1}^T \ell_t(\Tr[E_{A,B}^{(t)}N_{A,B}^{(t)}]) - \sum_{t=1}^T \ell_t(\Tr[E_{A,B}^{(t)}C_{A,B}^{\mathcal{N}}]) = \mathcal{O}\left(L\sqrt{T \log K}\right)\, ,
    \end{equation}
    for every $\mathcal{N}\in\text{conv}(\{\mathcal{N}_j\}_{j=1}^K)$.
\end{corollary}

\begin{proof}
    The proof is analogous to that of Theorem~\ref{theorem:pauli-regret}. In particular, we combine Lemma~\ref{lem-regret_via_pseudoregret} with the MWU algorithm (Algorithm~\ref{alg:mwu_general}) applied to a particular choice of loss vectors $\Vec{m}^{(t)}$. We let
    \begin{align}
        \Vec{m}^{(t)}&\coloneqq(m_j^{(t)})_{j\in\{1,2,\dotsc,K\}},\\
        m_j^{(t)}&\coloneqq\frac{1}{L}\ell'_t(\Tr[E_{A,B}^{(t)}N_{A,B}^{(t)}])\Tr[E_{A,B}^{(t)}C_{A,B}^{\mathcal{N}_j}]\quad\forall~j\in\{1,2,\dotsc,K\},
    \end{align}
    where
    \begin{equation}
        N_{A,B}^{(t)}\coloneqq\sum_{j=1}^K p_j^{(t)}C_{A,B}^{\mathcal{N}_j},
    \end{equation}
    with the probability vectors $\Vec{p}^{(t)}\coloneqq(p_j^{(t)})_{j\in\{1,2,\dotsc,K\}}$ defined according to the MWU algorithm. As in the proof of Theorem~\ref{theorem:pauli-regret}, it is straightforward to verify that $m_j^{(t)}\in[-1,1]$ for all $j\in\{1,2,\dotsc,K\}$, as required by the MWU algorithm.
    
    Now, consider an arbitrary channel $\mathcal{N}$ in the convex hull of the channels $\mathcal{N}_j$, specified as $\mathcal{N}=\sum_{j=1}^k q_j\mathcal{N}_j$ for some probability vector $\Vec{q}=(q_1,q_2,\dotsc,q_K)$. By linearity of the Choi representation, and using Theorem~\ref{thm:MWU}, we obtain
    \begin{align}
        &\sum_{t=1}^T \frac{1}{L}\ell'_t(\Tr[E_{A,B}^{(t)}N_{A,B}^{(t)}])\Tr[E_{A,B}^{(t)} N_{A,B}^{(t)}] - \sum_{t=1}^T \frac{1}{L}\ell'_t(\Tr[E_{A,B}^{(t)}N_{A,B}^{(t)}])\Tr[E_{A,B}^{(t)} C_{A,B}^{\mathcal{N}}]\\
         \nonumber
         &\qquad=\sum_{t=1}^T\Vec{m}^{(t)}\cdot\Vec{p}^{(t)}-\sum_{t=1}^T\Vec{m}^{(t)}\cdot\Vec{q}\\
          \nonumber
        &\qquad\leq\eta T + \frac{\log K}{\eta}.
    \end{align}
    Therefore, by Lemma~\ref{lem-regret_via_pseudoregret}, we obtain
    \begin{align}
      &\sum_{t=1}^T \left(\ell_t(\Tr[E_{A,B}^{(t)}N_{A,B}^{(t)}]) - \ell_t(\Tr[E_{A,B}^{(t)}C_{A,B}^{\mathcal{N}}])\right) \\
       \nonumber
      &\qquad\leq  \sum_{t=1}^T \left(\ell'_t(\Tr[E_{A,B}^{(t)}N_{A,B}^{(t)}]) (\Tr[E_{A,B}^{(t)}N_{A,B}^{(t)}] - \Tr[E_{A,B}^{(t)}C_{A,B}^{\mathcal{N}}]) \right) \\
       \nonumber
      &\qquad\leq L \left(\eta T + \frac{\log K}{\eta}\right).
    \end{align}
    Setting $\eta=\sqrt{\log K/T}$, we obtain the desired regret bound.
\end{proof}

\begin{corollary}[Mistake bounds for Pauli channels and convex combination of known channels]\label{corollary:pauli-mistake-bound}
    Let $\varepsilon\in (0,1)$.
    Let $K\in\mathbb{Z}_{>0}$, and let $\{\mathcal{N}_j\}_{j=1}^K$ be a set of $K$ known $n$-qubit channels $\mathcal{N}_j\in\mathsf{CPTP}_n$. Let $\mathcal{N}_{A\to B}\in\text{conv}(\{\mathcal{N}_j\}_{j=1}^K)$ be a fixed channel unknown to the learner. 
    There exists an online learning strategy that, in a realizable setting for $\text{conv}(\{\mathcal{N}_j\}_{j=1}^K)$, when presented sequentially with two-outcome channel test operators $E_{A,B}^{(t)}$, $t\in\{1,2,\dotsc,T\}$, and associated losses $\ell_t (\cdot) = \ell ((\cdot) -b_t)$, for some $L$-Lipschitz $\ell$, outputs a sequence $N_{A,B}^{(t)} \in \text{conv}(\{C_{A,B}^{\mathcal{N}_j}\}_{j=1}^K)$ of hypothesis Choi matrices that makes at most $\mathcal{O}(\frac{L^2 \log(K)}{\varepsilon^2})$ many $\varepsilon$-mistakes. 
\end{corollary}

\begin{proof}
    Mimicking the proof of \Cref{corollary:bounded-complexity-mistake} by invoking \Cref{lemma:regret-to-mistake-template}, but now for the choice $h_1(\varepsilon, T) = CL\sqrt{T\log K}$ for a suitable $C > 0$ and $h_2(\varepsilon) = 0$, we get that $T^{\ast} = \left( \frac{3CL \sqrt{\log K}}{2\varepsilon}\right)^2 \leq \mathcal{O}(\frac{L^2\log (K)}{\varepsilon^2})$.
\end{proof}

\subsection{Regret bounds for multi-time processes}\label{sec-online_learning_multi_time_processes}

In this section, we generalize the developments of Sections~\ref{section:online-learning-bounded-complexity} and \ref{section:online-learning-mixtures} to multi-time processes. We start by presenting the problem of online learning classes of multi-time processes, by casting it in terms of the general framework of online learning laid out in Section~\ref{sec-online_learning_basics}. With that, we formally state the problem in Problem~\ref{problem:online_learning_classes_multi_time_processes} below.

Let $r\in\mathbb{N}$. We fix finite-dimensional Hilbert spaces $\mathcal{H}_{A_1},\mathcal{H}_{B_1},\dotsc,\mathcal{H}_{A_r},\mathcal{H}_{B_r}$, and we let $\mathcal{H}_{A,B}^{(r)}\equiv\mathcal{H}_{A_1}\otimes\mathcal{H}_{B_1}\otimes\mathcal{H}_{A_2}\otimes\mathcal{H}_{B_2}\otimes\dotsb\otimes\mathcal{H}_{A_r}\otimes\mathcal{H}_{B_r}$. Then, the input set/domain $\mathsf{X}$ comprises multi-time test operators corresponding to two-outcome multi-time tests
\begin{equation}
    \mathsf{X}\coloneqq\Big\{E\in\Lin(\mathcal{H}_{A,B}^{(r)}):E\geq 0,\,\norm{E}_{\diamond r}^{\ast}\leq 1\Big\}.
\end{equation}
The output set is $\mathsf{Y}=[0,1]$, and $\mathcal{F}$ is defined via a subclass $\mathsf{C}\subseteq\mathsf{COMB}_r$ of interest, such that for every $N\in\mathsf{C}$ we define the function $f_N: \mathsf{X}\to\mathsf{Y}$ as $f_N(E)=\Tr[EN]$ for all $E\in\mathsf{X}$. In other words,
\begin{equation}
    \mathcal{F}=\Big\{f_N:\mathsf{X}\to[0,1]:N\in\mathsf{C},\,f_N(E)=\Tr[EN]~\forall\,E\in\mathsf{X}\Big\}.
\end{equation}

As in the case of channels, online learning proceeds is an interactive procedure with an adversary, who provides test operators $E^{(t)}\in\Lin(\mathcal{H}_{A,B}^{(r)})$ for two-output multi-time tests to the learner, and the learner outputs hypotheses $N^{(t)}\in\mathsf{C}$. The goal of the learner is to achieve a small regret,
\begin{equation}\label{eq-multi_time_process_regret}
    R_T=\sum_{t=1}^T\ell_t(\Tr[E^{(t)}N^{(t)}])-\min_{N^{\ast}\in\mathsf{C}}\sum_{t=1}^T\ell_t(\Tr[E^{(t)}N^{\ast}]),
\end{equation}
where the losses $\ell_t$ are revealed by the adversary. We aim for regret bounds scaling as $\mathcal{O}(\sqrt{T\cdot\text{poly}(n)})$. On the other hand, in the realizable scenario, where there exists a (to the learner unknown) comb operator $N^{\ast}\in\mathsf{C}$ such that all losses take the form $\ell_t(\cdot)=\ell((\cdot)-b_t)$, where $\ell$ is some fixed function and each $b_t$ satisfies $|b_t-\Tr[E^{(t)}N^{\ast}]|\leq\varepsilon/3$, the learner aims to achieve a small number of $\varepsilon$-mistakes. Here, our goal is to guarantee mistake bounds scaling as $\mathcal{O}(\text{poly}(n,\varepsilon^{-1}))$. We summarize the formulation of our multi-time process online learning problem as follows.

\begin{problem}[Online learning of multi-time quantum processes]\label{problem:online_learning_classes_multi_time_processes}
    Consider a subset $\mathsf{C}\subseteq\mathsf{COMB}_r$, for $r\in\mathbb{N}$, and let $N^{\ast}\in\mathsf{C}$ be unknown. Given a sequence of $T\in\mathbb{N}$ interactive rounds, in which test operators $E^{(1)},E^{(2)},\dotsc,E^{(T)}\in\mathsf{X}$ are presented sequentially by an adversary, the problem is to output a sequence $N^{(1)},N^{(2)},\dotsc,N^{(T)}\in\mathsf{C}$ of comb operators such that for losses $\ell_t$ of the form $\ell_t(\cdot)=\ell((\cdot)-b_t)$, where $\ell$ is some fixed function and each $b_t$ satisfies $|b_t-\Tr[E^{(t)}N^{\ast}]|\leq\varepsilon/3$, the regret in \eqref{eq-multi_time_process_regret} scales as $\mathcal{O}(\sqrt{T\cdot\text{poly}(n)})$ and the number of $\varepsilon$-mistakes scales as $\mathcal{O}(\text{poly}(n,\varepsilon^{-1}))$.
\end{problem}

We note that for $r=1$, Problem~\ref{problem:online_learning_classes_multi_time_processes} is equivalent to Problem~\ref{problem:online-learning-classes-of-channels}. We also note that the projected MMW algorithm for Choi states of quantum channels (Algorithm~\ref{alg:projected-mmw-choi-state}) generalizes straightforwardly to Choi states of multi-time processes; see Algorithm~\ref{alg:projected-mmw-choi-state_multi_time}.

\begin{algorithm}
     \begin{algorithmic}[1]
     \Require {$\eta \in (0, 1)$; Initialize $W^{(1)} = \bigotimes_{k=1}^r\mathbbm{1}_{A_k,B_k}$ and $P^{(1)}=\frac{1}{d}W^{(1)}$, $d\equiv\prod_{k=1}^r d_{A_k}d_{B_k}$}
     \For{\texttt{$t = 1, 2, \ldots, T$}}
        \State Output the decision/estimate $N^{(t)}$
        \State Receive the cost matrix $L^{(t)}$, $\norm{L}_{\diamond r}^{\ast}\leq 1$.
        \State Update the weight matrix: $W^{(t+1)}=\exp[\log(W^{(t)})-\eta L^{(t)}]$, $\omega^{(t+1)}=W^{(t+1)}/\Tr[W^{(t+1)}]$.
        \State Project: $P^{(t+1)}\coloneqq\frac{1}{d_A}\argmin_{P}\Big\{D(\frac{1}{d_A}P\Vert \omega^{(t+1)}):P\in\mathsf{COMB}_r\Big\}$, $d_A=\prod_{k=1}^r d_{A_k}$.
        \EndFor
     \end{algorithmic}
     \caption{Projected MMW algorithm for Choi states of $r$-step multi-time processes}\label{alg:projected-mmw-choi-state_multi_time}
\end{algorithm}

We provide an analysis of Algorithm~\ref{alg:projected-mmw-choi-state_multi_time} in Appendix~\ref{sec-projected_MMW}; see Remark~\ref{rem:projected_MMW_choi_state_multi_time} therein.

\begin{figure}
    \centering
    \includegraphics[width=0.85\textwidth]{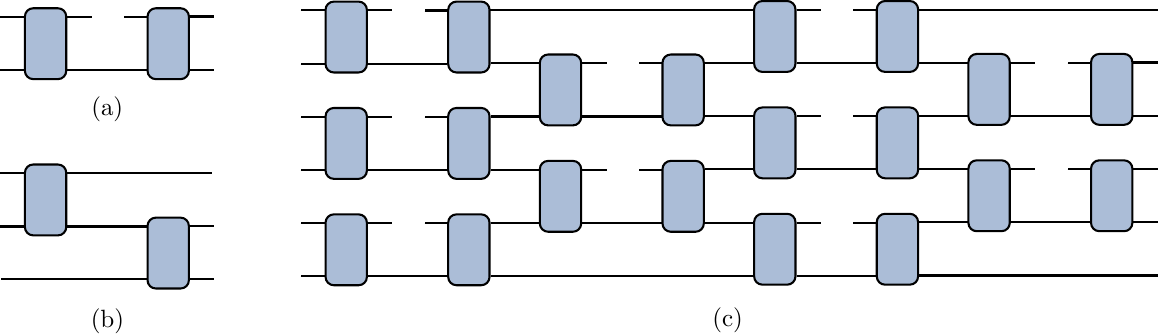}
    \caption{\textbf{Multi-time processes with bounded complexity.} (a) The basic unit of our multi-time processes with bounded complexity is a process consisting of two two-qubit channels connected by an inaccessible memory system. (b) By collapsing the causal structure of the inputs and outputs of the process in (a), we obtain a three-qubit channel belonging to the set $\mathsf{CPTP}_{3,2}$. (c) An example of a multi-time process obtained by composing (ten of) the basic elements in~(a) in a circuit.}
    \label{fig:G_gate_comb}
\end{figure}

\paragraph*{Multi-time processes with bounded gate complexity.} We can extend the results above to the case of multi-time processes with bounded complexity. The generalized form of multi-time processes allows for many possibilities for how to define multi-time processes with bounded complexity. Just as we defined $n$-qubit quantum channels with gate complexity $G$ as being composed of at most $G$ two-qubit quantum channels, here we consider multi-time processes composed of the basic ``unit'' shown in Figure~\ref{fig:G_gate_comb}(a). Specifically, we define $\mathsf{COMB}_{n,G}$ to be the set of all comb operators corresponding to multi-time processes on $n$ qubits that can be implemented by the composition of at most $G$ of the basic units in Figure~\ref{fig:G_gate_comb}(a), in the manner of a circuit as shown in Figure~\ref{fig:G_gate_comb}(c). The number of time steps, $r$, in the multi-time process depends on $G$.

First, we prove an analogue of \Cref{lemma:covering-number-gate-complexity} for multi-time processes with gate complexity $G$. While the diamond norm was the natural distance measure in the case of channels, here we use strategy norms instead (see \Cref{sec-multi_time_processes} for a definition and a discussion of their properties).

\begin{lemma}[Covering number bounds from gate complexity for multi-time processes]
    Let $\varepsilon\in (0,1)$.
    The (interior) $\varepsilon$-covering number of $\mathsf{COMB}_{n,G}$ with respect to $\norm{\cdot}_{\diamond r}$ is bounded as
    \begin{equation}
        N(\mathsf{COMB}_{n, G},\varepsilon,\norm{\cdot}_{\diamond r})
        \leq \binom{n}{2}^G\left(\frac{6G}{\varepsilon}\right)^{1024 G} \, .
    \end{equation}
\end{lemma}
\begin{proof}
    The proof again follows similar ideas as employed in the covering number bounds of \cite{caro2022generalization, zhao2023learning}.
    First, we recall (compare, e.g., Ref.~\cite{vershynin2018highdimensional}) a well-known fact about covering numbers of norm balls in $\mathbb{R}^K$: If $R>0$, $\varepsilon\in(0,R]$, and if $B_R^{\norm{\cdot}} (x)$ denotes the $\norm{\cdot}$-ball of radius $R$ around $x\in\mathbb{R}^K$ for some norm $\norm{\cdot}$, then 
    \begin{equation}
        N(B_R^{\norm{\cdot}} (x), \varepsilon, \norm{\cdot})
        \leq \left(1 + \frac{2R}{\varepsilon}\right)^K
        \leq \left(\frac{3R}{\varepsilon}\right)^K\, .
    \end{equation}
    To apply this in our scenario, notice that any basic unit as in Figure~\ref{fig:G_gate_comb}(a) lives in the $\norm{\cdot}_{\diamond 2}$-unit-ball in a $((2\cdot (2^4\times 2^4))=512)$-dimensional complex space, where the ambient dimension is that of two $2$-qubit channels.
    Consequently, via the approximate monotonicity of covering numbers w.r.t.~inclusion of sets, the above standard bound implies for our case that $\mathsf{COMB}_2$, the class of basic units, admits 
    \begin{equation}
        N(\mathsf{COMB}_{2},\varepsilon,\norm{\cdot}_{\diamond 2})
        \leq N\left(B_1^{\norm{\cdot}_{\diamond 2}} (0), \frac{\varepsilon}{2}, \norm{\cdot}_{\diamond 2}\right)
        \leq \left(\frac{6}{\varepsilon}\right)^{1024}\, 
    \end{equation}
    as a covering number bound.
    If we let $\mathsf{COMB}_2^{(n)}$ be the set of all basic units in Figure~\ref{fig:G_gate_comb}(a) that act on two out of $n$ qubits, then the covering number of this set is bounded from above as follows:
    \begin{equation}
        N(\mathsf{COMB}_2^{(n)},\varepsilon,\norm{\cdot}_{\diamond 2})\leq\binom{n}{2}\left(\frac{6}{\varepsilon}\right)^{1024},
    \end{equation}
    with the binomial factor $\binom{n}{2}$ coming from the fact that we allow the basic units to act on any pair of qubits.

    Now, consider an arbitrary $N\in\mathsf{COMB}_{n,G}$. By definition, every such comb operator has the form $N=N_1\star N_2\star\dotsb\star N_G$, where $N_i\in\mathsf{COMB}_2^{(n)}$. Let $\varepsilon'=\frac{\varepsilon}{G}$. By definition of the covering number $N(\mathsf{COMB}_2^{(n)},\varepsilon',\norm{\cdot}_{\diamond 2})$, for every $N_i\in\mathsf{COMB}_2^{(n)}$, we can find a corresponding $\widetilde{N}_i$ in an $\varepsilon'$-covering net for $\mathsf{COMB}_2^{(n)}$ such that $\norm{N_i-\widetilde{N}_i}_{\diamond 2}\leq\varepsilon'$. Then, if we let $\widetilde{N}\coloneqq\widetilde{N}_1\star\widetilde{N}_2\star\dotsb\star\widetilde{N}_G$, and by making use of the subadditivity-under-composition property of the strategy norm, as shown in \Cref{cor:strategy_norm_subadditivity_composition}, we obtain
    \begin{equation}
        \norm{N-\widetilde{N}}_{\diamond r}\leq\sum_{i=1}^G\norm{N_i-\widetilde{N}_i}_{\diamond 2}\leq G\varepsilon'=\varepsilon.
    \end{equation}
    Therefore, if we let $\widetilde{\mathsf{N}}_{\varepsilon'}\subseteq\mathsf{COMB}_2^{(n)}$ be an $\varepsilon'$-covering net for $\mathsf{COMB}_2^{(n)}$, then the set
    \begin{equation}
        \mathsf{N}_{\varepsilon}\coloneqq\Big\{\widetilde{N}_1\star \widetilde{N}_2\star\dotsb\star \widetilde{N}_G : N_i\in\widetilde{\mathsf{N}}_{\varepsilon'}\Big\}
    \end{equation}
    is an $\varepsilon$-covering net for $\mathsf{COMB}_{n,G}$. Noting that $|\mathsf{N}_{\varepsilon}|=|\widetilde{\mathsf{N}}_{\varepsilon'}|^G$, we obtain $N(\mathsf{COMB}_{n,G},\varepsilon,\norm{\cdot}_{\diamond r})\leq|\mathsf{N}_{\varepsilon}|\leq|\widetilde{\mathsf{N}}_{\varepsilon'}|^G$ for every $\varepsilon'$-covering net for $\mathsf{COMB}_2^{(n)}$. As the covering number $N(\mathsf{COMB}_2^{(n)},\varepsilon',\norm{\cdot}_{\diamond 2})$ is, by definition, the size of the smallest $\varepsilon'$-covering net for $\mathsf{COMB}_2^{(n)}$, we can conclude that 
    \begin{equation}
        N(\mathsf{COMB}_{n,G},\varepsilon,\norm{\cdot}_{\diamond r})\leq N(\mathsf{COMB}_2^{(n)},\varepsilon',\norm{\cdot}_{\diamond 2})^G\leq\binom{n}{2}^G\left(\frac{6G}{\varepsilon}\right)^{1024G},
    \end{equation}
    as required.
\end{proof}

\begin{theorem}[Regret bound for online learning multi-time processes of bounded complexity]\label{theorem:bounded-complexity-regret_multi_time}
    Let $\ell:[0,1]\to\mathbb{R}$ be convex and $L$-Lipschitz.
    There exists an online learning strategy that, when presented sequentially with multi-time test operators $E^{(t)}$ for $r$ time steps, $t\in\{1,2,\dotsc,T\}$, and associated loss functions $\ell_t (\cdot) = \ell ((\cdot) -b_t)$, outputs a sequence of hypothesis comb operators $N_{A,B}^{(t)}\in\mathsf{COMB}_{n,G}^{\prime}$, corresponding to $r$-step multi-time quantum processes, whose regret is bounded as
    \begin{equation}
        \sum_{t=1}^T \ell_t(\Tr[E^{(t)}N^{(t)}]) - \min_{N\in\mathsf{COMB}_{n,G}}\sum_{t=1}^T \ell_t(\Tr[E^{(t)}N])
        \leq \mathcal{O}\left(L\sqrt{T G\log(Gn)}\right).
    \end{equation}
\end{theorem}

\begin{proof}
    The proof is analogous to the proof of \Cref{theorem:bounded-complexity-regret}. First, it holds that
    \begin{equation}\label{eq-bounded-complexity-regret_multi_time_pf1}
        N_T(\mathsf{COMB}_{n,G},\varepsilon,p)\leq\binom{n}{2}^G \left(\frac{6G}{\varepsilon}\right)^{1024G},
    \end{equation}
    for $T\in\mathbb{N}_{\geq 1}$, $\varepsilon\in(0,1)$, and $p\geq 1$. This holds due to the fact that $N_{\mathbf{z}}(\mathsf{COMB}_{n,G},\varepsilon,p)\leq N_{\mathbf{z}}(\mathsf{COMB}_{n,G},\varepsilon,\infty)\leq N(\mathsf{COMB}_{n,G},\varepsilon,\norm{\cdot}_{\diamond r})$, where $\mathbf{z}$ is an arbitrary complete rooted binary tree of depth $T$. The first of these inequalities is due to monotonicity of the sequential covering numbers with respect to $p$, and the second follows from the fact that
    \begin{equation}
        \Abs{\Tr[EN]-\Tr[E\widetilde{N}]}\leq\norm{E}_{\diamond r}^{\ast}\norm{N-\widetilde{N}}_{\diamond r}\leq \norm{N-\widetilde{N}}_{\diamond r},
    \end{equation}
    for arbitrary multi-time test operators $E$ and arbitrary comb operators $N$ and $\widetilde{N}$, which holds because of the H\"{o}lder inequality for strategy norms in \eqref{eq-Holder_inequality_strategy_norm} and the fact that $\norm{E}_{\diamond r}^{\ast}\leq 1$, by definition of a multi-time test operator. We have thus established \eqref{eq-bounded-complexity-regret_multi_time_pf1}. From here, an application of Theorem~\ref{theorem:regret-bound-sequential-covering}, along with the reasoning in the proof of Theorem~\ref{theorem:bounded-complexity-regret}, gives us the desired result.
\end{proof}

\begin{corollary}[Mistake bound for online learning multi-time processes of bounded complexity]
    Let $\varepsilon\in(0,1)$. There exists an online learning strategy that, in a realizable setting for $\mathsf{COMB}_{n,G}$, when presented sequentially with multi-time test operator $E^{(t)}$, $t\in\{1,2,\dotsc,T\}$, and associated loss functions $\ell_t(\cdot)=\ell((\cdot)-b_t)$, for some $L$-Lipschitz $\ell$, outputs a sequence $N^{(t)}\in\mathsf{COMB}_{n,G}$ of hypothesis comb operators that makes at most $\mathcal{O}\!\left(\frac{L^2G\log(Gn)}{\varepsilon^2}\right)$ many $\varepsilon$-mistakes.
\end{corollary}

\begin{proof}
    The proof follows the analogous arguments as in the proof of \Cref{corollary:bounded-complexity-mistake}, in which we make use of the regret bound from Theorem~\ref{theorem:bounded-complexity-regret_multi_time}.
\end{proof}

\paragraph*{Convex mixtures of known multi-time processes.} We now show that Theorem~\ref{theorem:pauli-regret} and Corollary~\ref{corollary:regret-mixture} generalize straightforwardly to convex mixtures of arbitrary, known multi-time processes.

\begin{theorem}[Regret bound for online learning convex mixtures of multi-time processes]\label{theorem:regret-bound-multi-time-convex-combination}
    Let $K \in \mathbb Z_{> 0}$ and $r\in\mathbb{N}$. Consider a convex combination of $K$ known $r$-step multi-time processes, with Choi representations $N_j\in\mathsf{COMB}_r$, $j\in\{1,2,\dotsc,K\}$, such that $N^{\ast} = \sum_{j=1}^{K} q_{j} N_j$, where the unknown probability distribution is given by $\Vec{q}=(q_1,q_2,\dotsc,q_K)$. Let $\ell:[0,1]\to\mathbb{R}$ be convex and $L$-Lipschitz. There exists an online learning strategy that, when presented sequentially with multi-time test operators $E^{(t)}$, $t\in\{1,2,\dotsc,T\}$, and associated losses $\ell_t (\cdot) = \ell ((\cdot) -b_t)$, outputs hypotheses $N^{(t)}\in\mathsf{COMB}_r$ of the form
    \begin{equation}
        N^{(t)}=\sum_{j=1}^K p_j^{(t)}N_j,
    \end{equation}
    with $\Vec{p}^{(t)}\coloneqq(p_1^{(t)},p_2^{(t)},\dotsc,p_K^{(t)})\in\Delta_K$, whose regret is bounded as
    \begin{equation}
        \sum_{t=1}^T \ell_t(\Tr[E^{(t)}N^{(t)}]) - \sum_{t=1}^T \ell_t(\Tr[E^{(t)}N^{\ast}]) = \mathcal{O}\left(L\sqrt{T \log K}\right).
    \end{equation}
\end{theorem}

\begin{proof}
We combine Lemma~\ref{lem-regret_via_pseudoregret} with the MWU algorithm (Algorithm~\ref{alg:mwu_general}) applied to a particular choice of loss vectors $\Vec{m}^{(t)}$. We let
\begin{align}
    \Vec{m}^{(t)}&\coloneqq(m_j^{(t)})_{j\in\{1,2,\dotsc,K\}},\\
    m_j^{(t)}&\coloneqq\frac{1}{L}\ell'_t(\Tr[E^{(t)}N^{(t)}])\Tr[E^{(t)}N_j]\quad\forall~j\in\{1,2,\dotsc,K\},
\end{align}
where
\begin{equation}
    N^{(t)}\coloneqq\sum_{j=1}^K p_j^{(t)}N_j,
\end{equation}
with the probability vectors $\Vec{p}^{(t)}\coloneqq(p_j^{(t)})_{j\in\{1,2,\dotsc,K\}}$ defined according to the MWU algorithm.

Let us verify that $m_j^{(t)}\in[-1,1]$ for all $j\in\{1,2,\dotsc,K\}$, as required by the MWU algorithm. First, we readily have that $\frac{1}{L}\ell'_t(\Tr[E^{(t)}N^{(t)}])\in[-1,1]$, because $\ell$ is assumed to be $L$-Lipschitz. In addition, we have that
\begin{align}
    \Tr[E^{(t)}N_j]&\leq\Tr[(S\otimes\mathbbm{1}_{B_r})N_j]\\
     \nonumber
    &=\Tr[S\Tr_{B_r}[N_j]]\\
     \nonumber
    &=\Tr[S(N_j^{(r-1)}\otimes\mathbbm{1}_{A_r})]\\
     \nonumber
    &=\Tr[\Tr_{A_j}[S]N_j^{(r-1)}]\\
     \nonumber
    &=\Tr[(S^{(r-1)}\otimes\mathbbm{1}_{B_{r-1}})N_j^{(r-1)}]\\
     \nonumber
    &~\vdots\nonumber\\
     \nonumber
    &=\Tr[S^{(1)}\Tr_{B_1}[N_j^{(1)}]]\\
     \nonumber
    &=\Tr[S^{(1)}]\\
     \nonumber
    &=1,
\end{align}
where the inequality is due to the fact that $\norm{E^{(t)}}_{\diamond r}^{\ast}\leq 1$, which means by \eqref{eq-strategy_norm_dual_main} that there exists $S\in\mathsf{COMB}^{\ast}_r$ such that $E^{(t)}\leq S\otimes\mathbbm{1}_{B_r}$; the chain of equalities holds because of the fact that $N_j\in\mathsf{COMB}_r$ and $S\in\mathsf{COMB}^{\ast}_r$, such that there exists $N_j^{(1)},N_j^{(2)},\dotsc,N_j^{(r-1)}$ satisfying the constraints in \eqref{eq-strategies} and there exists $S^{(1)},\dotsc,S^{(r-1)}$ satisfying the constraints in \eqref{eq-costrategies}.

Now, let $N^{\ast}=\sum_{j=1}^k q_j N_j$ for some probability vector $\Vec{q}=(q_1,q_2,\dotsc,q_K)$. Using Theorem~\ref{thm:MWU}, we obtain
\begin{align}
    &\sum_{t=1}^T \frac{1}{L}\ell'_t(\Tr[E^{(t)}N^{(t)}])\Tr[E^{(t)} N^{(t)}] - \sum_{t=1}^T \frac{1}{L}\ell'_t(\Tr[E^{(t)}N^{(t)}])\Tr[E^{(t)} N^{\ast}]\\
     \nonumber
     &\qquad=\sum_{t=1}^T\Vec{m}^{(t)}\cdot\Vec{p}^{(t)}-\sum_{t=1}^T\Vec{m}^{(t)}\cdot\Vec{q}\\
    &\qquad\leq\eta T + \frac{\log K}{\eta}.
    \nonumber
\end{align}
Therefore, by Lemma~\ref{lem-regret_via_pseudoregret}, we obtain
\begin{align}
  &\sum_{t=1}^T \left(\ell_t(\Tr[E^{(t)}N^{(t)}]) - \ell_t(\Tr[E^{(t)}N^{\ast}])\right) \\
   \nonumber
  &\qquad\leq  \sum_{t=1}^T \left(\ell'_t(\Tr[E^{(t)}N^{(t)}]) (\Tr[E^{(t)}N^{(t)}] - \Tr[E^{(t)}N^{\ast}]) \right) \\
   \nonumber
  &\qquad\leq L \left(\eta T + \frac{\log K}{\eta}\right).
\end{align}
Setting $\eta=\sqrt{\log K/T}$, we obtain the desired regret bound.
\end{proof}

\begin{corollary}[Mistake bound for convex mixtures of multi-time processes]
    Let $\varepsilon\in(0,1)$. 
    Let $K \in \mathbb Z_{> 0}$ and $r\in\mathbb{N}$. Consider a convex combination of $K$ known $r$-step multi-time processes, with Choi representations $N_j\in\mathsf{COMB}_r$, $j\in\{1,2,\dotsc,K\}$, such that $N^{\ast} = \sum_{j=1}^{K} q_{j} N_j$, where the unknown probability distribution is given by $\Vec{q}=(q_1,q_2,\dotsc,q_K)$.
    There exists an online learning strategy that, in a realizable setting for $\mathrm{conv}(\{N_j\}_{j=1}^K)$, when presented sequentially with multi-time test operator $E^{(t)}$, $t\in\{1,2,\dotsc,T\}$, and associated loss functions $\ell_t(\cdot)=\ell((\cdot)-b_t)$, for some $L$-Lipschitz $\ell$, outputs a sequence $N^{(t)}\in\mathsf{COMB}_{n,G}$ of hypothesis comb operators of the form
    \begin{equation}
        N^{(t)}=\sum_{j=1}^K p_j^{(t)}N_j,
    \end{equation}
    with $\Vec{p}^{(t)}\coloneqq(p_1^{(t)},p_2^{(t)},\dotsc,p_K^{(t)})\in\Delta_K$, that makes at most $\mathcal{O}(\frac{L^2 \log(K)}{\varepsilon^2})$ many $\varepsilon$-mistakes.
\end{corollary}
\begin{proof}
    Using \Cref{lemma:regret-to-mistake-template}, we can derive this mistake bound from the regret bound of \Cref{theorem:regret-bound-multi-time-convex-combination}, analogously to how \Cref{corollary:pauli-mistake-bound} followed from \Cref{corollary:regret-mixture}.
\end{proof}

\subsection{Learning-theoretic implications}\label{section:learning-theory-implications}

To conclude our discussion of regret and mistake upper bounds for online learning certain classes of quantum channels, we highlight some learning-theoretic implications of our bounds.
On the one hand, we note that our regret bounds immediately give rise to bounds on a complexity measure called \emph{sequential fat-shattering dimension} \cite{rakhlin2015online} of the respective channel classes.
On the other hand, our mistake-bounded online learner can be used to construct sample compression schemes. 
For simplicity of presentation, we focus on Pauli channels in this discussion. However, these implications immediately extend to all the classes of channels and the multi-time quantum processes that we established regret and mistake bounds for. That is, also for these classes we obtain complexity bounds and (approximate) compression schemes.

For the first implication, we rely on known results (\cite[Proposition 9]{rakhlin2015online} and \cite[Lemma 2]{rakhlin2015sequential}) to derive a sequential fat-shattering dimension bound from the regret bound established in \Cref{theorem:pauli-regret}:

\begin{corollary}[Sequential fat-shattering dimension bound]
    Let $\mathsf{PAULI}_n$ be the class of all $n$-qubit Pauli channels, and let $\varepsilon \in (0,1)$.
    Then, $\operatorname{sfat}_\varepsilon (\mathsf{PAULI}_n)\leq\mathcal{O}\!\left(\frac{n}{\varepsilon^2}\right)$.
\end{corollary}
\begin{proof}
    We start from \Cref{theorem:pauli-regret} for the special case of $\ell(\cdot) = |\cdot|$.
    This gives us a regret bound of $\mathcal{O}(L\sqrt{nT})$.
    Combining this with the first inequality in Ref.\ \cite[Proposition 9]{rakhlin2015online} or with \cite[Lemma 2]{rakhlin2015sequential}, we get
    \begin{equation}
        \frac{1}{4\sqrt{2}}\sup_{\varepsilon}\left\{ \varepsilon\sqrt{\frac{\min\{\operatorname{sfat}_\varepsilon(\mathcal{F}), T\}}{T}}\right\}
        \leq 2\sqrt{Tn\log 4}\, .
    \end{equation}
    We can now rearrange and, after 
    plugging in $T=1$ and using that clearly $\operatorname{sfat}_\varepsilon (\mathsf{PAULI}_n)\geq 1$, we get the claimed bound of $\operatorname{sfat}_\varepsilon (\mathsf{PAULI}_n)\leq\mathcal{O}\left(\frac{n}{\varepsilon^2}\right)$. 
\end{proof}

Via Ref.\ \cite[Corollary 1]{rakhlin2015sequential}, this sequential fat-shattering dimension bounds also implies sequential covering number bounds.
While obtaining these complexity bounds from our regret bounds (\Cref{theorem:pauli-regret}) is standard given prior work, we highlight that the obtained bounds are exponentially better than what would arise from naive parameter counting. Namely, while a Pauli channel is specified by a $4^n$-dimensional probability vector, these complexity measure bounds show that the ``effective'' dimension relevant for online learning is only linear in $n$. Finally, we note that the sequential fat-shattering dimension and covering numbers are upper bounds on their non-sequential counterparts. Thus, the above upper bounds immediately carry over to the complexity measures relevant for (agnostic) \emph{probably approximately correct} (PAC) learning and lead to corresponding generalization bounds for PAC learning Pauli channels. In particular, this implies that Pauli channels are a restricted class of operations that allow for ``pretty good process tomography'', as asked for in Ref.~\cite[Section 4]{Aaronson07}.
We note that Refs.~\cite{Aaronson07, ACH+19} obtained (sequential) fat-shattering bounds for quantum states from bounds on quantum random access coding. It would be interesting to see whether this implication can be reversed in our setting: Do our (sequential) fat-shattering bounds imply limitations for encoding classical information into Pauli channels in a random access coding fashion?

Next, we turn our attention to implications for compression.
For the case of $\{0,1\}$-valued functions, the connection between mistake-bounded online learning and sample compression via the so-called one-pass compression scheme has already been observed in Ref.~\cite[Section 4]{floyd1995sample}.
We notice that, with minor adaptations, this reasoning also applies to $[0,1]$-valued function classes when considering $\varepsilon$-mistakes and uniformly $\varepsilon$-approximate compression schemes (defined in
Ref.\ \cite{hanneke2019sample-compression}):

\begin{lemma}[Compression from mistake-driven online learning]\label{lemma:compression-from-online-learning}
    Let $\mathcal{F}\subseteq [0,1]^{\mathsf{X}}$. 
    Let $\varepsilon\in (0,1)$.
    Suppose $\mathsf{X}$ admits some total order and suppose $\mathcal{F}$ admits a mistake-driven online learner $\mathcal{A}$ that makes at most $M = M_{\mathcal{F}}(\varepsilon)$ many $\varepsilon$-mistakes when sequentially presented with challenges $x_1,\ldots,x_m$ and the corresponding values $f_\ast(x_1),\ldots,f_\ast(x_m)$ for some unknown $f_\ast\in\mathcal{F}$.
    Then, $\mathcal{F}$ admits a uniformly $\varepsilon$-approximate sample compression scheme.
\end{lemma}

The assumption that $\mathsf{X}$ admits a total order is trivially satisfied whenever $\mathsf{X}$ is finite, which is typically the case in computational learning theory. If we assume the ordering principle (i.e., that every set can be totally ordered), then we can apply \Cref{lemma:compression-from-online-learning} for a general instance space $\mathsf{X}$. 
We note that the ordering principle is strictly weaker than the well-ordering theorem \cite{halpern1964ordering, halpern1971boolean, gonzalez1995dense} (see also \cite[Section 4.4]{fraenkel1973foundations}), which in turn is well known to be equivalent to the axiom of choice.

\begin{proof}[Proof of Lemma~\ref{lemma:compression-from-online-learning}]
    We first describe the compression and reconstruction maps, then we prove that they have the desired compression scheme property.
    The compression map $\kappa:\bigcup_{m\geq 1} (\mathsf{X}\times [0,1])^m \to \bigcup_{1\leq m\leq M} (\mathsf{X}\times [0,1])^m$ is defined as follows:
    Given a dataset $S\subseteq \mathsf{X}\times [0,1]$, reorder $S$ according to the total order on $\mathsf{X}$, run the online learner $\mathcal{A}$ with an adversary that sequentially presents the learner with the reordered elements of $S$, and let $\kappa (S)$ be the set of (labeled) examples on which $\mathcal{A}$ made an $\varepsilon$-mistake.
    Next, we define the reconstruction map $\rho: \bigcup_{1\leq m\leq M} (\mathsf{X}\times [0,1])^m \to [0,1]^{\mathsf{X}}$. Let $S\subset \mathsf{X}\times [0,1]$ and $x\in\mathsf{X}$. If $\exists y\in [0,1]$ such that $(x,y)\in S$, then we set $\rho(S)(x)=y$.
    Otherwise, we reorder $S$ according to the total order on $X$, run the online learner $\mathcal{A}$ with an adversary that sequentially presents the learner with the reordered elements of $S$ that precede $x$ in the total order on $\mathsf{X}$, and let $\rho (S)(x)$ be the value predicted by $\mathcal{A}$ for $x$. 

    Now, we prove that $\kappa$ and $\rho$ as defined above indeed form a form a uniformly $\varepsilon$-approximate sample compression scheme for $\mathcal{F}$. That is, we show that, for all $f\in\mathcal{F}$ and for all $S=\{(x_i, f(x_i))\}_{i=1}^m\subset\mathsf{X}\times [0,1]$, the function $\hat{f} = \rho(\kappa(S))$ satisfies $\max_{1\leq i\leq m}\lvert \hat{f}(x_i) - f(x_i) \rvert\leq\varepsilon$.
    So, let $f\in\mathcal{F}$ and let $S=\{(x_i, f(x_i))\}_{i=1}^m\subset\mathsf{X}\times [0,1]$. 
    If $1\leq i\leq m$ is such that $\exists y_i\in [0,1]: (x_i,y_i)\in \kappa(S)$, then by definition of $\kappa$ and $\rho$ we have $\hat{f}(x_i) = y_i = f(x_i)$.
    If $1\leq i\leq m$ is such that $\nexists y_i\in [0,1]: (x_i,y_i)\in \kappa(S)$, then by definition of $\kappa$, this means that the online learner $\mathcal{A}$ does not make an $\varepsilon$-mistake on $x_i$ when presented sequentially with the elements of $S$ reordered according to the total order on $\mathsf{X}$. 
    As $\mathcal{A}$ is mistake-driven, this implies that $\mathcal{A}$ also does not make an $\varepsilon$-mistake on $x_i$ when presented sequentially only with the (reordered) elements of $S$ that precede $x_i$ in the total order on $\mathsf{X}$ and on which $\mathcal{A}$ made an $\varepsilon$-mistake. This is exactly the sequence of examples that $\mathcal{A}$ is run on when determining the value that $\hat{f}=\rho(\kappa(S))$ assigns to $x_i$, thus $\lvert \hat{f}(x_i) - f(x_i)\rvert\leq\varepsilon $.
\end{proof}

\begin{remark}\label{remark:improved-compression}
    Let us comment on a natural variant of \Cref{lemma:compression-from-online-learning}:
    If the online learner $\mathcal{A}$ for $\mathcal{F}$ makes at most $M=M_\mathcal{F}(\varepsilon)$ many $\varepsilon$-mistakes even when presented only with $(\varepsilon/3)$-accurate approximations to the true function values --  as is for instance the case for our MWU-based Pauli channel online learner --, this translates over to the resulting compression scheme. That is, even when sequentially presented with a training data set $S=\{(x_i,y_i)\}_{i=1}^m$ with $\lvert y_i - f(x_i)\rvert\leq \varepsilon/3$ for all $1\leq i\leq m$ (instead of a ``perfect'' data set with $y_i=f(x_i)$ for all $1\leq i\leq m$), the function $\hat{f} = \rho(\kappa(S))$ after compression still satisfies $\lvert \hat{f}(x_i) - f(x_i)\rvert\leq\varepsilon $ for all $1\leq i\leq m$.
    Thus, this sample compression scheme is successful also if applied to data that has been affected by (possibly adversarial) label noise of strength $\varepsilon/3$.~\qedgen
\end{remark}

We can combine this with our mistake bound for Pauli channel learning to get a compression scheme for Pauli channels:

\begin{corollary}[Compression scheme for Pauli channels]\label{corollary:pauli-compression}
    The set $\mathsf{PAULI}_n$ of $n$-qubit Pauli channels admits a uniformly $\varepsilon$-approximate sample compression scheme of size $\mathcal{O}\left(\frac{n}{\varepsilon^2}\right)$.
    This sample compression scheme even succeeds on training data whose labels have been corrupted by adversarial label noise of strength $\varepsilon/3$.
\end{corollary}

\begin{proof}
    Given our mistake bound for online learning Pauli channels (\Cref{corollary:pauli-mistake-bound}) and \Cref{lemma:compression-from-online-learning} (together with \Cref{remark:improved-compression}), we only have to show that our instance space, which is the space of channel test operators, admits a total order.
    This can be seen as follows:
    Similar to \Cref{section:online-learning-mixtures}, we can associate to any channel test operator $E_{A,B}$ the vector $(e_{\Vec{z},\Vec{x}}=\Tr[E_{A,B}\Gamma_{A,B}^{\Vec{z},\Vec{x}}])_{\Vec{z},\Vec{x}\in \{0,1\}^n}\in [0,1]^{4^n}$. The mapping $E_{A,B}\mapsto (e_{\Vec{z},\Vec{x}})_{\Vec{z},\Vec{x}\in \{0,1\}^n}$ is injective, because the operators $\Gamma_{A,B}^{\Vec{z},\Vec{x}}$ form an orthogonal basis. 
    Thus, any total order on $[0,1]^{4^n}$, for instance the lexicographic order, induces a total order on the set of channel test operators.
\end{proof}

\Cref{corollary:pauli-compression} implies that, if we care about the statistics of quantum experiments, then Pauli channels admit a significantly more parsimonious representation than via an exponentially long vector of Pauli error rates. 
This can be illustrated as follows: 
Suppose $A$(lice) and $B$(ob) want to understand how an unknown Pauli channel $\mathcal{N}$ in $A$'s lab acts on the channel test operators $E^{(i)}_{A,B}$, $1\leq i\leq m$.
To do so, $A$ performs experiments and collects data $S=\{(E^{(i)}_{A,B}, y_i)\}_{i=1}^m$, where the $y_i$ are $(\varepsilon/3)$-approximations of the corresponding expectation values $\Tr[E^{(i)}_{A,B} C^\mathcal{N}_{A,B}]$.
She now wants to communicate her findings to $B$.
Then, no matter how large $m$ is, $A$ can compress $S$ to a set of size at most $\mathcal{O}\left(\frac{n}{\varepsilon^2}\right)$ data points, send those to $B$, and $B$ can $\varepsilon$-approximately reconstruct the expectation values for all $m$ test operators without doing any further experiments.

\section{Mistake lower bounds}\label{sec:mistake-lower-bounds}

In the previous section, we provided regret and mistake upper bounds for online learning certain subclasses of quantum channels and multi-time processes. 
In this section, we prove complementary mistake lower bounds. 
While these, in principle, lead to regret lower bounds via (the contrapositive of) \Cref{lemma:regret-to-mistake-template}, we only discuss mistake lower bounds here.
Throughout this section, our focus is on the dependence on $\varepsilon$-mistake lower bounds for a constant $\varepsilon<1/2$. 
Thus, our lower bounds do not scale with $\varepsilon$. We conjecture that the $(1/\varepsilon^2)$-scaling achieved in \Cref{section:online-learning-bounded-complexity,section:online-learning-mixtures} is optimal, but leave the proof to future work.

\subsection{Mistake lower bounds for general unitaries and channels}\label{subsection:mistake-lower-bounds-general}

We first recall the folkore result that the class of arbitrary Boolean functions on $n$ bits cannot be online learned with subexponentially many mistakes.

\begin{lemma}[Arbitrary Boolean functions cannot be online learned with subexponentially many mistakes]\label{lemma:mistake-lower-bound-arbitrary-functions}
    Let $\mathcal{F}=\{0,1\}^{\{0,1\}^n}$ be the class of all Boolean functions on $n$ bits. 
    Any online learner for $\mathcal{F}$ makes at least $2^n$ mistakes against a worst-case adversary. 
    This remains true even if the adversary is forced to decide on a labeling function before the interaction with the learner.
\end{lemma}
\begin{proof}
    The class $\mathcal{F}$ of all $\{0,1\}$-valued functions on $\{0,1\}^n$ has VC-dimension \cite{vapnik1971uniform} $\operatorname{VCdim}(\mathcal{F})=2^n$ and thus Littlestone dimension \cite{littlestone1988learning} $\operatorname{Ldim}(\mathcal{F})\geq \operatorname{VCdim}(\mathcal{F})=2^n$. 
    Essentially by definition of the Littlestone dimension, any online learner for $\mathcal{F}$ makes at least $\operatorname{Ldim}(\mathcal{F})\geq 2^n$ mistakes. This proves the first part of the statement.

    Now for the second part of the statement.
    Fix an arbitrary learning algorithm. 
    Consider an adversary that initially chooses a function uniformly at random from $\{0,1\}^{\{0,1\}^n}$ and, in round $1\leq t\leq 2^n$, asks for the label of $\Vec{x}_t$, where $\{\Vec{x}_s\}_{s=1}^{2^n}$ is some (fixed) enumeration of $\{0,1\}^n$.
    Let $F$ denote the function-valued random variable describing the function chosen by the learner, let $S_t=\{(\Vec{x}_\tau, F(\Vec{x}_\tau))\}_{\tau=1}^t$ denote the instance-label pairs that the online learner has seen in the first $t$ rounds.
    Moreover, let $Y_{t+1}$ be the label predicted by the online learner in round $t+1$. Note that the random variable $Y_{t+1}$ depends only on $S_t$ (and on the internal randomness of the online learner).
    Thus, as $F(\Vec{x}_{t+1})$ is uniformly random and independent of $S_{t}$ (as well as of the online learning algorithm), also $Y_{t+1}$ and $F(\Vec{x}_{t+1})$ are independent.
    Therefore,
    \begin{align}
        \mathbb{P}[Y_{t+1} 
        = F(\Vec{x}_{t+1})]
        &= \mathbb{E}_{Y_{t+1}} [ \mathbb{P}_{F}[Y_{t+1} = F(\Vec{x}_{t+1}) | Y_{t+1} ] ]\\
        \nonumber
        &= \mathbb{E}_{Y_{t+1}} [ \mathbb{P}_{F}[Y_{t+1} = F(\Vec{x}_{t+1}) ] ]\\
         \nonumber
        &= \mathbb{E}_{Y_{t+1}}\!\left[ \frac{1}{2} \right]
        = \frac{1}{2} \, .
         \nonumber
    \end{align}
    Hence, in each round, the online learner makes a mistake with probability $1/2$. 
    As mistakes occur independently in each round, the probability that the online learner makes a mistake in every round is strictly greater than zero.
    Thus, there exists a function $f:\{0,1\}^n\to\{0,1\}$ that, when chosen initially by the adversary, forces the online learner to make $2^n$ mistakes. 
\end{proof}

We can now embed the problem of online learning an arbitrary classical Boolean function into that of learning an arbitrary quantum channel and therefore inherit similar mistake lower bounds. In the following, we describe two different ways of achieving such lower bounds.

\begin{corollary}\label{corollary:mistake-lower-bound-unitaries}
    Let $\mathbb{U}_{n}$ be the class of all unitary $n$-qubit channels, let $\varepsilon<1/2$.
    Any online learner for $\mathbb{U}_{n}$ makes $\Omega(2^n)$ many $\varepsilon$-mistakes against a worst-case adversary.
    This remains true even if the adversary is forced to decide on a unitary before the interaction with the learner.
\end{corollary}

\begin{proof}
    Our proof is via reduction to \Cref{lemma:mistake-lower-bound-arbitrary-functions}.
    To do so, we associate to every $f:\{0,1\}^{n-1}\to\{0,1\}$ the $n$-qubit unitary $U_f$ defined via $U_f \ket{x,b} = \ket{x,b\oplus f(x)}$ for $x\in\{0,1\}^{n-1}$ and $b\in\{0,1\}$. 
    We denote the corresponding unitary channel by $\mathcal{U}_f$\footnote{We note that this standard embedding of classical functions into unitaries has already been observed to give lower bounds for PAC learning quantum channels in 
    Ref.~\cite{Aaronson07}.}.
    Now, if we consider a channel test operator $E_{A,B}(x)=\rho_A(x)^\t \otimes M_B(x)$ with $\rho_A(x) = \ketbra{x,0}{x,0}$ and $M_B(x) = \ketbra{x,1}{x,1}$ for some $x\in\{0,1\}^{n-1}$, then 
    \begin{equation}
        \Tr[E_{A,B} (x) C^{\mathcal{U}_f}_{A,B}]
        = \Tr[M_{B}(x)  \mathcal{U}_f(\rho_A(x))]
        =\abs{\braket{1}{f(x)}}^2 = f(x)\, .
    \end{equation}
    Thus, if $\varepsilon<1/2$, then any online learner for $\mathbb{U}_n$ that makes at most $N$ $\varepsilon$-mistakes gives rise to an online learner for $\{0,1\}^{\{0,1\}^{n-1}}$ that makes at most $N$ mistakes, by rounding the produced estimates to obtain a label in $\{0,1\}$. 
    Hence, by \Cref{lemma:mistake-lower-bound-arbitrary-functions}, we conclude that $N\geq 2^{n-1} \geq\Omega(2^n)$.
\end{proof}

\begin{corollary}\label{corollary:mistake-lower-bound-channels}
    Let $\mathsf{CPTP}_n$ be the class of all $n$-qubit quantum channels, let $\varepsilon<1/2$.
    Any online learner for $\mathsf{CPTP}_n$ makes $\Omega(2^n)$ many $\varepsilon$-mistakes against a worst-case adversary.
    This remains true even if the adversary is forced to decide on a channel before the interaction with the learner.
\end{corollary}

\begin{proof}
    Our proof is via reduction to \Cref{lemma:mistake-lower-bound-arbitrary-functions}.
    To do so, we associate to every $f:\{0,1\}^{n}\to\{0,1\}$ the $n$-qubit channel $\mathcal{N}_f$ defined via 
    \begin{equation}
        \mathcal{N}_f(\rho)
        = \sum_{x\in\{0,1\}^n} \bra{x}\rho\ket{x} (\ketbra{f(x)}{f(x)}\otimes \ketbra{0^{n-1}}{0^{n-1}})\, .
    \end{equation}
    Now, if we consider a channel test operator $E_{A,B}(x)=\rho_A(x)^\t \otimes M_B(x)$ with $\rho_A(x) = \ketbra{x}{x}$ and $M_B(x) = \ketbra{1}{1}\otimes\ketbra{0^{n-1}}{0^{n-1}}$ for some $x\in\{0,1\}^{n}$, then 
    \begin{equation}
        \Tr[E_{A,B} (x) C^{\mathcal{N}_f}_{A,B}]
        = \Tr[M_{B}(x)  \mathcal{N}_f(\rho_A(x))]
        = \abs{\braket{1}{f(x)}}^2 = f(x)\, .
    \end{equation}
    Thus, if $\varepsilon<1/2$, then any online learner for $\mathsf{CPTP}_n \supset \{\mathcal{N}_f\}_{f\in\{0,1\}^{\{0,1\}^n}}$ that makes at most $N$ $\varepsilon$-mistakes gives rise to an online learner for $\{0,1\}^{\{0,1\}^{n}}$ that makes at most $N$ mistakes, by rounding the produced estimates to obtain a label in $\{0,1\}$.
    Hence, by \Cref{lemma:mistake-lower-bound-arbitrary-functions}, we conclude $N\geq 2^{n}$.
\end{proof}

These lower bounds demonstrate that, unsurprisingly, arbitrary quantum channels cannot be online learned with a number of mistakes that scales efficiently with the system size. 
This should be contrasted with the case of states: Ref.~\cite{ACH+19} proved that we can online learn the class of all $n$-qubit states with a number of mistakes that grows only linearly in $n$.
Moreover, the exponential mistake lower bounds above motivate the focus on restricted subclasses of channels, such as channels of bounded gate complexity or mixtures of known channels, which we considered in \Cref{section:online-learning-bounded-complexity,section:online-learning-mixtures}.
In the next two subsections, we prove mistake lower bounds to be juxtaposed with the upper bounds established in \Cref{section:online-learning-bounded-complexity,section:online-learning-mixtures}.

\subsection{Mistake lower bounds for channels of bounded complexity}

Recall that \Cref{theorem:bounded-complexity-regret,corollary:bounded-complexity-mistake} established regret and mistake upper bounds for online learning channels of gate complexity $G$ that scaled effectively linearly in $G$. Our next result, which follows by combining \Cref{corollary:mistake-lower-bound-channels} with a ``zooming in'' on a suitable subset of qubits (as previously employed in Ref.~\cite{zhao2023learning}) shows that this scaling is essentially optimal.

\begin{corollary}[Essentially optimal scaling]\label{corollary:mistake-lower-bound-bounded-complexity}
    Let $\mathsf{CPTP}_{n,G}$ be the class of all $n$-qubit $G$-gate channels, i.e., channels of gate complexity (at most) $G$, and let $\varepsilon<1/2$.
    Any online learner for $\mathsf{CPTP}_{n,G}$ makes $\Omega(\min\{2^n,G\})$ many $\varepsilon$-mistakes against a worst-case adversary.
    This remains true even if the adversary is forced to decide on a $G$-gate channel before the interaction with the learner.
\end{corollary}
\begin{proof}
	As we explain in \Cref{remark:channel-mistake-lower-bound-G<=n} below, for $G\leq n$ the claimed lower bound follows from our analysis for Pauli channels. Therefore, for the rest of this proof, we consider only $G>n$.
	
	Recall that there is a universal constant $C>0$ such that every Boolean function $f:\{0,1\}^k\to\{0,1\}$ can be implemented with a de Morgan circuit of size at most $C\cdot 2^k/k$~\cite{lupanov1958method}. Here, a de Morgan circuit is a circuit consisting of AND, OR, and NOT gates, where the AND and OR gates have fan-in two. 
	Hence, as any classical two-bit gate can be implemented by a two-qubit quantum channel gate, and as a computational basis measurement on a single qubit corresponds to one single-qubit channel gate, we see that every $k$-qubit channel $\mathcal{N}_f$ as in the proof of \Cref{corollary:mistake-lower-bound-channels}, with $f:\{0,1\}^k\to\{0,1\}$, can be implemented with $C\cdot 2^k/k + k \leq C2^{k+1}$ many two-qubit channel gates. 
	Therefore, if we set $k=\lfloor \log_2( G/2C)\rfloor$, then $\mathsf{CPTP}_{n,G}\supseteq \{\mathcal{N}_f\otimes\operatorname{id}_{n-q}\}_{f:\{0,1\}^q\to\{0,1\}}$, where we consider the channels $\mathcal{N}_f$ from the proof of \Cref{corollary:mistake-lower-bound-channels} and where we set $q=\min\{n, k\}$.
	We can now straightforwardly modify the states and effect operators used in the proof of \Cref{corollary:mistake-lower-bound-channels} (by attaching, say, the all-zero state on the last $n-q$ qubits) to show that the $\varepsilon$-mistake bound from \Cref{lemma:mistake-lower-bound-arbitrary-functions}, with $n$ replaced by $q$, applies to $\{\mathcal{N}_f\otimes\operatorname{id}_{n-q}\}_{f:\{0,1\}^q\to\{0,1\}}$.
	Because of the previously observed inclusion, we conclude that also $\mathsf{CPTP}_{n,G}$ comes with an $\varepsilon$-mistake lower bound of $2^{q}\geq \min\{2^n, G/4C\} = \Omega(\min\{2^n, G\})$.
\end{proof}

\Cref{corollary:mistake-lower-bound-bounded-complexity} establishes a mistake lower bound for online learning the class $\mathsf{CPTP}_{n,G}$ whose $G$-dependence matches our upper bound up to logarithmic factors. However, as the construction in the proof above uses a measurement followed by an in general non-reversible classical circuit, it is far from unitary. Therefore, we next give an alternative proof for \Cref{corollary:mistake-lower-bound-bounded-complexity} that, while still using non-unitary building blocks, is motivated by reversible computation and therefore can maybe serve as a stepping stone towards an analogue of \Cref{corollary:mistake-lower-bound-bounded-complexity} for unitary channels. 

\begin{proof}[Alternative proof of \Cref{corollary:mistake-lower-bound-bounded-complexity}]
	As before, Ref.~\cite{lupanov1958method} tells us that every Boolean function $f:\{0,1\}^k\to\{0,1\}$ can be implemented with a de Morgan circuit of size at most $C\cdot 2^k/k$. As any OR gate can be rewritten in terms of three NOT gates and one AND gate, we can also achieve such implementations with circuit size $C\cdot 2^{k+2} / k$ using only AND and NOT gates.
	The NOT gate can trivially be quantumly implemented by the Pauli-$X$ gate.
	Using what Ref.~\cite{toffoli1980reversible} called the AND/NAND gate, and which (with a different ordering convention for the inputs) is now known as the Toffoli gate, we can implement an AND with a reversible three-bit gate when the ``source'' is a suitable constant bit, thereby producing two garbage output bits in the ``sink''.
	Therefore, we can quantumly implement any two-bit AND gate using one three-qubit unitary in conjunction with a single-qubit channel gate that resets one of the ``sink'' qubits to the needed constant input, so that it can serve as a ``source'' for the next AND. The reset also allows us to use a single auxiliary qubit only throughout. 
	Consequently, as any three-qubit unitary can be decomposed into a constant number of two-qubit unitaries, we can implement every function $f:\{0,1\}^k\to\{0,1\}$ with a $(k+1)$-qubit quantum circuit of size $C\cdot 2^{k+3}/k$ (where $C$ is now a new constant).
	
	By this line of reasoning, if we set $k=\lfloor \log_2( G/8C) \rfloor$ and $q=\min\{n-1, k\}$, then $\mathsf{CPTP}_{n,G}\supseteq \{\mathcal{N}_f\otimes \operatorname{id}_{n-(q+1)}~|~f:\{0,1\}^q\to\{0,1\}\}$, where we abused notation---by not writing out the restriction to computational basis inputs and measurements on the first $q+1$ qubits, and by ignoring the ``source'' and the ``sink'' subsystems at the output, which we can achieve by having identity tensor factors on the corresponding subsystems of the output effect operator. At this point, we again inherit a mistake lower bound from \Cref{corollary:mistake-lower-bound-channels}, which here becomes $2^q = \min\{2^{n-1}, G/16C\}\geq\Omega(\min\{2^n, G\})$.
\end{proof}

This second proof already hints at a challenge in establishing the same $\Omega(\min\{2^n,G\})$ mistake lower bound for $\mathbb{U}_{n,G}$, the class of all unitary $n$-qubit channels of gate complexity $G$. Namely, when aiming to implement a unitary $U_f$ for $f:\{0,1\}^q \to\{0,1\}$, a natural approach is to take a classical circuit implementation for $f$ and make it reversible. Above, we relied on Toffoli's construction to achieve this with a small overhead in gate complexity. However, since there we need a specific constant input in the ``source'', this required us to reset (some of) our ``sink'' qubits. Such a reset is a non-reversible operation. 
Without the ability to reset, making the implementation reversible naively requires to add one auxiliary qubit per gate, thus exceeding the number of available qubits if $G\geq n$.
While we do not yet know how to overcome this obstacle when only using unitary gates, the following result at least demonstrates a lower bound for the unitary case that deviates from the $G$-dependence in the upper bound by only a square root.

\begin{corollary}[Lower bound for the unitary case]\label{corollary:mistake-lower-bound-bounded-complexity-unitary}
    Let $\mathbb{U}_{n,G}$ be the class of all unitary $n$-qubit channels of gate complexity (at most) $G$, let $\varepsilon<1/2$.
    Any online learner for $\mathbb{U}_{n,G}$ makes $\Omega(\min\{2^n,\sqrt{G}\})$ many $\varepsilon$-mistakes against a worst-case adversary.
    This remains true even if the adversary is forced to decide on a $G$-gate channel before the interaction with the learner.
\end{corollary}
\begin{proof}
    We first recall from Refs.~\cite{vartiainen2004efficient, mottonen2004quantum, shende2005synthesis}: There is a universal constant $C>0$ such that any $k$-qubit unitary can be implemented with $C 4^k$ many two-qubit gates. 
    Thus, if we set $k = \lfloor\log_4(G/C)\rfloor$, then $G$ many two-qubit gates suffice to implement arbitrary unitaries on $k$ qubits.
    In particular, this implies that $\mathbb{U}_{n,G}\supset \mathbb{U}_q \otimes\operatorname{id}_{n-q}$, where $q=\min\{n,k\}$.
    Therefore, a quantum circuit with $G$ many unitary $2$-qubit gates is able to implement all the unitaries $U_f$ for functions $f:\{0,1\}^{q-1}\to\{0,1\}$ on the first $q$ qubits.
    Hence, with a straightforward modification of the reasoning used in proving \Cref{corollary:mistake-lower-bound-unitaries}---tensoring the input states and output effects used there with, say, the all-zero state on the remaining $n-q$ qubits---, we inherit the $\varepsilon$-mistake bound from \Cref{lemma:mistake-lower-bound-arbitrary-functions} with $n$ replaced by $q$.
    That is, we obtain a $\varepsilon$-mistake lower bound of $\geq 2^{q-1} = \min\{2^n/2, \sqrt{G/C}/4\} = \Omega(\min\{2^n, \sqrt{G}\})$.
\end{proof}

\subsection{Mistake lower bounds for Pauli channels}

To complement the linear-in-$n$ mistake upper bound from \Cref{section:online-learning-mixtures}, we now give a mistake lower bound for online learning Pauli channels. 
Again, we obtain this as a consequence of \Cref{lemma:mistake-lower-bound-arbitrary-functions}.

\begin{corollary}[Mistake lower bound for online learning Pauli channels]\label{corollary:mistake-lower-bound-pauli-channels}
    Let $\mathsf{PAULI}_n$ be the class of all $n$-qubit Pauli channels, let $\varepsilon<1/2$.
    Any online learner for $\mathsf{PAULI}_n$ makes $\Omega(n)$ many $\varepsilon$-mistakes against a worst-case adversary.
    This remains true even if the adversary is forced to decide on a Pauli channel before the interaction with the learner.
\end{corollary}
\begin{proof}
    Our proof is via reduction to \Cref{lemma:mistake-lower-bound-arbitrary-functions}.
    To do so, we associate to any $f:\{1,\ldots,n\}\to\{0,1\}$ the unitary $n$-qubit Pauli channel $\mathcal{N}_f$ defined via $\mathcal{N}_f (\rho) = \left(\bigotimes_{i=1}^n Z_i^{f(i)}\right)\rho \left(\bigotimes_{i=1}^n Z_i^{f(i)}\right)$.
    Now, if we consider a channel test operator $E_{A,B}(t)=\rho_A(t)^\t \otimes M_B(t)$ with $\rho_A(t)=\ketbra{0^{t-1}}{0^{t-1}}\otimes \ketbra{+}{+}\otimes \ketbra{0^{n-t}}{0^{n-t}}$ and $M_B(t) = \ketbra{0^{t-1}}{0^{t-1}}\otimes \ketbra{-}{-}\otimes \ketbra{0^{n-t}}{0^{n-t}}$ for some $t\in\{1,\ldots,b\}$, then
    \begin{equation}
        \Tr[E_{A,B} (t) C^{\mathcal{N}_f}_{A,B}]
        = \Tr[M_{B}(t)  \mathcal{N}_f(\rho_A(t))]
        = f(t)\, .
    \end{equation}
    Thus, if $\varepsilon<1/2$, then any online learner for $\mathsf{PAULI}_n$ that makes at most $N$ $\varepsilon$-mistakes gives rise to an online learner for $\{0,1\}^{\{1,\ldots,n\}}$ that makes at most $N$ mistakes, by rounding the produced estimates to obtain a label in $\{0,1\}$.
    Hence, by \Cref{lemma:mistake-lower-bound-arbitrary-functions}, we conclude $N\geq 2^{\lfloor\log(n)\rfloor}\geq  \Omega(2^{\log(n)}) = \Omega(n)$.
\end{proof}

\Cref{corollary:mistake-lower-bound-pauli-channels} shows that the linear-in-$n$ scaling achieved in \Cref{section:online-learning-mixtures} is optimal for the special case of Pauli channel online learning.
This also tells us that for learning a mixture of $K$ known channels, the $\log(K)$-dependence in the mistake bound can in general not be improved.

\begin{remark}\label{remark:channel-mistake-lower-bound-G<=n}
    The channels $\mathcal{N}_f$ used in the proof of \Cref{corollary:mistake-lower-bound-pauli-channels} are unitary channels of gate complexity $n$. 
    We can use essentially the same construction to show that for $G\leq n$, any online learner for $\mathsf{CPTP}_{n,G}$ makes at least $\Omega(G)$ many $\varepsilon$-mistakes if $\varepsilon <1/2$. To see this, consider unitary $n$-qubit channels of the form $\rho\mapsto \left(\bigotimes_{i=1}^G Z_i^{f(i)} \otimes \mathbbm{1}_2^{\otimes (n-G)}\right)\rho \left(\bigotimes_{i=1}^G Z_i^{f(i)} \otimes \mathbbm{1}_2^{\otimes(n-G)}\right)$ for functions $f:\{1,\ldots,G\}\to\{0,1\}$, and argue as in the proof of \Cref{corollary:mistake-lower-bound-pauli-channels}.~\qedgen
\end{remark}

\section{Computational complexity lower bounds}

\subsection{Computational complexity lower bounds for Pauli channels}\label{sec:computational-lower-bounds-pauli-channels}

In \Cref{sec:mistake-lower-bounds}, we established mistake lower bounds for channels. In this section, we focus on computational complexity lower bounds for online learning classes of quantum channels. In particular, we show that while we achieve polynomial mistake upper bounds for learning classes of quantum channels, the exponential computational complexity of our learning algorithms cannot be avoided under standard cryptographic assumptions. But let us first start with an \textit{unconditional} computational hardness lower bound for online learning Pauli channels. The heart of the proof strategy lies in exploiting the fact that any polynomial time learner is only ever able to access polynomially many entries of an exponentially sized input. By the adversarial nature of the game, the adversary can always `hide' the information useful for answering a challenge in an entry that was never seen by the (polynomial-time) learner. We make this intuition formal in the following result.

\begin{theorem}\label{theorem:computational-hardness-pauli-channel-online-learning}
    Consider any polynomial-time online learner of $n$-qubit Pauli channels that runs in time $q^{(t)}(n)$ at time step $t\in\{1,2,\dotsc,T\}$, for any polynomials $q^{(1)}(n), q^{(2)}(n),\dotsc, q^{(T)}(n)$. There exists an explicit adversarial strategy that forces the learner to make $\frac{4^n - 1}{Q}$ mistakes, where $Q=\min\{q^{(t)}(n):t\in\{1,2,\dotsc,T\}\}$.
\end{theorem}

\begin{proof}
    Recall that for any unknown Pauli channel with associated Choi representation $N_{A,B}$ predicting $\Tr[E_{A,B}^{(t)}N_{A,B}]$ (the task of the online learner) for a given channel observable $E_{A,B}^{(t)}$ is exactly equivalent to predicting
    $\Vec{e}^{(t)}\cdot \Vec{p}$, where $\Vec{p}$ is the (unknown) error-rate distribution and the vector $\Vec{e}^{(t)}=(e_{\Vec{z},\Vec{x}}^{(t)})_{\Vec{z},\Vec{x}\in\{0,1\}^n}$ has entries $e_{\Vec{z},\Vec{x}}^{(t)} \coloneqq \Tr[E_{A,B}^{(t)}\Gamma_{A,B}^{\Vec{z},\Vec{x}}]$ for all $\Vec{z},\Vec{x}\in\{0,1\}^n$.
    
    We work in a simplified setting of a $\{0,1\}$-valued game, in which the learner does not have to evaluate entries $e_{\Vec{z},\Vec{x}}^{(t)} \coloneqq \Tr[E_{A,B}^{(t)}\Gamma_{A,B}^{\Vec{z},\Vec{x}}]$ of the vector $\Vec{e}^{(t)}=(e_{\Vec{z},\Vec{x}}^{(t)})_{\Vec{z},\Vec{x}\in\{0,1\}^n}$ for all $\Vec{z},\Vec{x}\in\{0,1\}^n$; the learner directly receives the vector $\Vec{e}^{(t)}$ (which we shall refer to as the `challenge' vector), instead of the test operators $E_{A,B}^{(t)}$. The learning task is still not obviously computationally tractable, because the challenge vector is of size $4^n$. We prove that a polynomial-time learner, who by definition only looks at polynomially-many entries of $\Vec{e}^{(t)}$, can be forced to make $\tilde{\Omega}(4^n)$ many mistakes by a simple adversarial strategy. The adversarial strategy is to always play $\Vec{e}^{(t)}=\Vec{0}=(0,0,\dotsc,0)$, i.e., the all-zeros vector, corresponding to $E_{A, B}^{(t)} = 0$.\footnote{Note that we need to ensure that there is always a channel observable $E_{A, B}^{(t)}$ that realizes $\Vec{e}^{(t)}$ in each round, because the input to the online learning game is defined with respect to $E_{A, B}^{(t)}$.} We now show that by using the same all-zeros challenge vector in every round, the adversary can always contradict the learner's prediction, and thereby claim that they made a mistake.
    In any $T$-round interaction with the adversary, let the (deterministic) learner predict $y_t \in\{0,1\}$ in round $t\in\{1,2,\dotsc,T\}$. In response, the adversary claims the correct answer to be $b_t=\neg y_t$. In other words, the adversary always contradicts the learner's prediction, and thereby forces the learner to make a mistake. Note that this contradictory feedback indeed constitutes a mistake when we take the loss function to be $\ell_t(y_t) \coloneqq |y_t - b_t|$, because $\ell_t(y_t) = |y_t - (\neg y_t)| = 1$ for all $t \in\{1,2,\dotsc,T\}$.

    Let us now prove that after the end of the $T$ rounds, the adversary can claim to be ``consistent'', despite their contradictory feedback, as long as it does not contradict the entries that the learner has seen. In other words, after the end of the $T$ rounds, the adversary can always exhibit a $\Vec{p}^{\ast}\in\Delta_{4^n}$ and challenge vectors $\tilde{\Vec{e}}^{(t)}$ such that, for all $t\in \{1,2,\ldots,T\}$, $\neg y_t=\Vec{p}^{\ast}\cdot\tilde{\Vec{e}}^{(t)}$ and such that $\tilde{\Vec{e}}^{(t)}$ has $0$-entries at the positions of the challenge vector that the online learner accessed in round $t$. The key to proving this is the fact that the learner is computationally bounded, and therefore, they can only query polynomially-many entries of $\Vec{e}^{(t)}$ in every round.
    
    In any given round, the adversarial feedback, $\neg y_t$, is either $0$ or $1$. For $\neg y_t = 0$, the adversary claims to actually have played the all-zero challenge vector. For the rounds where the adversary claimed $\neg y_t = 1$, the adversary needs to demonstrate that, while the learner only saw all zeroes in such rounds, there was in fact a non-zero entry in $\Vec{e}^{(t)}$ (that the learner did not look at) and that this entry leads to $\Vec{p}^{\ast}\cdot \Vec{e}^{(t)} = 1$. In fact, as we argue below, it suffices for the adversary to only claim that $\Vec{e}^{(t)}$ had a single non-zero entry. 
    Vectors $\Vec{e}^{(t)}$ with this structure are indeed realized in our learning scenario by choosing $E_{A, B}^{(t)} = \Gamma_{A,B}^{\Vec{z}',\Vec{x}'}/4^n$ for some $\Vec{z}',\Vec{x}'\in\{0,1\}^n$, because $\{\Gamma_{A,B}^{\Vec{z},\Vec{x}}\}_{\Vec{z},\Vec{x}\in\{0,1\}^n}$ forms an orthogonal basis and $\Tr[(\Gamma_{A,B}^{\Vec{z},\Vec{x}})^2] = 4^n$. Thus, choosing $E_{A,B}^{(t)} = \Gamma_{A,B}^{\Vec{z}',\Vec{x}'}/4^n$ ensures that $\Vec{e}^{(t)}$ is has $4^n - 1$ entries equal to $0$ and a single entry $e_{\Vec{z},\Vec{x}}$ equal to $1$, when $\Vec{z} = \Vec{z'}$ and $\Vec{x} = \Vec{x'}$. 

Let us partition the set $R\coloneqq\{1,2,\dotsc,T\}$ into two disjoint subsets as $R=R_0\cup R_1$, where $R_0$ is the set of time steps in which the adversary claimed the correct answer to be $0$ and $R_1$ is the set of time steps in which they claimed the correct answer to be $1$. Let $q^{(t)}(n)$ be the polynomial number of entries of $\Vec{e}^{(t)}$ accessed by the learner at time step $t$. Also, let $I_t\subseteq\{0,1\}^n\times\{0,1\}^n$ be the indices corresponding to the entries of $\Vec{e}^{(t)}$ accessed by the learner in round $t$. To be consistent, the adversary sets to $0$ all entries of $\Vec{e}^{(t)}$ and $\Vec{p}^{\ast}$ corresponding to the indices in $(\bigcup_{t\in R_0}I_{t})\cup(\bigcup_{t\in R_1}I_{t})$, i.e., all entries that the learner saw. Then, retroactively, for the rounds in which they claimed $1$ to be the correct answer, the adversary can always claim that the $1$-entry in both $\Vec{e}^{(t)}$ and $\Vec{p}^{\ast}$ was an entry that the learner never saw, i.e., an entry whose index is in $\left((\bigcup_{t\in R_0}I_{t})\cup(\bigcup_{t\in R_1}I_{t})\right)^c$. It is sufficient that there exists at least one such index. Then, this adversarial strategy works as long as the number of entries seen by the learner does not exceed $4^n-1$ (in order to account for at least one entry that the learner has not seen). In other words, 
\begin{align}
    4^{n} - 1 
    &\geq \sum_{t\in R_0} q^{(t)}(n) +  \sum_{t\in R_1} q^{(t)}(n) \\
    \nonumber
    &\geq q_0(n)|R_0| + q_1(n)|R_1| \\
    \nonumber
    &\geq QT,
\end{align}
where $q_0(n) \coloneqq \min\{q^{(t)}(n):t\in R_0\}$ and $q_1(n) \coloneqq \min\{q^{(t)}(n):t\in R_1\}$, and the final inequality holds because $q_0(n)\geq Q$, $q_1(n)\geq Q$ (recall that $Q\coloneqq\min\{q^{(t)}(n):t\in\{1,2,\dotsc,T\}\}$), and $|R_0|+|R_1|=T$. Thus, the adversary can force the learner to make mistakes for $\frac{4^n - 1}{Q}$ many rounds.
\end{proof}

\begin{remark}
    The computational complexity of MWU (\Cref{alg:mwu_general}) for online learning convex combinations of $K$ known channels scales polynomially with $K$, which in the worst case could be exponential in the number of qubits, as is the case for general Pauli channels. If, however, the learner is given challenge vectors $\Vec{e}^{(t)}$ with entries $e_{\Vec{z},\Vec{x}}^{(t)} \coloneqq \Tr[E_{A,B}^{(t)}\Gamma_{A,B}^{\Vec{z},\Vec{x}}]$ that are $\text{poly}(n)$-sparse (with known sparsity structure), then an online learner can learn such a channel computationally efficiently and also saturate optimal regret and mistake bounds using MWU (\Cref{alg:mwu_general}). A relevant example (from quantum error correction) of a class of channels that can be written as convex combinations of polynomially many known channels is that of polynomially-sparse Pauli channels with a known sparsity structure. Hence, our results imply that this class of channels is computationally efficiently online learnable with regret and mistake bounds that scale with $\log(n)$.~\qedgen
\end{remark}

\begin{remark}
    The proof strategy in \Cref{theorem:computational-hardness-pauli-channel-online-learning} straightforwardly implies that any polynomial time learner for online learning quantum \textit{states} (in the sense defined in Ref.\ \cite{ACH+19}) can be forced to make exponentially many $\varepsilon$-mistakes as long as the `challenge' effect operators admit exponentially long descriptions. To see this, write, in its spectral decomposition, any $n$-qubit state $\rho = \sum_{i=1}^{2^n} p_i \ketbra{\psi_i}{\psi_i}$ that a learner wishes online learn. For the lower bound, it suffices to work in a simpler scenario in which the learner knows the eigenbasis $\{\ket{\psi_i}\}_{i=1}^{2^n}$ in advance but not the eigenvalues $\{p_i\}_{i=1}^{2^n}$. For every effect operator $E$, we have $\Tr[\rho E] = \sum_{i=1}^{2^n} p_i \bra{\psi_i}E\ket{\psi_i}$. Defining a vector $\Vec{e}=(e_i)_{i=1}^{2^n}$ with entries $e_i = \bra{\psi_i}E\ket{\psi_i}$, we can rewrite this as $\Tr[\rho E] = \Vec{p}\cdot \Vec{e}$, where $\Vec{p}=(p_i)_{i=1}^{2^n}$. Since $\Vec{e}$ is still an exponentially long challenge vector, the strategy in the proof of \Cref{theorem:computational-hardness-pauli-channel-online-learning} suffices to make the online learner make exponentially many mistakes.~\qedgen
\end{remark}

\subsection{Computational complexity lower bounds for channels of bounded complexity}\label{sec:computational-lower-bounds-bounded-gate-comlexity}

The online learner for $\mathsf{CPTP}_{n,G}$ presented in \Cref{section:online-learning-bounded-complexity} achieves good regret and $\varepsilon$-mistake bounds, but is computationally inefficient. In this section, we prove that, under a widely held cryptographic hardness assumption, namely hardness of \textsf{RingLWE}, there cannot be a computationally efficient online learning algorithm for $\mathsf{CPTP}_{n,G}$ that achieves favorably scaling mistake bounds. Via \Cref{lemma:regret-to-mistake-template}, this also implies that good regret bounds for online learning $\mathsf{CPTP}_{n,G}$ cannot be achieved computationally efficiently, but we again focus on mistake bounds here. Our proof, which is conceptually analogous to arguments in \cite{zhao2023learning, mele2024efficient}, is yet another instance of the well known fact that pseudorandom functions cannot be learned efficiently, which we phrase in an online learning framework.

First, we recall the definition of pseudorandom functions.

\begin{definition}[Pseudorandom functions (PRFs)~\cite{goldreich1986construct}]\label{definition:prfs}
    Let $\lambda$ be a security parameter.
    Let $\mathcal{K} = \{\mathcal{K}_\lambda\}_{\lambda\in \mathbb{N}}$ be key space, assumed to be efficiently sampleable.
    Let $\mathcal{X} = \{\mathcal{X}_\lambda\}_{\lambda\in \mathbb{N}}, \{\mathcal{Y}_\lambda\}_{\lambda\in\mathbb{N}}$ be families of finite sets.
    Let $\mathcal{F} = \{f_\lambda\}_{\lambda\in\mathbb{N}}$ be a family of efficiently-computable functions $f_\lambda: \mathcal{K}_\lambda \times \mathcal{X}_\lambda \to \mathcal{Y}_\lambda$, where the input from $\mathcal{K}_\lambda$ corresponds to the function key.
    The family $\mathcal{F}$ is a \emph{pseudorandom function (family) secure against (classical) $t(\lambda)$-time adversaries} if for every $t(\lambda)$-time probabilistic algorithm $\mathsf{Adv}$, there exists a negligible function $\mathsf{negl}(\cdot)$---that is, a function satisfying $\frac{\mathsf{negl}(\lambda)}{p(\lambda)} = o(1)$ for every polynomial $p$---such that, for every security parameter $\lambda\in\mathbb{N}$, it holds that
    \begin{equation}
        \left|\Pr_{\mathbf{k} \sim \mathcal{K}_\lambda}[\mathsf{Adv}^{f(\mathbf{k}, \cdot)}(\cdot) = 1] - \Pr_{g \sim \mathcal{Y}_\lambda^{\mathcal{X}_\lambda}}[\mathsf{Adv}^{g}(\cdot) = 1]\right| \leq \mathsf{negl}(\lambda),
    \end{equation}
    where the key $\mathbf{k}$ is drawn uniformly at random from $\mathcal{K}_\lambda$ and $g$ is drawn uniformly at random from the set of all functions from $\mathcal{X}_\lambda$ to $\mathcal{Y}_\lambda$.
    Here, we use $\mathsf{Adv}$ with a function superscript to mean the action of the algorithm $\mathsf{Adv}$ when given oracle access to that function.~\qedgen
\end{definition}

Typically, the runtime $t(\cdot)$ of interest in this definition is taken to be polynomial. However, other choices of $t(\cdot)$ are possible.

Next, we formalize the hardness of learning PRFs in the context of online learning:

\begin{theorem}[Hardness of learning PRFs in online learning]\label{theorem:no-efficient-online-learning-prfs}
    Take the security parameter to be $\lambda=n$.
    Let $\mathcal{F}=\{f_\lambda\}_{\lambda\in\mathbb{N}}$ be a PRF that is secure against classical $\mathcal{O}(t(n))$-time adversaries. 
    Let $\Delta:\mathbb{N}\to\mathbb{N}$ be a polynomial and let $p:\mathbb{N}\to\mathbb{N}$ be a function such that $p(n),\ln (\frac{\Delta(n)}{\Delta(n)-1})\leq \mathcal{O}(t(n))$.
    Suppose $\mathcal{G}\subseteq [0,1]^{\{0,1\}^n}$ is a function class such that $\mathcal{F}\subseteq \mathcal{G}$ and $\ln\left(\mathsf{N}_{T} (\mathcal{G}, 1/6, \infty)\right)\leq p(n)$. 
    There exists no classical $\mathcal{O}(t(n))$-time algorithm for properly online learning $\mathcal{G}$ with at most $6(p(n) + \left\lceil \ln (\frac{\Delta(n)}{\Delta(n)-1})\right\rceil)$ many $(1/3)$-mistakes in an online game with $18(p(n) + \left\lceil \ln (\frac{\Delta(n)}{\Delta(n)-1})\right\rceil)$ rounds.
\end{theorem}

\Cref{theorem:no-efficient-online-learning-prfs} in particular says: If a hypothesis class of interest has polynomial sequential metric entropies and contains a class of PRFs secure against polynomial-time adversaries, then that class cannot be efficiently online learned with polynomially many $\mathcal{O}(1)$-mistakes.

\begin{proof}
    Suppose for contradiction that $\mathcal{A}$ is an  $\mathcal{O}(t(n))$-time algorithm for properly online learning $\mathcal{G}$ with at most $m(n)=6(p(n) + \left\lceil \ln (\frac{\Delta(n)}{\Delta(n)-1})\right\rceil)$ many $(1/3)$-mistakes in an online learning game with at least $T(n)=3m(n)$ many rounds. Note that $T(n)\leq\mathcal{O}(t(n))$, by our assumptions on $p$ and $\Delta$.
    We then construct a procedure for distinguishing between a random element of $\mathcal{F}$ and a truly random function with success probability $\geq 1/\Delta(n)$. Namely, we define $\mathcal{D}$, when given query access to a function $f$, to act as follows: First, simulate an online learning game between the learner $\mathcal{A}$ and an adversary that uses an arbitrary sequence of pairwise distinct challenges $x_1,\ldots,x_T\in\{0,1\}^n$ and the corresponding true values $f(x_1),\ldots, f(x_T)$. Second, if $\mathcal{A}$ made at most $m(n)$ many $(1/3)$-mistakes, $\mathcal{D}$ outputs ``$f\in\mathcal{F}$'', otherwise $\mathcal{D}$ outputs ``$f$ truly random''.
    
    Let us analyze the success probability of $\mathcal{D}$. On the one hand, if $f\in\mathcal{F}\subseteq\mathcal{G}$, then the simulated online learning game takes place in a realizable scenario. So $\mathcal{A}$ makes at most $p(n)$ many $(1/3)$-mistakes by assumption, and thus $\mathcal{D}$ correctly outputs ``$f\in\mathcal{F}$''.
    On the other hand, suppose $f$ is chosen as a truly random function from $\{0,1\}^n$ to $\{0,1\}$. 
    As $\mathcal{A}$ is assumed to be proper, we know that for any $1\leq m\leq T(n)$,
    \begin{multline}
        \mathbb{P}_f[\mathcal{A}\textrm{ makes }\leq m\textrm{ many }(1/3)\textrm{-mistakes}]
        \\\leq \mathbb{P}_f[\exists g_1,\ldots,g_{T(n)}\in\mathcal{G}: |\{|f(x_t)- g_t(x_t)|> 1/3\}|\leq m]\, .
    \end{multline}
    
    Next, notice that by the definition of sequential covering with $p=\infty$, if we let $V=V(\mathbf{x})$ be a smallest sequential $\infty$-norm $(1/6)$-cover for $\mathcal{G}$, where $\mathbf{x}$ is a complete rooted binary tree of depth $T(n)$ such that there is a path $\pi$ with $\mathbf{x}_t(\pi)=x_t$ for all $1\leq t\leq T(n)$, 
    then for any $g_1,\ldots,g_{T(n)}\in\mathcal{G}$, there exists a $\mathbf{v}\in V$ such that $|v_t(\pi) - g_t(x_t)| = | v_t(\pi) - g_t(\pi) |\leq 1/6$. So, by the triangle inequality and a union bound,
    \begin{align}
        &\mathbb{P}_f[\exists g_1,\ldots,g_{T(n)}\in\mathcal{G}: |\{|f(x_t)- g_t(x_t)|> 1/3\}|\leq m]\nonumber\\
        &\qquad\leq \mathbb{P}_f[\exists \mathbf{v}\in V: |\{|f(x_t)-\mathbf{v}_t(\pi)|>1/6\}|\leq m]\nonumber \\
        &\qquad\leq \sum_{\mathbf{v}\in V} \mathbb{P}_f[|\{|f(x_t)-\mathbf{v}_t(\pi)|>1/6\}|\leq m]\, .
    \end{align}
    As the $x_1,\ldots,x_T$ are pairwise distinct and $f$ is a random function, the values $f(x_1),\ldots, f(x_T)$ are independent Bernoulli random variables, each with parameter $1/2$. 

    Hence, for any fixed $\mathbf{v}$ and for any $0\leq m\leq \frac{T(n)}{2}$, the probability that at most $m$ of the predictions made by $\mathbf{v}$ are $(1/6)$-mistakes is
    \begin{align}
        \mathbb{P}_f[|\{|f(x_t)- \mathbf{v}_t (\pi)|>1/6\}|\leq m]
        =\mathbb{P}[\mathrm{Binom}(T(n), 1/2)\leq m]
        \leq \exp\left( -\frac{2(\frac{T(n)}{2}-m)^2}{T(n)}\right)\, ,
    \end{align}
    where we have used a Chernoff-Hoeffding bound.
    Plugging this into our previous bound, we see that
    \begin{align}
        \mathbb{P}_f[\mathcal{A}\textrm{ makes }\leq m(n)\textrm{ many }(1/3)\textrm{-mistakes}]
        &\leq |V|\cdot \exp\left( -\frac{2(\frac{T(n)}{2}-m(n))^2}{T(n)}\right)\\
        \nonumber
        &\leq \mathsf{N}_{T(n)} (\mathcal{G}, 1/6, \infty) \cdot \exp\left( -\frac{T(n)}{18}\right)\\
        \nonumber
        &\leq \exp\left( p(n) -\frac{T(n)}{18}\right)\\
        \nonumber
        &\leq 1 - \frac{1}{\Delta (n)}\, ,\,
    \end{align}
    where we have used $T(n)=3m(n)$ and our choice of $m(n)$. Thus, in the case of a truly random $f$, the distinguisher $\mathcal{D}$ correctly outputs ``$f$ truly random'' with probability $\geq \frac{1}{\Delta (n)}$. As $\Delta$ is by assumption polynomial, this means that $\mathcal{D}$ successfully distinguishes between pseudorandom and random with non-negligible success probability.

    Thus, to complete the proof by contradiction (to the pseudorandomness guarantee required in \Cref{definition:prfs}), it remains to argue that $\mathcal{D}$ runs in time $\mathcal{O}(t(n))$. This can be seen as follows. On the one hand, $\mathcal{D}$ plays the online learning game. Here, the learning side takes time $\mathcal{O}(t(n))$ by assumption, and the adversary side takes time $\mathcal{O}(T(n))$, since oracle queries take unit time. On the other hand, $\mathcal{D}$ checks how many mistakes the online learner makes, which takes time $\mathcal{O}(T(n))$. Thus, the overall time taken by $\mathcal{D}$ is $\mathcal{O}(t(n) + T(n))\leq\mathcal{O}(t(n))$.
\end{proof}

\begin{remark}
    Notice that in the proof of \Cref{theorem:no-efficient-online-learning-prfs}, the only property of the challenges $x_1,\ldots,x_T$ that mattered to the argument was that they are chosen to be pairwise distinct. Hence, this line of reasoning can be extended to learner-adversary interactions in which the learner, rather than the adversary, actively chooses pairwise distinct inputs to the unknown function.
\end{remark}

We now apply \Cref{theorem:no-efficient-online-learning-prfs} for our scenario of online learning bounded-complexity quantum channels.
To obtain efficiently implementable PRFs, we make a common hardness assumption, namely the hardness of the \emph{ring learning with errors} (\textsf{RingLWE}) problem \cite{lyubashevsky2010ideal}. 
Note that despite our online learning problems revolving around function classes coming from quantum physics, the classes and learners under consideration are all classical. Therefore, we only assume classical hardness of \textsf{RingLWE}.

\begin{corollary}[Computational hardness]\label{corollary:computational-hardness-bounded-complexity}
    Take the security parameter to be $\lambda = n$.
    Let $\mathsf{CPTP}_{n,G}$ be the class of all $n$-qubit channels of gate complexity at most $G$. 
    If there is no classical polynomial-time algorithm for solving \textsf{RingLWE}, then already for $G=\mathcal{O}(n\mathrm{polylog}(n))$ there exists no (classical) polynomial-time algorithm for properly online learning $\mathsf{CPTP}_{n,G}$ with at most polynomially many $(1/3)$-mistakes, even under the promise that all challenges consist of input states and output effect operators given by rank-$1$ projections on computational basis elements (without an auxiliary system).
\end{corollary}

\begin{proof}
    First, recall from \Cref{corollary:sequential-covering-number-gate-complexity} that the sequential metric entropies of $\mathsf{CPTP}_{n,G}$ satisfy the bound $\ln N_T(\mathsf{CPTP}_{n,G}, 1/6, \infty)\leq \mathcal{O}(G \log(Gn))$, which scales at most polynomially in $n$ if $G$ does.
    Thus, to apply \Cref{theorem:no-efficient-online-learning-prfs}, it remains to argue that $\mathsf{CPTP}_{n,G}$ contains a suitable PRF family. We combine results from prior work to obtain such a PRF from the assumed hardness of \textsf{RingLWE}.
    Namely, Ref.~\cite[Theorem 5.3]{banerjee2012pseudorandom} shows that polynomial-time hardness of (decision-)\textsf{RingLWE} with suitable parameters (see 
    Ref.~\cite{lyubashevsky2010ideal, banerjee2012pseudorandom} for more context and a formal discussion) gives rise to a PRF family $\mathcal{RF}=\{f_\lambda \}_{\lambda\in\mathbb{N}}$ on $m=\omega(\log(\lambda))$-bit inputs that is secure against polynomial-time classical adversaries.
    Moreover, as shown in 
    Ref.~\cite[Lemma 3.16]{arunachalam2021quantum}, every $f_\lambda\in\mathcal{RF}$ can be computed by a $\mathsf{TC}^0$ circuit, that is, by a constant-depth, polynomial-size circuit consisting of AND, OR, NOT, and MAJORITY gates with unbounded fan-in.
    As shown in Ref.~\cite[Proposition 2]{zhao2023learning}, if $f:\{0,1\}^m\to\{0,1\}$ can be implemented by a $\mathsf{TC}^0$ circuit, then there is a quantum circuit on $n=\mathcal{O}(\mathrm{poly}(m))$ qubits with size $\mathcal{O}(n\mathrm{polylog}(n))$ that implements the unitary $U_f$ acting as $U_f\ket{x}\ket{b} = \ket{x}\ket{b\oplus f(x)}$ (ignoring auxiliary qubits). 
    Consequently, choosing the security parameter as $\lambda=n$, every function in $\mathcal{RF}$ can be implemented by a unitary quantum circuit of size $\mathcal{O}(n\mathrm{polylog}(n))$, and thus $G=\mathcal{O}(n\mathrm{polylog}(n))$ suffices to guarantee the inclusion $\mathcal{RF}\subseteq \mathsf{CPTP}_{n,G}$. 
    (Note: Here, we slightly abused notation by not explicitly restricting $\mathsf{CPTP}_{n,G}$ to the input space of $\mathcal{RF}$, which can be embedded into the Boolean hypercube.)

    Hence, for $G=\mathcal{O}(n\mathrm{polylog}(n))$, we can apply \Cref{theorem:no-efficient-online-learning-prfs} with $\mathcal{F}=\mathcal{RF}$, $\mathcal{G}=\mathsf{CPTP}_{n,G}$ (again not writing out the restriction of the input space), $\Delta (n)=3$, $p(n)=\mathcal{O}(G \log(Gn))$, and $t(n)=\mathrm{poly}(n)$. This yields the claimed result.
\end{proof}

As can be seen in the proof of \Cref{corollary:computational-hardness-bounded-complexity}, the hardness assumption on \textsf{RingLWE} is ``only'' used to obtain a PRF class implementable by relatively small circuits. Therefore, one may replace this widely believed assumption about the hardness of a concrete problem~\cite{regev2009lattices, ananth2023revocable, aggarwal2023lattice} by a more abstract cryptographic assumption on the existence of PRFs implementable by small circuits, and the above line of reasoning can still be applied.

\section{Shadow tomography of quantum processes}\label{sec:shadow_tomography}

For quantum states, shadow tomography~\cite{Aar17, badescu2021improved, king2024triply} is the task of using few copies of an unknown state to predict the expectation values of $M$ effect operators, which may be chosen adaptively/adversarially. In this section, we consider the analogous problem for quantum processes, starting with quantum channels and then going to multi-time processes. 

\begin{problem}[Shadow tomography of quantum channels]\label{problem:shadow-tomography}
    Let $\mathcal{N}_{A\to B}\in\mathsf{CPTP}_n$ be an (unknown) $n$-qubit channel, and let $\varepsilon,\delta>0$. When sequentially presented with any adversarially chosen sequence of two-outcome test operators, $E_{A,B}^{(1)}, E_{A,B}^{(2)},\dotsc, E_{A,B}^{(M)}$, for $M \in \mathbb{N}$, return quantities $b_i\in\mathbb{R}$ such that $|b_i - \Tr[E_{A,B}^{(i)}C_{A,B}^{\mathcal{N}}]|\leq \varepsilon$ for all $i\in\{1,2,\dotsc,M\}$ with probability at least $1-\delta$. Do this by querying the channel $k$ times (adaptively or in parallel), with $k$ being as small as possible.~\qedgen
\end{problem}

Embedding classical functions into quantum channels similarly to \Cref{subsection:mistake-lower-bounds-general}, one can see that in the case of general quantum channels, no non-trivial shadow tomography strategy---achieving a query complexity that is simultaneously sublinear in $M$ and polynomial in $n$---is possible. Therefore, we again have to consider restricted classes of channels.
We primarily focus on Pauli channels and Pauli multi-time processes, for which we introduce a shadow tomography scheme via classical adaptive data analysis~\cite{dwork2015adaptivedataanalysis,bassily2016adaptivedataanalysis} that requires few measurements (of the Choi state) of the unknown process. 

\begin{theorem}[Shadow tomography of Pauli channels]\label{theorem:Pauli-shadow-tomography}
    There exists an explicit strategy that solves \Cref{problem:shadow-tomography} for any $n$-qubit Pauli channel using 
    \begin{equation}
        k = \mathcal{O}\left(\frac{\sqrt{n}\log (M) \log^{3/2}((\varepsilon \delta)^{-1}) }{\varepsilon^3}\right)
    \end{equation}
    copies of the channel. The strategy runs in time $\mathrm{poly}(4^n, k)$ per query.
\end{theorem}

\begin{proof}
    Let $A$ and $B$ be $n$-qubit systems, and consider a Pauli channel $\mathcal{P}_{A\to B}$ as in \eqref{eq-pauli_channel}, with error-rate vector $\Vec{p}=(p_{\Vec{z},\Vec{x}})_{\Vec{z},\Vec{x}\in\{0,1\}^n}$. Consider also test operators $E_{A,B}^{(i)}$, for $i\in\{1,2,\dotsc,M\}$. For every such operator, we have
    \begin{align}
        \Tr[E_{A,B}^{(i)}C_{A,B}^{\mathcal{P}}]&=\sum_{\Vec{z},\Vec{x}\in\{0,1\}^n}p_{\Vec{z},\Vec{x}}\Tr[E_{A,B}^{(i)}\Gamma_{A,B}^{\Vec{z},\Vec{x}}]=\mathbb{E}_{(\Vec{z},\Vec{x})\sim\Vec{p}}[e_{\Vec{z},\Vec{x}}^{(i)}],\label{eq:PShadow-equiv-SQ}\,
    \end{align}
    where we have defined $e_{\Vec{z},\Vec{x}}^{(i)}\coloneqq\Tr[E_{A,B}^{(i)}\Gamma_{A,B}^{\Vec{z},\Vec{x}}]$.
    From this, we can see that every desired expectation value is exactly the expectation value of the function $(\Vec{z},\Vec{x})\mapsto e_{\Vec{z},\Vec{x}}^{(i)}=\Tr[E_{A,B}^{(i)}\Gamma_{A,B}^{\Vec{z},\Vec{x}}]$ with respect to the error-rate probability distribution $\Vec{p}$ of the unknown Pauli channel $\mathcal{P}_{A\to B}$. Now, we can obtain samples from the error-rate distribution by performing Bell measurements on the Choi state. Specifically, to obtain one sample, we prepare a $(2n)$-qubit maximally-entangled state $\Phi_{A,A'}=\ketbra{\Phi}{\Phi}_{A,A'}$ (recall \eqref{eq-n_qubit_Bell_states}), send the $A'$ system through the channel, and then measure systems $A$ and $B$ with respect to the Bell basis POVM $\{\Phi^{\Vec{z},\Vec{x}}\}_{\Vec{z},\Vec{x}\in\{0,1\}^n}$. Note that this indeed amounts to measuring the Choi state of the channel with respect to the Bell basis POVM. Then, using the definition of the $n$-qubit Bell states in \eqref{eq-n_qubit_Bell_states}, the probability of obtaining an outcome $(\Vec{z},\Vec{x})$ is given by
    \begin{align}
        \Tr[\Phi_{A,B}^{\Vec{z},\Vec{x}}(\id_A\otimes\mathcal{P}_{A'\to B})(\Phi_{A,A'})]&=\Tr[\Phi_{A,B}^{\Vec{z},\Vec{x}}\Phi_{A,B}^{\mathcal{P}}]\\
        \nonumber
        &=\sum_{\Vec{z}',\Vec{x}'\in\{0,1\}^n}p_{\Vec{z}',\Vec{x}'}\Tr[\Phi_{A,B}^{\Vec{z},\Vec{x}}\Phi_{A,B}^{\Vec{z}',\Vec{x}'}]\\
         \nonumber
        &=p_{\Vec{z},\Vec{x}},
    \end{align}
    for all $\Vec{z},\Vec{x}\in\{0,1\}^n$. 
    
    We can now combine \eqref{eq:PShadow-equiv-SQ} with the ability to sample from the error-rate distribution obtained by Bell measurements to make use of known results in \textit{classical} adaptive data analysis~\cite{dwork2015adaptivedataanalysis,bassily2016adaptivedataanalysis}. In classical data analysis, the goal is to answer a sequence of adaptively chosen queries $q_1,q_2,\dotsc,q_M$ with answers $b_1,b_2,\dotsc,b_M$, such that $|b_i-q_i(\Vec{p})|\leq\varepsilon$ for all $i\in\{1,2,\dotsc,M\}$, given $k$ samples from the underlying (unknown) probability distribution $\Vec{p}$. This setting precisely matches our setting of Pauli channel shadow tomography, by recognizing that the underlying distribution $\Vec{p}$ can be taken to be the error-rate distribution of the unknown Pauli channel, and the queries $q_i$ can be taken to be $q_i(\Vec{p})\equiv\mathbb{E}_{(\Vec{z},\Vec{x})\sim\Vec{p}}[e_{\Vec{z},\Vec{x}}^{(i)}]$, $i\in\{1,2,\dotsc,M\}$. Then, in the regime $M\gg k$, we make direct use of Ref.~\cite[Corollary~6.3]{bassily2016adaptivedataanalysis} to obtain our desired result.
\end{proof}

We highlight that \Cref{theorem:Pauli-shadow-tomography} for restricted/approximate shadow tomography of Pauli channels can be used to perform shadow tomography of arbitrary channels, essentially by applying \Cref{theorem:Pauli-shadow-tomography} to the Pauli twirled version of the channel. The upshot is that our bounds scale with the diamond norm distance between the unknown channel and its corresponding Pauli twirled version.

    \begin{corollary}[Shadow tomography of arbitrary quantum channels]\label{corollary:shadow_tomography_channels}
         Let $\mathcal{N}\in\mathsf{CPTP}_n$ be an arbitrary quantum channel. There exists an explicit strategy that solves Problem~\ref{problem:shadow-tomography} for $\mathcal{N}$ using
         \begin{equation}
            k = \mathcal{O}\left(\frac{\sqrt{n}\log(M) \log^{3/2}(((\varepsilon - \frac{1}{2}\norm{\mathcal{N}-\mathcal{N}^{\mathsf{P}}}_{\diamond}) \delta)^{-1})}{(\varepsilon - \frac{1}{2}\norm{\mathcal{N}-\mathcal{N}^{\mathsf{P}}}_{\diamond})^3}\right)
        \end{equation}
        copies of $\mathcal{N}$, where $\mathcal{N}^{\mathsf{P}}$ is the Pauli twirled version of $\mathcal{N}$ and $\varepsilon>\frac{1}{2}\norm{\mathcal{N}-\mathcal{N}^{\mathsf{P}}}_{\diamond}$.
    \end{corollary}
    
    \begin{proof}
        We perform a Bell measurement on the Choi state of the unknown channel. The measurement probabilities define the error-rate vector of the Pauli-twirled version of the channel, based on the developments in Appendix~\ref{sec-Pauli_twirl_quantum_channel}. Then, we make use of the triangle inequality as follows, for an arbitrary $b\in\mathbb{R}$ and arbitrary channel test operator, to get
        \begin{align}
            \Abs{b-\Tr[E_{A,B}C_{A,B}^{\mathcal{N}}]}&\leq \Abs{b-\Tr[E_{A,B}C_{A,B}^{\mathcal{N}^{\mathsf{P}}}]}+\Abs{\Tr[E_{A,B}C_{A,B}^{\mathcal{N}^{\mathsf{P}}}]-\Tr[E_{A,B}C_{A,B}^{\mathcal{N}}]}\\
             \nonumber
            &\leq \Abs{b-\Tr[E_{A,B}C_{A,B}^{\mathcal{N}^{\mathsf{P}}}]}+\sup_{E_{A,B}}\Abs{\Tr[E_{A,B}C_{A,B}^{\mathcal{N}^{\mathsf{P}}}]-\Tr[E_{A,B}C_{A,B}^{\mathcal{N}}]}\\
             \nonumber
            &=\Abs{b-\Tr[E_{A,B}C_{A,B}^{\mathcal{N}^{\mathsf{P}}}]}+\frac{1}{2}\norm{\mathcal{N}-\mathcal{N}^{\mathsf{P}}}_{\diamond},
             \nonumber
        \end{align}
        where in the final equality we made use of the fact that (see, e.g., \cite[Section~4]{Wat09})
        \begin{align}
            \frac{1}{2}\norm{\mathcal{N}-\mathcal{M}}_{\diamond}&=\sup_{\substack{E_{A,B}\geq 0\\\sigma_A\geq 0}}\Big\{\Tr[E_{A,B}(C_{A,B}^{\mathcal{N}}-C_{A,B}^{\mathcal{M}})]:E_{A,B}\leq\sigma_A\otimes\mathbbm{1}_B,\,\Tr[\sigma_A]=1\Big\}\\
            &=\sup_{\substack{E_{A,B}\geq 0\\\sigma_A\geq 0}}\Big\{\Abs{\Tr[E_{A,B}(C_{A,B}^{\mathcal{N}}-C_{A,B}^{\mathcal{M}})]}:E_{A,B}\leq\sigma_A\otimes\mathbbm{1}_B,\,\Tr[\sigma_A]=1\Big\},
             \nonumber
        \end{align}
        for arbitrary channels $\mathcal{N}$ and $\mathcal{M}$. Finally, as we assume $\varepsilon - \frac{1}{2}\norm{\mathcal{N}-\mathcal{N}^{\mathsf{P}}}_{\diamond} > 0$, can run the Pauli channel shadow tomography from \Cref{theorem:Pauli-shadow-tomography} with accuracy parameter  $\tilde{\varepsilon} = \varepsilon - \frac{1}{2}\norm{\mathcal{N}-\mathcal{N}^{\mathsf{P}}}_{\diamond}$ to get the desired approximation guarantee $\Abs{b-\Tr[E_{A,B}C_{A,B}^{\mathcal{N}}]} \leq \varepsilon$ with the claimed sample complexity bounds.
    \end{proof}

\subsection{Shadow tomography of multi-time processes}

    By analogy with the shadow tomography problem for quantum channels, we can formulate the problem of shadow tomography for multi-time quantum processes.

    \begin{problem}[Shadow tomography of multi-time quantum processes]\label{problem:shadow_tomography_multi_time_processes}
        Let $r\in\{1,2,\dotsc\}$, let $N\in\mathsf{COMB}_r$ be a comb operator corresponding to a multi-time quantum process with $r$ time steps, and let $\varepsilon,\delta>0$. When sequentially presented with any adversarially chosen sequence of two-outcome multi-time test operators $E^{(1)},E^{(2)},\dotsc,E^{(M)}$, for $M\in\{1,2,\dotsc\}$, return quantities $b_i\in\mathbb{R}$ such that $|b_i-\Tr[EN]|\leq\varepsilon$ for all $i\in\{1,2,\dotsc,M\}$ with probability at least $1-\delta$. Do this by querying the process $k$ times (adaptively or in parallel), with $k$ being as small as possible.~\qedgen
    \end{problem}

    By following arguments similar to those in the proof of Corollary~\ref{corollary:shadow_tomography_channels}, we can prove a shadow tomography result for arbitrary multi-time quantum processes as follows. In particular, this involves introducing the idea of Pauli-twirling a multi-time process, which we illustrate in Figure~\ref{fig:comb_twirl}(a).

    \begin{corollary}
        Let $r\in\{1,2,\dotsc\}$, let $N\in\mathsf{COMB}_r$ be a comb operator corresponding to a multi-time quantum process with $r$ time steps. There exists an explicit strategy that solves Problem~\ref{problem:shadow_tomography_multi_time_processes} for $N$ using
        \begin{equation}
            k = \mathcal{O}\left(\frac{\sqrt{nr}\log(M) \log^{3/2}(((\varepsilon - \frac{1}{2}\norm{N-N^{\mathsf{P}}}_{\diamond r}) \delta)^{-1})}{(\varepsilon - \frac{1}{2}\norm{N-N^{\mathsf{P}}}_{\diamond r})^3}\right),
        \end{equation}
        where $N^{\mathsf{P}}$ is the comb operator corresponding to the Pauli-twirled version of the multi-time process (see Figure~\ref{fig:comb_twirl}(a) for a depiction), and $\varepsilon>\frac{1}{2}\norm{N-N^{\mathsf{P}}}_{\diamond r}$.
    \end{corollary}

    \begin{figure}
        \centering
        \includegraphics[width=0.95\textwidth]{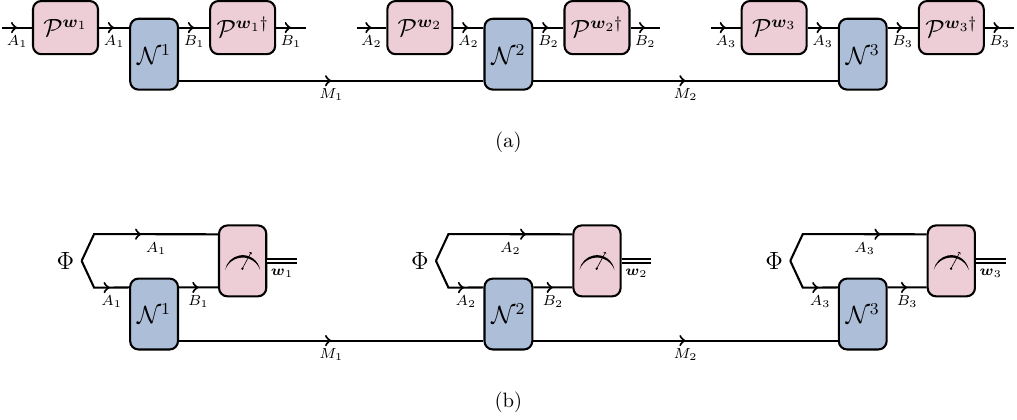}
        \caption{\textbf{Twirling of multi-time quantum processes.} (a) A ``time-local'' Pauli twirl of a multi-time quantum process with $r$ time steps consists of independently applying a random Pauli channel $\mathcal{P}^{\Vec{w}_k}(\cdot)\coloneqq P^{\Vec{w}_k}(\cdot)P^{\Vec{w}_k\dagger}$, where $\Vec{w}_k\equiv (\Vec{z}_k,\Vec{x}_k)\in\{0,1\}^n\times\{0,1\}^n$, to the input and output of every time step $k\in\{1,2,\dotsc,r\}$. (b) After twirling, the process is characterized by an error-rate probability vector, in the same way as Pauli channels. This error-rate vector can be obtained via time-local Bell measurements, as shown. The outcomes of the measurements are $\Vec{w}_1\equiv(\Vec{z}_1,\Vec{x}_1),\Vec{w}_2\equiv(\Vec{z}_2,\Vec{x}_2),\dotsc,\Vec{w}_r\equiv(\Vec{z}_r,\Vec{x}_r)$.}
        \label{fig:comb_twirl}
    \end{figure}

    \begin{proof}
        We proceed by performing ``time-local'' Bell measurements on the multi-time process; see Figure~\ref{fig:comb_twirl}(b). This means that, for every time step, we prepare a Bell state, send one-half of it through the process, and then measure the output system and the other-half of the Bell state in the $(2n)$-qubit Bell basis. Doing this once for each time step leads to measurement outcomes $\Vec{w}_1\equiv(\Vec{z}_1,\Vec{x}_1),\Vec{w}_2\equiv(\Vec{z}_2,\Vec{x}_2),\dotsc,\Vec{w}_r\equiv(\Vec{z}_r,\Vec{x}_r)$.
        The probability of any such collection of measurement outcomes is given by
        \begin{equation}
            p_{\Vec{z}_1,\Vec{x}_1,\Vec{z}_2,\Vec{x}_2,\dotsc,\Vec{z}_r,\Vec{x}_r}=\frac{1}{2^{nr}}\Tr\!\left[\left(\Phi_{A_1,B_1}^{\Vec{z}_1,\Vec{x}_1}\otimes\Phi_{A_2,B_2}^{\Vec{z}_2,\Vec{x}_2}\otimes\dotsb\otimes\Phi_{A_r,B_r}^{\Vec{z}_r,\Vec{x}_r}\right)N_{A_1,B_1,A_2,B_2,\dotsc,A_r,B_r}\right],
        \end{equation}
        which is due to the fact that applying one-half of a maximally-entangled state to every input of the process defines the Choi state of the process, which is equal to $\frac{1}{2^{nr}}N$; see Appendix~\ref{sec-multi_time_processes}.
        
        Let us now consider the Pauli-twirl of the process, as depicted in Figure~\ref{fig:comb_twirl}(a). By combining \Cref{lem-Choi_rep_twirling} and \Cref{prop-qudit_Pauli_channel_twirling}, we find that the comb operator $N^{\mathsf{P}}$ for the Pauli-twirled process is equal to
        \begin{align}
            N_{A_1,B_1,\dotsc,A_r,B_r}^{\mathsf{P}}&=(\mathcal{S}_{\mathsf{P}}\otimes\dotsb\otimes\mathcal{S}_{\mathsf{P}})(N_{A_1,B_1,\dotsc,A_r,B_r})\\
            &=\sum_{\Vec{z}_1,\Vec{x}_1,\dotsc,\Vec{z}_r,\Vec{x}_r\in\{0,1\}^n} \Tr\!\left[\left(\Phi_{A_1,B_1}^{\Vec{z}_1,\Vec{x}_1}\otimes\Phi_{A_2,B_2}^{\Vec{z}_2,\Vec{x}_2}\otimes\dotsb\otimes\Phi_{A_r,B_r}^{\Vec{z}_r,\Vec{x}_r}\right)N_{A_1,B_1,A_2,B_2\dotsc,A_r,B_r}\right]\nonumber\\
            &\qquad\qquad\qquad\qquad\qquad\qquad\times\Phi_{A_1,B_1}^{\Vec{z}_1,\Vec{x}_1}\otimes\Phi_{A_2,B_2}^{\Vec{z}_2,\Vec{x}_2}\otimes\dotsb\otimes\Phi_{A_r,B_r}^{\Vec{z}_r,\Vec{x}_r}\\
             \nonumber&=\sum_{\Vec{z}_1,\Vec{x}_1,\Vec{z}_2,\Vec{x}_2,\dotsc,\Vec{z}_r,\Vec{x}_r\in\{0,1\}^n}p_{\Vec{z}_1,\Vec{x}_1,\Vec{z}_2,\Vec{x}_2,\dotsc,\Vec{z}_r,\Vec{x}_r}\Gamma_{A_1,B_1}^{\Vec{z}_1,\Vec{x}_1}\otimes\Gamma_{A_2,B_2}^{\Vec{z}_2,\Vec{x}_2}\otimes\dotsb\otimes\Gamma_{A_r,B_r}^{\Vec{z}_r,\Vec{x}_r}.
             \nonumber
        \end{align}
        We now observe that the comb operator for the Pauli-twirled process can be thought of as simply the Choi representation of an $(nr)$-qubit Pauli channel. As such, we can apply \Cref{theorem:Pauli-shadow-tomography}. Furthermore, for an arbitrary $b\in\mathbb{R}$ and an arbitrary test operator $E$, we have
        \begin{align}
            \Abs{b-\Tr[EN]}&\leq\Abs{b-\Tr[EN^{\mathsf{P}}]}+\Abs{\Tr[EN^{\mathsf{P}}]-\Tr[EN]}\\
             \nonumber&\leq \Abs{b-\Tr[EN^{\mathsf{P}}]}+\sup_{E}\Abs{\Tr[EN^{\mathsf{P}}]-\Tr[EN]}\\
              \nonumber
            &=\Abs{b-\Tr[EN]}+\frac{1}{2}\norm{N-N^{\mathsf{P}}}_{\diamond r},
             \nonumber
        \end{align}
        where for the final equality we made use of the fact that
        \begin{align}
            \frac{1}{2}\norm{N-M}_{\diamond r}&=\sup_{E,S\geq 0}\Big\{\Tr[E(N-M)]:E\leq S\otimes\mathbbm{1}_{B_r},\,S\in\mathsf{COMB}_r^{\ast}\Big\}\\
            &=\sup_{E,S\geq 0}\Big\{\Abs{\Tr[E(N-M)]}:E\leq S\otimes\mathbbm{1}_{B_r},\,S\in\mathsf{COMB}_r^{\ast}\Big\},
             \nonumber
        \end{align}
        for arbitrary $N,M\in\mathsf{COMB}_r$. To conclude, we again invoke the Pauli channel shadow tomography of \Cref{theorem:Pauli-shadow-tomography}  $\varepsilon$ with accuracy $\varepsilon - \frac{1}{2}\norm{N-N^{\mathsf{P}}}_{\diamond r} > 0$ to obtain the desired approximation error $\Abs{b-\Tr[EN]} \leq \varepsilon$ with the claimed sample complexity bound.
    \end{proof}

\newpage

\bibliographystyle{unsrtnat}
\bibliography{refs_main} 

\begin{appendix}

\section{From qubits to qudits}\label{sec:Pauli}

    Although all of our results have been phrased for $n$-qubit systems and channels, they apply equally well to qudit systems and qudit channels. This essentially amounts to replacing all of the Pauli operators $P^{\Vec{z},\Vec{x}}$ with the \textit{Heisenberg--Weyl operators}, sometimes known as the \textit{qudit/generalized Pauli operators}. These operators are defined as ~\cite{Wat18_book}
    \begin{align}
        W^{z,x}&\coloneqq Z(z)X(x),\quad z,x\in\{0,1,\dotsc,d-1\},\\
        \nonumber
        Z(z)&\coloneqq \sum_{k=0}^{d-1} \e^{\frac{2\pi\I kz}{d}}\ketbra{k}{k},\\
         \nonumber
        X(x)&\coloneqq \sum_{k=0}^{d-1} \ketbra{k+x}{k},
         \nonumber
    \end{align}
    where the addition in the definition of $X(x)$ is performed modulo $d$, with $d\in\{2,3,\dotsc\}$. These operators are unitary and orthogonal, i.e.,
    \begin{equation}
        (W^{z,x})^{\dagger}W^{z,x}=W^{z,x}(W^{z,x})^{\dagger}=\mathbbm{1},\quad \inner{W^{z,x}}{W^{z',x'}}=d\delta_{z,z'}\delta_{x,x'},
    \end{equation}
    for $z,x,z',x'\in\{0,1,\dotsc,d-1\}$. Consequently, they form a basis for the vector space $\Lin(\mathbb{C}^d)$ of linear operators acting on $\mathbb{C}^d$. The \textit{qudit Bell states} are then defined as
    \begin{equation}\label{eq-qudit_Bell_states}
        \Phi^{z,x}\coloneqq\ketbra{\Phi^{z,x}}{\Phi^{z,x}},\quad\ket{\Phi^{z,x}}\coloneqq(\mathbbm{1}_d\otimes W^{z,x})\ket{\Phi},\quad\ket{\Phi}=\frac{1}{\sqrt{d}}\sum_{k=0}^{d-1}\ket{k,k}.
    \end{equation}
    The qudit Bell state vectors $\ket{\Phi^{z,x}}$ form an orthonormal basis for $\mathbb{C}^d\otimes\mathbb{C}^d$, and the qudit Bell states $\Phi^{z,x}$ form a POVM, meaning that $\sum_{z,x=0}^{d-1}\Phi^{z,x}=\mathbbm{1}_d\otimes\mathbbm{1}_d$.

    The qudit/generalized Pauli channels are defined analogously to $n$-qubit Pauli channels as
    \begin{equation}
        \mathcal{N}(\rho)=\sum_{z,x=0}^{d-1} p(z,x) W^{z,x} \rho W^{z,x\dagger},
    \end{equation}
    where $p(z,x)\in[0,1]$ and $\sum_{z,x=0}^{d-1}p(z,x)=1$. The Choi representations of these channels have the form
    \begin{equation}\label{eq-qudit_Pauli_channel_Choi_rep}
        C(\mathcal{N})=\sum_{z,x=0}^{d-1}p(z,x)\Gamma^{z,x}, \quad \Gamma^{z,x}\coloneqq(\mathbbm{1}_d\otimes W^{z,x})\ketbra{\Gamma}{\Gamma}(\mathbbm{1}_d\otimes W^{z,x})^{\dagger}.
    \end{equation}
    It follows from this (and using orthonormality of the Bell states) that the error rates $p(z,x)$ can be obtained as
    \begin{equation}\label{eq-qudit_Pauli_channel_error_vector}
        p(z,x)=\frac{1}{d}\Tr[\Phi^{z,x}C(\mathcal{N})],
    \end{equation}
    for all $z,x\in\{0,1,\dotsc,d-1\}$. In other words, we can recover the error rates by performing the qudit Bell basis measurement on the Choi state $\frac{1}{d}C(\mathcal{N})$ of the channel. Conversely, every positive semi-definite bipartite operator $Y\in\Lin(\mathbb{C}^d\otimes\mathbb{C}^d)$ of the form \eqref{eq-qudit_Pauli_channel_Choi_rep} corresponds to a Pauli channel with error rates gives by $p(z,x)=\frac{1}{d}\Tr[\Phi^{z,x}Y]$ for all $z,x\in\{0,1,\dotsc,d-1\}$.

    For a quantum channel $\mathcal{N}:\Lin(\mathbb{C}^d)\to\Lin(\mathbb{C}^d)$, $d\in\{2,3,\dotsc\}$, we define its (qudit) Pauli-twirled version as
    \begin{equation}\label{eq-qudit_Pauli_twirled_channel}
        \mathcal{N}^{\mathsf{W}}(\rho)=\frac{1}{d^2}\sum_{z,x=0}^{d-1}W^{z,x\dagger}\mathcal{N}(W^{z,x}\rho W^{z,x\dagger})W^{z,x},
    \end{equation}
    where the superscript ``$\mathsf{W}$'' in $\mathcal{N}^{\mathsf{W}}$ refers to the set $\mathsf{W}\coloneqq\{W^{z,x}:z,x\in\{0,1,\dotsc,d-1\}\}$ of qudit Pauli operators. In Proposition~\ref{prop-qudit_Pauli_channel_twirling}, we show that the twirled channel $\mathcal{N}^{\mathsf{W}}$ is indeed a Pauli channel.

\section{Multi-time quantum processes}\label{sec-multi_time_processes}

In this section, we provide some background on multi-time quantum processes. Such objects are also known as ``quantum strategies''~\cite{GW07} and ``quantum combs''~\cite{CDP09}, and they also constitute specific examples of quantum causal networks~\cite{tucci1995quantumbayesiannets,beckman2001causalquantumoperations,eggeling2002semicausal,CDP09,oreshkov2012noncausal,costa2016quantumcausalmodelling,allen2017quantumcausalmodels} and quantum channels with memory~\cite{kretschmann2005memorychannel}. They also provide a model for discrete-time non-Markovian quantum stochastic processes~\cite{pollock2018nonmarkovian,MM20}, and it is in this context that they are known as ``multi-time quantum processes''---see also Refs.~\cite{aharanov2009multipletimestates,berk2021multitimeprocesses}.

\begin{figure}
    \centering
    \includegraphics{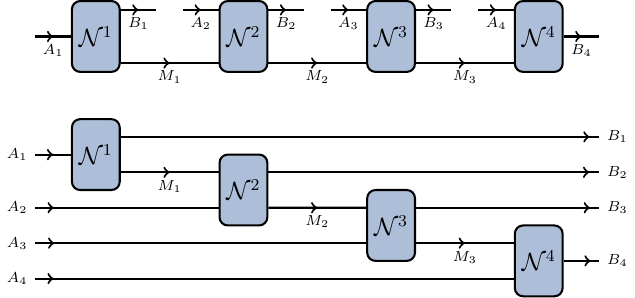}
    \caption{(Top) A multi-time quantum process with $r=4$ time steps. The input systems are $A_1,\dotsc,A_4$, the output systems are $B_1,\dotsc,B_4$, and the memory systems are $M_1,M_2,M_3$. (Bottom) Every multi-time process is associated with the channel $\mathcal{N}^{[r]}$, obtained by collapsing the causal ordering of the inputs and outputs.}
    \label{fig:comb_to_memory_channel}
\end{figure}

\subsection{Definitions and basic properties}

A general multi-time process is depicted in Figure~\ref{fig:comb_to_memory_channel}~(top) as the comb object in blue with $r=4$ time steps. The input systems are $A_1,A_2,\dotsc,A_r$, the output systems are $B_1,B_2,\dotsc,B_r$, and the memory systems are $M_1,M_2,\dotsc,M_{r-1}$. Every multi-time process is associated with a quantum channel $\mathcal{N}^{[r]}:\Lin(\mathcal{H}_{A_1}\otimes\mathcal{H}_{A_2}\otimes\dotsb\otimes\mathcal{H}_{A_r})\to\Lin(\mathcal{H}_{B_1}\otimes\mathcal{H}_{B_2}\otimes\dotsb\otimes\mathcal{H}_{B_r})$, defined by concatenating the maps $\mathcal{N}^i$ in the manner shown in Figure~\ref{fig:comb_to_memory_channel}~(bottom). As shown in Refs.~\cite{GW07,CDP09}, due to the causal constraints, the Choi representation of the channel $\mathcal{N}^{[r]}$, which completely characterizes the multi-time process, is in one-to-one correspondence with a set of so-called \textit{comb operators}, which have a very specific structure based on the causal ordering of the individual elements of the process.

\begin{definition}[Comb operator for multi-time quantum process]\label{def-multi_time_process}
    Every multi-time quantum process with $r\in\mathbb{N}$ time steps, as depicted in Figure~\ref{fig:comb_to_memory_channel}, is represented by a positive semi-definite \textit{comb operator} $N_r$, defined to be the Choi representation of the quantum channel associated with the process (see Figure~\ref{fig:comb_to_memory_channel}). For every comb operator $N_r$, there exist positive semi-definite operators $N_k\in\Lin(\mathcal{H}_{A,B}^{(k)})$, $k\in\{1,2,\dotsc,r-1\}$, such that the following constraints are satisfied:
    \begin{align}
        & N_r\geq 0,\, \Tr_{B_r}[N_r]=N_{r-1}\otimes\mathbbm{1}_{A_r},\label{eq-comb_strategy_constraints_1}\\[1ex]
        & N_{k}\geq 0,\, \Tr_{B_k}[N_{k}]=N_{k-1}\otimes\mathbbm{1}_{A_k},\quad k\in\{2,3,\dotsc,r-1\}, \label{eq-comb_strategy_constraints_2}\\[1ex]
        & N_{1}\geq 0,\,\Tr_{B_1}[N_{1}]=\mathbbm{1}_{A_1}.\label{eq-comb_strategy_constraints_3}
    \end{align}
    A positive semi-definite operator $P\in\Lin(\mathcal{H}_{A,B}^{(r)})$ is the comb operator of an $r$-step multi-time quantum process with input systems $A_1,\dotsc,A_r$ and output systems $B_1,\dotsc,B_r$ if and only if there exists a set $\{N_k\}_{k=1}^r$ of positive semi-definite operators such that $P=N_r$ and the constraints in \eqref{eq-comb_strategy_constraints_1}--\eqref{eq-comb_strategy_constraints_3} are satisfied. We let $\mathsf{COMB}_r(A_1,\dotsc, A_r;B_1,\dotsc B_r)$ denote the convex set of all operators in $\Lin(\mathcal{H}_{A,B}^{(r)})$ representing $r$-step multi-time quantum processes with input systems $A_1,\dotsc,A_r$ and output systems $B_1,\dotsc,B_r$.~\qedgen
\end{definition}

To understand Definition~\ref{def-multi_time_process}, let us see how the constraints on the comb operators in \eqref{eq-comb_strategy_constraints_1}--\eqref{eq-comb_strategy_constraints_3} are manifested in the Choi representation of the channel $\mathcal{N}^{[4]}$ depicted in Figure~\ref{fig:comb_to_memory_channel}~(bottom). By definition, we have
\begin{multline}
    N_4\equiv C(\mathcal{N}^{[4]})=\left(\mathcal{N}_{M_3A_4'\to B_4}^4\circ\mathcal{N}_{M_2A_3'\to M_3B_3}^3\right.\\\left.\circ\mathcal{N}_{M_1A_2'\to M_2B_2}^2\circ\mathcal{N}_{A_1'\to M_1B_1}^1\right)(\Gamma_{A_1A_1'}\otimes\Gamma_{A_2A_2'}\otimes\Gamma_{A_3A_3'}\otimes\Gamma_{A_4A_4'}). \label{eq-multi_time_process_Choi_rep}
\end{multline}
It is clear that $N_4\in\Lin(\mathcal{H}_{A,B}^{(4)}$ is positive semi-definite. Now, observe that
\begin{align}
    &\Tr_{B_4}[N_4]\nonumber\\
    &\quad=\Tr_{M_3A_4'}\!\left[\left(\mathcal{N}_{M_2A_3'\to M_3B_3}^3\circ\mathcal{N}_{M_1A_2'\to M_2B_2}^2\circ\mathcal{N}_{A_1'\to M_1B_1}^1\right)(\Gamma_{A_1A_1'}\otimes\Gamma_{A_2A_2'}\otimes\Gamma_{A_3A_3'}\otimes\Gamma_{A_4A_4'})\right]\\
     \nonumber
     &\quad=\Tr_{M_3}\!\left[ \left(\mathcal{N}_{M_2A_3'\to M_3B_3}^3\circ\mathcal{N}_{M_1A_2'\to M_2B_2}^2\circ\mathcal{N}_{A_1'\to M_1B_1}^1\right)(\Gamma_{A_1A_1'}\otimes\Gamma_{A_2A_2'}\otimes\Gamma_{A_3A_3'})\right]\otimes\underbrace{\Tr_{A_4'}[\Gamma_{A_4A_4'}]}_{\mathbbm{1}_{A_4}}\\
     \nonumber
    &\quad= N_3\otimes\mathbbm{1}_{A_4},
     \nonumber
\end{align}
which is precisely the constraint in \eqref{eq-comb_strategy_constraints_1}, where in the last line we let
\begin{equation}
    N_3\equiv \Tr_{M_3}\!\left[ \left(\mathcal{N}_{M_2A_3'\to M_3B_3}^3\circ\mathcal{N}_{M_1A_2'\to M_2B_2}^2\circ\mathcal{N}_{A_1'\to M_1B_1}^1\right)(\Gamma_{A_1A_1'}\otimes\Gamma_{A_2A_2'}\otimes\Gamma_{A_3A_3'})\right].
\end{equation}
In a similar manner, we find that
\begin{align}
    \Tr_{B_3}[N_3]&=N_2\otimes\mathbbm{1}_{A_3},\\
    N_2&:= \Tr_{M_2}\!\left[ \left(\mathcal{N}_{M_1A_2'\to M_2B_2}^2\circ\mathcal{N}_{A_1'\to M_1B_1}^1\right)(\Gamma_{A_1A_1'}\otimes\Gamma_{A_2A_2'})\right],\\[1ex]
    \Tr_{B_2}[N_2]&=N_1\otimes\mathbbm{1}_{A_2},\\
    N_1&:= \Tr_{M_1}\!\left[ \left(\mathcal{N}_{A_1'\to M_1B_1}^1\right)(\Gamma_{A_1A_1'})\right], \\[1ex]
    \Tr_{B_1}[N_1]&=\mathbbm{1}_{A_1},
\end{align}
which reproduces the constraints in \eqref{eq-comb_strategy_constraints_2} and \eqref{eq-comb_strategy_constraints_3}.

We can also consider a multi-time process with a measurement in the final time step, sometimes called a ``measuring strategy''.

\begin{definition}[Comb operators for multi-time quantum process with measurement]\label{def-comb_operators}
    Every multi-time quantum process with measurement, consisting of $r\in\mathbb{N}$ time steps and input systems $A_1,\dotsc,A_r$ and output systems $B_1,\dotsc,B_r$, is represented by a set $\{N_{r;x}\}_{x}$ of positive semi-definite operators $N_{r;x}\in\Lin(\mathcal{H}_{A,B}^{(r)})$, such that $\sum_{x}N_{r;x}\in\mathsf{COMB}_r(A_1,\dotsc,A_r;B_1,\dotsc,B_r)$. A finite set $\{S_x\}_{x}$ of positive semi-definite operators in $\Lin(\mathcal{H}_{A,B}^{(r)})$ defines an $r$-step multi-time quantum process with measurement, with input systems $A_1,\dotsc,A_r$ and output systems $B_1,\dotsc,B_r$, if and only if $\sum_x S_x\in\mathsf{COMB}_r(A_1,\dotsc,A_r;B_1,\dotsc,B_r)$.~\qedgen
\end{definition}

\begin{figure}
    \centering
    \includegraphics{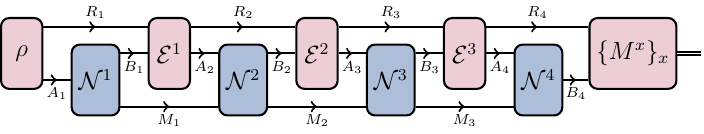}
    \caption{Concatenation of a multi-time process $\mathcal{N}^{[4]}$ with $r=4$ time steps, represented by the blue quantum comb, with a corresponding tester $\mathcal{E}^{[4]}$, which is represented by the red quantum comb. The operations $\mathcal{N}^i$ and $\mathcal{E}^i$ can be arbitrary quantum channels, and they can also more generally be arbitrary Hermiticity-preserving maps.}
    \label{fig:quantum_strategy_meas_co_strategy}
\end{figure}

A measurement, or \textit{tester}, for a multi-time process is another multi-time process that consists of an input state $\rho$ in the first time step and a measurement in the final time step, as shown by the red comb object in Figure~\ref{fig:quantum_strategy_meas_co_strategy} for $r=4$ time steps. Testers are sometimes called ``measuring co-strategies''.

\begin{definition}[Comb operators for multi-time quantum tester]\label{def-testers}
    Every multi-time quantum tester with $r\in\mathbb{N}$ time steps, consisting of input systems $B_1,\dotsc,B_r$ and output systems $A_1,\dotsc,A_r$, is represented by a set $\{E_{r;x}\}_{x}$ of positive semi-definite operators, with $E_{r;x}\in\Lin(\mathcal{H}_{A,B}^{(r)})$, such that $\sum_x E_{r;x}=S_r\otimes\mathbbm{1}_{B_r}$, with $S_r\in\mathsf{COMB}^{\ast}_r(A_1,\dotsc,A_r;B_1,\dotsc,B_{r-1})$. Here, $\mathsf{COMB}^{\ast}_r(A_1,\dotsc,A_r;B_1,\dotsc,B_{r-1})\coloneqq\mathsf{COMB}_r(\varnothing,B_1,\dotsc,B_{r-1};A_1,\dotsc,A_r)$ is the set of all multi-time processes in which the first input system is trivial; see \eqref{eq-costrategies}. A finite set $\{T_x\}_{x}$ of positive semi-definite operators in $\Lin(\mathcal{H}_{A,B}^{(r)})$ defines an $r$-step multi-time quantum tester, with input systems $B_1,\dotsc,B_r$ and output systems $A_1,\dotsc,A_r$, if and only if there exists $E_r\in\mathsf{COMB}^{\ast}_r(A_1,\dotsc,A_r;B_1,\dotsc,B_{r-1})$ such that $\sum_x T_x=E_r\otimes\mathbbm{1}_{B_r}$.~\qedgen
\end{definition}

\subsection{Norms}\label{sec-strategy_norms}

Norms for multi-time quantum processes have been defined in Refs.~\cite{CDP08,Gut12}. Here, we follow the presentation in Ref.~\cite{Gut12}.

\begin{definition}[Strategy norm and its dual~\cite{Gut12}]\label{def-strategy_norm}
    Let $r\in\mathbb{N}$. For every Hermitian operator $H\in\Lin(\mathcal{H}_{A,B}^{(r)})$, we define the \textit{strategy norm} $\norm{H}_{\diamond r}$ and its dual $\norm{H}_{\diamond r}^{\ast}$ as 
    \begin{align}
        \norm{H}_{\diamond r}&\coloneqq\sup\Big\{\Tr[H(T_0-T_1)]: T_0,T_1\geq 0,\,T_0+T_1=S\otimes\mathbbm{1}_{B_r},\,S\in\mathsf{COMB}_r^{\ast}\Big\}\label{eq-strategy_norm}\\
        &=\inf\Big\{t:t\geq 0,\,-tN\leq H\leq tN,\,N\in\mathsf{COMB}_r\Big\} \label{eq-strategy_norm_dual_formulation_0},\\[1ex]
        \norm{H}_{\diamond r}^{\ast}&\coloneqq\sup\Big\{\Tr[H(N_0-N_1)]: N_0,N_1\geq 0,\,N_0+N_1\in\mathsf{COMB}_r\Big\}\label{eq-strategy_norm_dual}\\
         \nonumber
        &=\inf\Big\{t:t\geq 0,\,-tS\otimes\mathbbm{1}_{B_r}\leq H\leq tS\otimes\mathbbm{1}_{B_r},\,S\in\mathsf{COMB}^{\ast}_r\Big\}.
    \end{align}
    For every Hermiticity-preserving linear map $\mathcal{N}^{[r]}:\Lin(\mathcal{H}_{A_1}\otimes\dotsb\otimes\mathcal{H}_{A_r})\to\Lin(\mathcal{H}_{B_1}\otimes\dotsb\otimes\mathcal{H}_{B_r})$, in particular those corresponding to multi-time processes, its strategy $r$-norm is defined via the Choi representation as $\norm{\mathcal{N}^{[r]}}_{\diamond r}\coloneqq\norm{C(\mathcal{N}^{[r]})}_{\diamond r}$.~\qedgen
\end{definition}

The norms $\norm{\cdot}_{\diamond r}$ and $\norm{\cdot}_{\diamond r}^{\ast}$ are (H\"{o}lder) dual to each other, and the proof of this can be found in Ref.~\cite{Gut12}. These norms should be thought of as generalizations of the trace norm and its dual (the spectral/operator norm) for quantum states and the diamond norm and its dual for quantum channels. Indeed, with respect to Figure~\ref{fig:quantum_strategy_meas_co_strategy}, in the case $r=1$ and $d_{A_1}=1$, it holds that $\norm{\cdot}_{\diamond 1}\equiv\norm{\cdot}_1$ and $\norm{\cdot}_{\diamond 1}^{\ast}\equiv\norm{\cdot}_{\infty}$, where we note that
\begin{align}
    \norm{H}_1&=\sup\Big\{\Tr[H(M_1-M_2)]:M_1,M_2\geq 0,\,M_2\geq 0,\,M_1+M_2\leq\mathbbm{1}\Big\}\label{eq-trace_norm_primal}\\
    &=\inf\Big\{t:t\geq 0,\,-t\sigma\leq H\leq t\sigma,\,\sigma\geq 0,\Tr[\sigma]=1\Big\}\label{eq-trace_norm_dual},\\[1ex]
    \norm{H}_{\infty}&=\sup\Big\{\Tr[H(M_1-M_2)]:M_1,M_2\geq 0,\,\Tr[M_1+M_2]\leq 1\Big\}\label{eq-infty_norm_primal}\\
    &=\inf\Big\{t:t\geq 0,-t\mathbbm{1}\leq H\leq t\mathbbm{1}\Big\},\label{eq-infty_norm_dual}
\end{align}
for every Hermitian operator $H$. Similarly, in the case $r=1$ and arbitrary dimension for the system $A_1$, we have that $\norm{\cdot}_{\diamond 1}\equiv\norm{\cdot}_{\diamond}$. The norm $\norm{\cdot}_{\diamond 1}^{\ast}$, i.e., the H\"{o}lder dual to the diamond norm, has been considered before in Ref.~\cite[Section~5.3]{gutoski2010thesis}. This dual norm is the relevant norm when considering observables for quantum channels, analogous to the role that the spectral norm $\norm{\cdot}_{\infty}$ has for observables for states. Using \eqref{eq-strategy_norm_dual}, it is straightforward to see that the diamond norm dual is given by the following primal-dual pair of semi-definite programs, where $H\in\Lin(\mathcal{H}_A\otimes\mathcal{H}_B)$ is Hermitian:
\begin{align}
    \norm{H}_{\diamond 1}^{\ast}&\coloneqq\sup\Big\{\Tr[H(S_0-S_1)]:S_0,S_1\in\Lin(\mathcal{H}_A\otimes\mathcal{H}_B),\,S_0,S_1\geq 0,\,\Tr_B[S_0+S_1]\leq\mathbbm{1}_A\Big\}\label{eq-diamond_norm_dual_primalSDP}\\
    &=\inf\Big\{\Tr[Y]:Y\in\Lin(\mathcal{H}_A),\,Y\geq 0,\,-Y\otimes\mathbbm{1}_B\leq H\leq Y\otimes\mathbbm{1}_B\Big\}.\label{eq-diamond_norm_dual_dualSDP}
\end{align}
We now show that this norm is multiplicative for tensor-product operators, which is a relevant property when considering channel observables without memory; see Section~\ref{sec-basics_quantum_info} for the relevant background information.

\begin{lemma}[Diamond norm dual for tensor-product operators]
    Let $K\in\Lin(\mathcal{H}_A)$ and $L\in\Lin(\mathcal{H}_B)$ be Hermitian operators. It holds that $\norm{K\otimes L}_{\diamond 1}^{\ast}=\norm{K}_1\norm{L}_{\infty}$.
\end{lemma}

\begin{proof} 
    This follows straightforwardly from semi-definite programming duality. Let $M_1,M_2\in\Lin(\mathcal{H}_A)$ be the operators achieving the trace norm $\norm{K}_1$, as in \eqref{eq-trace_norm_primal}, and let $M_1',M_2'\in\Lin(\mathcal{H}_B)$ be the operators achieving the spectral norm $\norm{L}_{\infty}$, as in \eqref{eq-infty_norm_primal}. Then, $S_0\equiv M_1\otimes M_1'+M_2\otimes M_2'$ and $S_1\equiv M_1\otimes M_2'+M_2\otimes M_1'$ are readily verified to be feasible points in the SDP \eqref{eq-diamond_norm_dual_primalSDP}, which means that $\norm{K\otimes L}_{\diamond 1}^{\ast}\geq \Tr[(K\otimes L)(S_0-S_1)]=\norm{K}_1\norm{L}_{\infty}$. For the reverse inequality, observe that $Y\equiv Z\norm{L}_{\infty}$, where $Z\in\Lin(\mathcal{H}_A)$ achieves the trace norm $\norm{K}_1$ according to \eqref{eq-trace_norm_dual}, is a feasible point in the SDP \eqref{eq-diamond_norm_dual_dualSDP}. Consequently, which this choice of $Y$, we obtain $\norm{K\otimes L}_{\diamond 1}^{\ast}\leq \Tr[Y]=\norm{K}_1\norm{L}_{\infty}$. This completes the proof.
\end{proof}

Let $N_1$ be the comb operator of a multi-time process with $r_1$ time steps, and let $N_2$ be the comb operator of a multi-time process with $r_2$ time steps, such that both processes have some compatible input and output Hilbert spaces. The multi-time process resulting from the composition of the two processes is represented by the comb operator $N_1\star N_2$, where $\star$ represents the \textit{link product}~\cite{CDP09}. 

Just as the diamond norm is submultiplicative with respect to composition of linear maps, so too is the strategy norm submultiplicative with respect to concatenation of comb operators according to the link product. We prove this in our next result.

\begin{proposition}[Submultiplicativity of the strategy norm]\label{prop-submult_strategy_norm}
    Let $N_1$ be the representation of a multi-time process with $r_1$ time steps, and let $N_2$ be the representation of a multi-time process with $r_2$ time steps. Suppose that the composition of these processes, represented by $N_1\star N_2$, produces a multi-time process with $r$ time steps. Then, it holds that
    \begin{equation}
        \norm{N_1\star N_2}_{\diamond r}\leq \norm{N_1}_{\diamond r_1}\norm{N_2}_{\diamond r_2}.
    \end{equation}
\end{proposition}

\begin{proof}
    This result follows straightforwardly from semi-definite programming duality. In particular, we make use
    of \eqref{eq-strategy_norm_dual_formulation_0}, which we restate here for convenience as
    \begin{equation}\label{eq-strategy_norm_dual_formulation}
        \norm{H}_{\diamond r}=\inf\Big\{t:t\geq 0,\,-t P\leq H\leq t P,\,P\in\mathsf{COMB}_r\Big\}.
    \end{equation}
    Now, let $(t_1,P_1)$ and $(t_2,P_2)$ be the optimal feasible points corresponding to $\norm{N_1}_{\diamond r_1}$ and $\norm{N_2}_{\diamond r_2}$, respectively, meaning that $t_1=\norm{N_1}_{\diamond r_1}$ and $t_2=\norm{N_2}_{\diamond r_2}$. We now show that $(t_1t_2,P_1\star P_2)$ constitutes a feasible point in the SDP in \eqref{eq-strategy_norm_dual_formulation} for $\norm{N_1\star N_2}_{\diamond r}$, implying the desired result.

    By definition, we have
    \begin{align}
        -t_1P_1\leq &N_1\leq t_1P_1,\\
        -t_2P_2\leq &N_2\leq t_2P_2.
    \end{align}
    The right-most inequalities can be rewritten as $t_1P_1-N_1\geq 0$ and $t_2P_2-N_2\geq 0$. Using the fact that the link product of positive semi-definite operators is positive semi-definite (see Ref.~\cite[Theorem~2]{CDP09}), we obtain $(t_1P_1-N_1)\star(t_2P_2-N_2)\geq 0$. Expanding the left-hand side of this inequality, we have
    \begin{align}
        & t_1t_2P_1\star P_2-t_1P_1\star N_2-t_2N_1\star P_2+N_1\star N_2\geq 0\\
        \Rightarrow & \, t_1t_2P_1\star P_2+N_1\star N_2\geq t_1P_1\star N_2+t_2N_1\star P_2\geq 0\\
        \Rightarrow & \, N_1\star N_2\geq -t_1t_2 P_1\star P_2.
    \end{align}
    Similarly, we have
    \begin{align}
        &N_1+t_1P_1\geq 0,\quad t_2P_2-N_2\geq 0\\
        \Rightarrow & \, (N_1+t_1P_1)\star (t_2P_2-N_2)\geq 0\\
        \Rightarrow & \, -N_1\star N_2+t_1t_2P_1\star P_2\geq t_1P_1\star N_2-t_2N_1\star P_2,\label{eq-strategy_norm_submult_pf1}\\[1ex]
        &t_1P_1-N_1\geq 0,\quad N_2+t_2P_2\geq 0\\
        \Rightarrow & \, (t_1P_1-N_1)\star (N_2+t_2P_2)\geq 0\\
        \Rightarrow & \, -N_1\star N_2+t_1t_2 P_1\star P_2\geq t_2 N_1\star P_2-t_1P_1\star N_2.\label{eq-strategy_norm_submult_pf2}
    \end{align}
    Adding the inequalities in \eqref{eq-strategy_norm_submult_pf1} and \eqref{eq-strategy_norm_submult_pf2}, we obtain
    \begin{equation}
        -N_1\star N_2+t_1t_2P_1\star P_2\geq 0\Rightarrow N_1\star N_2\leq t_1t_2 P_1\star P_2. 
    \end{equation}
    As $P_1\star P_2$ is positive semi-definite and defines a multi-time process with $r$ steps, we conclude that $(t_1t_2,P_1\star P_2)$ is a feasible point in the SDP in \eqref{eq-strategy_norm_dual_formulation} for $\norm{N_1\star N_2}_{\diamond r}$, which implies the desired result.
\end{proof}

\begin{corollary}[Subadditivity of the strategy norm under composition]\label{cor:strategy_norm_subadditivity_composition}
    Let $N_1,M_1$ be representations of multi-time processes with $r_1$ time steps, and let $N_2,M_2$ be representations of multi-time processes with $r_2$ time steps. Suppose that the composition of $N_1$ with $N_2$ and $M_1$ with $M_2$ produces multi-time processes with $r$ time steps. Then,
    \begin{equation}
        \norm{N_1\star N_2-M_1\star M_2}_{\diamond r}\leq \norm{N_1 - M_1}_{\diamond r_1}+\norm{N_2 - M_2}_{\diamond r_2}.
    \end{equation}
\end{corollary}

\begin{proof}
    By the triangle inequality, and making use of Proposition~\ref{prop-submult_strategy_norm}, we have
    \begin{align}
        \norm{N_1\star N_2-M_1\star M_2}_{\diamond r}&=\norm{N_1\star N_2-N_1\star M_2+N_1\star M_2-M_1\star M_2}_{\diamond r}\\
         \nonumber
        &\leq \norm{N_1\star N_2-N_1\star M_2}_{\diamond r}+\norm{N_1\star M_2-M_1\star M_2}_{\diamond r}\\
         \nonumber
        &=\norm{N_1\star(N_2-M_2)}_{\diamond r}+\norm{(N_1-M_1)\star M_2}_{\diamond r}\\
         \nonumber
        &\leq \norm{N_2-M_2}_{\diamond r_2}+\norm{N_1-M_1}_{\diamond r_1},
         \nonumber
    \end{align}
    as required. Here, the last step used Proposition~\ref{prop-submult_strategy_norm} together with the fact that $\norm{N_1}_{r_1},\norm{M_2}_{r_2}\leq 1$. The latter can for example easily be seen from \Cref{eq-strategy_norm_dual_formulation}.
\end{proof}

\section{Pauli-twirl of quantum channels}\label{sec-Pauli_twirl_quantum_channel}

    In this section, we prove \eqref{eq-Pauli_twirl_pinching}, and thereby prove that the error-rate vector of the Pauli-twirled version of an arbitrary quantum channel $\mathcal{N}:\Lin(\mathbb{C}^d)\to\Lin(\mathbb{C}^d)$, $d\in\{2,3,\dotsc\}$, is given by $p(z,x)=\frac{1}{d}\Tr[\Phi^{z,x}C(\mathcal{N})]$ for all $z,x\in\{0,1,\dotsc,d-1\}$. For generality, we prove the result in terms of qudit Pauli channels (see Appendix~\ref{sec:Pauli}), but an analogous proof to the one we present holds when the qudit Pauli operators are replaced by the $n$-qubit Pauli operators.

    \begin{lemma}\label{lem-Choi_rep_twirling}
        For every qudit quantum channel $\mathcal{N}:\Lin(\mathbb{C}^d)\to\Lin(\mathbb{C}^d)$, $d\in\{2,3,\dotsc\}$, the Choi representation of its Pauli-twirled version $\mathcal{N}^{\mathsf{W}}$, as defined in \eqref{eq-qudit_Pauli_twirled_channel}, is given by
        \begin{equation}
    C(\mathcal{N}^{\mathsf{W}})=\frac{1}{d^2}\sum_{z,x=0}^{d-1} (\conj{W^{z,x}}\otimes W^{z,x})^{\dagger}C(\mathcal{N})(\conj{W^{z,x}}\otimes W^{z,x}).
        \end{equation}
    \end{lemma}

    \begin{proof}
        By definition of the Choi representation, we have
        \begin{align}
            C(\mathcal{N}^{\mathsf{W}})&=(\id_d\otimes\mathcal{N}^{\mathsf{W}})(\ketbra{\Gamma}{\Gamma})\\
             \nonumber
            &=\frac{1}{d^2}\sum_{z,x=0}^{d-1}(\mathbbm{1}_d\otimes W^{z,x\dagger})(\id_d\otimes\mathcal{N})\left((\mathbbm{1}_d\otimes W^{z,x})\ketbra{\Gamma}{\Gamma}(\mathbbm{1}_d\otimes W^{z,x})^{\dagger}\right)(\mathbbm{1}_d\otimes W^{z,x})\\
             \nonumber
            &=\frac{1}{d^2}\sum_{z,x=0}^{d-1}(\mathbbm{1}_d\otimes W^{z,x\dagger})(\id_d\otimes\mathcal{N})\left((W^{z,x\t}\otimes\mathbbm{1}_d)\ketbra{\Gamma}{\Gamma}(W^{z,x\t}\otimes\mathbbm{1}_d)^{\dagger}\right)(\mathbbm{1}_d\otimes W^{z,x})\\
             \nonumber
            &=\frac{1}{d^2}\sum_{z,x=0}^{d-1}(W^{z,x\t}\otimes W^{z,x\dagger})C(\mathcal{N})(\conj{W^{z,x}}\otimes W^{z,x})\\
             \nonumber
            &=\frac{1}{d^2}\sum_{z,x=0}^{d-1}(\conj{W^{z,x}}\otimes W^{z,x})^{\dagger}C(\mathcal{N})(\conj{W^{z,x}}\otimes W^{z,x}),
             \nonumber
        \end{align}
        as required, where for the third equality we have used the ``transpose trick'' $(\mathbbm{1}_d\otimes X)\ket{\Gamma}=(X^{\t}\otimes\mathbbm{1}_d)\ket{\Gamma}$, which holds for every linear operator $X\in\Lin(\mathbb{C}^d)$.
    \end{proof}

    \begin{proposition}\label{prop-qudit_Pauli_channel_twirling}
        Let $X\in\Lin(\mathbb{C}^d\otimes\mathbb{C}^d)$, $d\in\{2,3,\dotsc\}$, and define $\mathcal{S}_{\mathsf{W}}$ to be the pinching channel in the qudit Bell basis:
        \begin{equation}\label{eq-pinching_twirling}
            \mathcal{S}_{\mathsf{W}}(X)\coloneqq\sum_{z,x=0}^{d-1}\ketbra{\Phi^{z,x}}{\Phi^{z,x}}X\ketbra{\Phi^{z,x}}{\Phi^{z,x}}.
        \end{equation}
        It holds that
        \begin{equation}
            \mathcal{S}_{\mathsf{W}}(X)=\frac{1}{d^2}\sum_{z,x=0}^{d-1}(\conj{W^{z,x}}\otimes W^{z,x})^{\dagger}X(\conj{W^{z,x}}\otimes W^{z,x}).
        \end{equation}
        Consequently, for every channel $\mathcal{N}:\Lin(\mathbb{C}^d)\to\Lin(\mathbb{C}^d)$, the Choi representation of its Pauli-twirled version $\mathcal{N}^{\mathsf{W}}$, defined in \eqref{eq-qudit_Pauli_twirled_channel}, is given by $C(\mathcal{N}^{\mathsf{W}})=\mathcal{S}_{\mathsf{W}}(C(\mathcal{N}))$. Furthermore, the corresponding error-rate vector of the Pauli-twirled channel is given by $p(z,x)=\frac{1}{d}\Tr[\Phi^{z,x}C(\mathcal{N})]$ for all $z,x\in\{0,1,\dotsc,d-1\}$.
    \end{proposition}

    \begin{proof}
        We make repeated use of the following facts about the qudit Pauli operators~\cite{Wat18_book}:
        \begin{align}
            W^{z,x\dagger}&=\e^{-\frac{2\pi\I zx}{d}}W^{-z,-x},\label{eq-qudit_Pauli_1}\\
            W^{z,x\t}&=\e^{\frac{2\pi\I zx}{d}}W^{z,-x},\label{eq-qudit_Pauli_2}\\
            W^{z,x}W^{z',x'}&=\e^{-\frac{2\pi\I xz'}{d}}W^{z+z',x+x'},\label{eq-qudit_Pauli_3}
        \end{align}
        which hold for all choices of $z,x,z',x'\in\{0,1,\dotsc,d-1\}$. Now, let us start by showing that 
        \begin{equation}\label{eq-qudit_Pauli_twirl_pf1}
            \Phi^{z,x}=\frac{1}{d^2}\sum_{z',x'=0}^{d-1}\e^{\frac{2\pi\I}{d}(x'z-xz')}\conj{W^{z',x'}}\otimes W^{z',x'},
        \end{equation}
        for all $z,x\in\{0,1,\dotsc,d-1\}$. To show this, we use the definition of $\Phi^{z,x}$ in \eqref{eq-qudit_Bell_states}, the properties in \eqref{eq-qudit_Pauli_1}--\eqref{eq-qudit_Pauli_3}, and the fact that the operators $\{\conj{W^{z_1',x_1'}}\otimes W^{z_2',x_2'}:z_1',x_1',z_2',x_2'\in\{0,1,\dotsc,d-1\}\}$ form a basis for $\Lin(\mathbb{C}^d\otimes\mathbb{C}^d)$,
        \begin{align}
            &\Tr\!\left[\left(\conj{W^{z_1',x_1'}}\otimes W^{z_2',x_2'}\right)^{\dagger}\Phi^{z,x}\right]\\
             \nonumber
            &\quad=\frac{1}{d}\Tr\!\left[(\conj{W^{z_1',x_1'}}\otimes W^{z_2',x_2'})^{\dagger}(\mathbbm{1}_d\otimes W^{z,x})\ketbra{\Gamma}{\Gamma}(\mathbbm{1}_d\otimes W^{z,x\dagger})\right]\\
             \nonumber
            &\quad=\frac{1}{d}\Tr\!\left[W^{z,x\dagger}W^{z_2',x_2'\dagger}W^{z,x}W^{z_1',x_1'}\right]\\
             \nonumber
            &\quad=\frac{1}{d}\Tr\!\left[\left(W^{z_2',x_2'}W^{z,x}\right)^{\dagger}W^{z,x}W^{z_1',x_1'}\right]\\
             \nonumber
            &\quad=\frac{1}{d}\e^{\frac{2\pi\I}{d}(x_2'z-xz_1')}\underbrace{\Tr\!\left[\left(W^{z_2'+z,x_2'+x}\right)^{\dagger}W^{z+z_1',x+x_1'}\right]}_{d\delta_{z_2'+z,z+z_1'}\delta_{x_2'+x,x+x_1'}}\\
             \nonumber
            &\quad=\e^{\frac{2\pi\I}{d}(x_2'z-xz_1')}\delta_{z_2',z_1'}\delta_{x_2',x_1'}.
             \nonumber
        \end{align}
        Therefore, using \eqref{eq-qudit_Pauli_twirl_pf1}, along with the properties in \eqref{eq-qudit_Pauli_1}--\eqref{eq-qudit_Pauli_3} once more, we obtain
        \begin{align}
            \mathcal{S}_{\mathsf{W}}(X)&=\sum_{z,x=0}^{d-1}\Phi^{z,x}X\Phi^{z,x}\\
            \nonumber
            &=\frac{1}{d^4}\sum_{z,x=0}^{d-1}\sum_{z_1',x_1'=0}^{d-1}\sum_{z_2',x_2'=0}^{d-1} \left(\e^{\frac{2\pi\I}{d}(-xz_1'+x_1'z)}\conj{W^{z_1',x_1'}}\otimes W^{z_1',x_1'}\right)X\left(\e^{\frac{2\pi\I}{d}(-xz_2'+x_2'z)}\conj{W^{z_2',x_2'}}\otimes W^{z_2',x_2'}\right)\\
            \nonumber
            &=\frac{1}{d^4}\sum_{z_1',x_1'=0}^{d-1}\sum_{z_2',x_2'=0}^{d-1}\underbrace{\left(\sum_{x=0}^{d-1}\e^{-\frac{2\pi\I}{d}x(z_1'+z_2')}\right)}_{d\delta_{z_1',-z_2'}}\underbrace{\left(\sum_{z=0}^{d-1}\e^{\frac{2\pi\I}{d}z(x_1'+x_2')}\right)}_{d\delta_{x_1',-x_2'}}(\conj{W^{z_1',x_1'}}\otimes W^{z_1',x_1'})X(\conj{W^{z_2',x_2'}}\otimes W^{z_2',x_2'})\\
            \nonumber
            &=\frac{1}{d^2}\sum_{z,x=0}^{d-1}(\conj{W^{-z,-x}}\otimes W^{-z,-x})X(\conj{W^{z,x}}\otimes W^{z,x})\\
            &=\frac{1}{d^2}\sum_{z,x=0}^{d-1}\e^{\frac{2\pi\I zx}{d}} (W^{z,-x}\otimes W^{z,x\dagger})X(\conj{W^{z,x}}\otimes W^{z,x})\\
            \nonumber
            &=\frac{1}{d^2}\sum_{z,x=0}^{d-1}\e^{\frac{2\pi\I zx}{d}}\e^{-\frac{2\pi\I zx}{d}}(W^{z,x\t}\otimes W^{z,x\dagger})X(\conj{W^{z,x}}\otimes W^{z,x})\\
            \nonumber
            &=\frac{1}{d^2}\sum_{z,x=0}^{d-1}(\conj{W^{z,x}}\otimes W^{z,x})^{\dagger}X(\conj{W^{z,x}}\otimes W^{z,x}),
            \nonumber
        \end{align}
        which proves \eqref{eq-pinching_twirling}. Then, if $X\equiv C(\mathcal{N})$ is the Choi representation of a quantum channel $\mathcal{N}:\Lin(\mathbb{C}^d)\to\Lin(\mathbb{C}^d)$, using Lemma~\ref{lem-Choi_rep_twirling} we see that $C(\mathcal{N}^{\mathsf{W}})=\mathcal{S}_{\mathsf{W}}(C(\mathcal{N}))=\sum_{z,x=0}^{d-1}\Tr[\Phi^{z,x}C(\mathcal{N})]\Phi^{z,x}=\sum_{z,x=0}^{d-1}\frac{1}{d}\Tr[\Phi^{z,x}C(\mathcal{N})]\Gamma^{z,x}$. By identifying with \eqref{eq-qudit_Pauli_channel_Choi_rep}, we can see that the twirled channel is indeed a Pauli channel, and using \eqref{eq-qudit_Pauli_channel_error_vector}, we see that the error-rate vector of the twirled channel is $p(z,x)=\frac{1}{d}\Tr[\Phi^{z,x}C(\mathcal{N})]$ for all $z,x\in\{0,1,\dotsc,d-1\}$. This completes the proof.
    \end{proof}

\section{Entropic analysis of the MMW algorithm}\label{sec-MMW_entropic}

In this section, we provide an entropic analysis of the \emph{matrix multiplicative weights} (MMW) algorithm (Algorithm~\ref{alg:mmw} and its projected variant for Choi states (Algorithm~\ref{alg:projected-mmw-choi-state}). We note that an entropic analysis similar to the one we provide here can be found in Ref.~\cite[Theorem~2.4]{arora2012multiplicative} for the \emph{multiplicative weights update} (MWU) algorithm.
On it highest level, the MMW algorithm assigns 
 initial weights to experts iteratively updates these weights multiplicatively according to the feedback 
 on how well the expert has performed. 
It is known as a method for highly efficiently solve
convex optimization problems.

We start with a proof of \Cref{prop-MMW_amortization_regret_bound}. The bound we obtain applies also to the Hedge algorithm (Algorithm~\ref{alg:hedge} below), which is a special case of the MMW algorithm (Algorithm~\ref{alg:mmw}) when the loss matrices $L^{(t)}$ of the MMW algorithm are all diagonal in the same basis, such that the diagonal elements of $L^{(t)}$ constitute the loss vector $\Vec{m}^{(t)}$ in the Hedge algorithm. The bound we obtain for the Hedge algorithm is in general tighter than the one obtained in Ref.~\cite[Theorem~2.3]{arora2012multiplicative}.

\begin{algorithm} 
    \begin{algorithmic}[1]
    \Require {$0 < \eta \leq 1$; Initialize $\Vec{w}^{(1)} = \Vec{1}$}
    \For{\texttt{$t = 1, 2, \ldots T$}}
        \State Output the decision/estimate $\Vec{p}^{(t)}=\frac{\Vec{w}^{(t)}}{\Tr[\Vec{w}^{(t)}]}$.
        \State Receive the cost vector $-1\leq \Vec{m}^{(t)}\leq 1$ (element-wise inequalities).
        \State Update the weights as $\Vec{w}^{(t+1)}=\Vec{w}^{(t)}\e^{-\eta\Vec{m}^{(t)}}$ (element-wise multiplication).
    \EndFor
    \end{algorithmic}
    \caption{The Hedge algorithm~\cite{freund1997decision,arora2012multiplicative}}
    \label{alg:hedge}
\end{algorithm}

\subsection{Proof of Proposition~\ref{prop-MMW_amortization_regret_bound}}\label{sec-prop-MMW_amortization_regret_bound_pf}
    We start by noticing that in the proof of \cite[Theorem~3.1]{arora2016MMW}, the inequality
    \begin{equation}
        \Tr[W^{(t+1)}]\leq \Tr[W^{(t)}]\e^{-\eta\Tr[L^{(t)}\omega^{(t)}]+\eta^2\Tr[(L^{(t)})^2\omega^{(t)}]}
    \end{equation}
    can be written as
    \begin{equation}
        \frac{\Tr[W^{(t+1)}]}{\Tr[W^{(t)}]}\leq \e^{-\eta\Tr[L^{(t)}\omega^{(t)}]+\eta^2\Tr[(L^{(t)})^2\omega^{(t)}]}.
    \end{equation}
    Taking the logarithm on both sides leads to
    \begin{equation}\label{eq-MMW_amortization}
        \log\Tr[W^{(t+1)}]-\log\Tr[W^{(t)}]\leq -\eta\Tr[L^{(t)}\omega^{(t)}]+\eta^2\Tr[(L^{(t)})^2\omega^{(t)}].
    \end{equation}
    Applying \eqref{eq-MMW_amortization} recursively leads to
    \begin{align}
        \log\Tr[W^{(T+1)}]-\log\Tr[W^{(1)}]&=\log\Tr[W^{(T+1)}]-\log\Tr[W^{(T)}]\nonumber\\
        &\quad+\log\Tr[W^{(T)}]-\log\Tr[W^{(T-1)}]+\dotsb-\log\Tr[W^{(2)}]\nonumber\\
        &\quad+\log\Tr[W^{(2)}]-\log\Tr[W^{(1)}]\\
        &\leq \sum_{t=1}^T\left(-\eta\Tr[L^{(t)}\omega^{(t)}]+\eta^2\Tr[(L^{(t)})^2\omega^{(t)}]\right).
         \nonumber
    \end{align}
    Now, let $\rho$ be an arbitrary density operator. Then, noting that the von Neumann entropy of $\rho$ is given by $H(\rho)\coloneqq-\Tr[\rho\log\rho]$, we have that the relative entropy between $\rho$ and $\omega^{(t)}$ is\footnote{Observe that the updates $\omega^{(t)}$ always have full rank, so the support condition in \eqref{eq-relative_entropy} is satisfied.}
    \begin{align}
        D(\rho\Vert \omega^{(t)})&\coloneqq\Tr[\rho\log\rho]-\Tr[\rho\log\omega^{(t)}]\\
         \nonumber
        &=-H(\rho)-\Tr[\rho\log \omega^{(t)}]\\
         \nonumber
        &=-H(\rho)-\Tr\!\left[\rho\log\frac{W^{(t)}}{\Tr[W^{(t)}]}\right]\\
         \nonumber
        &=-H(\rho)+\log\Tr[W^{(t)}]-\Tr[\rho\log W^{(t)}].
         \nonumber
    \end{align}
    Noting further that
    \begin{align}
        \log W^{(T+1)}&=\log\exp\!\left(-\eta\sum_{t=1}^{T}L^{(t)}\right)\\
         \nonumber
        &=-\eta\sum_{t=1}^{T+1}L^{(t)}\\
         \nonumber
        &=-\eta L^{(T)}-\eta\sum_{t=1}^{T-1}L^{(t)}\\
         \nonumber
        &=-\eta L^{(T)}+\log W^{(T)},
         \nonumber
    \end{align}
    we obtain
    \begin{align}
     \nonumber
        D(\rho\Vert\omega^{(t+1)})-D(\rho\Vert\omega^{(t)})&=\log\Tr[W^{(t+1)}]-\log\Tr[W^{(t)}]-\Tr[\rho(-\eta L^{(t)}+\log W^{(t)})]\nonumber\\
        &\qquad +\Tr[Q\log W^{(t)}]
         \nonumber
         \\
         \nonumber
        &=\log\Tr[W^{(t+1)}]-\log\Tr[W^{(t)}]+\eta\Tr[\rho L^{(t)}]\\
        &\leq -\eta \Tr[L^{(t)}\omega^{(t)}]+\eta^2\Tr[(L^{(t)})^2\omega^{(t)}]+\eta\Tr[\rho L^{(t)}],\label{eq-MMW_amortization_regret_bound_pf}
    \end{align}
    where for the inequality on the last line we have used \eqref{eq-MMW_amortization}. This implies that
    \begin{align}
        D(\rho\Vert\omega^{(T+1)})-D(\rho\Vert\omega^{(1)})&=D(\rho\Vert\omega^{(T+1)})-D(\rho\Vert\omega^{(T)})\\
         \nonumber
        &\quad +D(\rho\Vert\omega^{(T)})-D(\rho\Vert\omega^{(T-1)})+\dotsb-D(\rho\Vert\omega^{(2)})\\
         \nonumber
        &\quad +D(\rho\Vert\omega^{(2)})-D(\rho\Vert\omega^{(1)})\\
         \nonumber
        &\leq -\eta\sum_{t=1}^T\Tr[L^{(t)}\omega^{(t)}]+\eta^2\sum_{t=1}^T\Tr[(L^{(t)})^2\omega^{(t)}]+\eta\sum_{t=1}^T\Tr[L^{(t)}\rho].
         \nonumber
    \end{align}
    Then, because $D(\rho\Vert\omega^{(T+1)})\geq 0$, we obtain
    \begin{equation}
        0\leq D(\rho\Vert\omega^{(T+1)})\leq -\eta\sum_{t=1}^T \Tr[L^{(t)}\omega^{(t)}]+\eta^2\sum_{t=1}^T\Tr[(L^{(t)})^2\omega^{(t)}]+\eta\sum_{t=1}^T\Tr[L^{(t)}\rho]+D(\rho\Vert\omega^{(1)}).
    \end{equation}
    Finally, because $\omega^{(1)}=\frac{1}{d}\mathbbm{1}$, and because $D(\rho\Vert\frac{1}{d}\mathbbm{1})=\log d-H(\rho)$, we can rearrange the inequality above to obtain the desired result.

\begin{remark}
    Let us make the following observations about the result in \eqref{eq-MMW_amortization_regret_bound}.
    \begin{itemize}
        \item If we let $\rho$ be a rank-one density operator, then $H(\rho)=0$, and then we can further minimize over all such rank-one density operators to obtain
            \begin{equation}
                \sum_{t=1}^T \Tr[L^{(t)}\omega^{(t)}]\leq \lambda_{\min}\!\left(\sum_{t=1}^TL^{(t)}\right)+\eta\sum_{t=1}^T\Tr[(L^{(t)})^2\omega^{(t)}]+\frac{\log d}{\eta},
            \end{equation}
            which is precisely the result in Ref.~\cite[Theorem~3.1]{arora2016MMW}. 
        \item It is worth noting that the MMW-based result for online learning of quantum states in Ref.~\cite[Theorem~4]{ACH+19} (see, in particular, the proof) presents the regret bound (using the notation of this section, and $d=2^n$)
            \begin{equation}
                \sum_{t=1}^T\Tr[L^{(t)}\omega^{(t)}]\leq \Tr\!\left[\rho\left(\sum_{t=1}^T L^{(t)}\right)\right]+\eta\sum_{t=1}^T\Tr[(L^{(t)})^2\omega^{(t)}]+\frac{\log(2^n)}{\eta},
            \end{equation}
            for every density operator $\rho$. Note that the regret bound in \eqref{eq-MMW_amortization_regret_bound} can in general be tighter than this bound, because of the entropy term $H(\rho)$ in \eqref{eq-MMW_amortization_regret_bound}, which is always non-negative.
        \item We can minimize the right-hand side of \eqref{eq-MMW_amortization_regret_bound} with respect to $\rho$, in order to obtain the best possible upper bound on the total expected loss. In other words,
            \begin{equation}
                \sum_{t=1}^T\Tr[L^{(t)}\omega^{(t)}]\leq\inf_{\substack{\rho\geq 0\\\Tr[\rho]=1}}\left(\Tr\!\left[\rho\left(\sum_{t=1}^TL^{(t)}\right)\right]-\frac{H(\rho)}{\eta}\right)+\eta\sum_{t=1}^T\Tr[(L^{(t)})^2\omega^{(t)}]+\frac{\log d}{\eta}.
            \end{equation}
    \end{itemize}
\end{remark}

From the connection between the MWW and Hedge algorithms noted at the beginning of this section, we immediately obtain the following regret bound for the Hedge algorithm from Proposition~\ref{prop-MMW_amortization_regret_bound}.

\begin{corollary}[Entropic regret bound for the Hedge algorithm]\label{cor-Hedge_amortization_regret_bound}
    Let $T\in\mathbb{N}$, and consider a sequence $\Vec{m}^{(1)},\Vec{m}^{(2)},\dotsc,\Vec{m}^{(T)}$ of loss vectors of size $d\in\{2,3,\dotsc\}$ along with the updates $\Vec{p}^{(t)}$ provided by the Hedge algorithm in Algorithm~\ref{alg:hedge}. Then, the following inequality holds:
    \begin{equation}
        \sum_{t=1}^T\Vec{p}^{(t)}\cdot\Vec{m}^{(t)}\leq \Vec{q}\cdot\left(\sum_{t=1}^T\Vec{m}^{(t)}\right)+\eta\sum_{t=1}^T\Vec{p}^{(t)}\cdot(\Vec{m}^{(t)})^2+\frac{\log d-H(\Vec{q})}{\eta},
    \end{equation}
    where $\Vec{q}$ is an arbitrary probability vector.
\end{corollary}

\subsection{The projected MMW algorithm}\label{sec-projected_MMW}

    We start by defining the projection map as 
    \begin{equation}\label{eq-rel_ent_CPTP_proj}
        \Pi(\sigma_{A,B})=\argmin_{\rho_{A,B}}\Big\{D(\rho_{A,B}\Vert\sigma_{A,B}):\rho_{A,B}\geq 0,\,\Tr_B[\rho_{A,B}]=\frac{\mathbbm{1}_A}{d_A}\Big\},
    \end{equation}
    where $\sigma_{A,B}$ is a density operator and the relative entropy is $D(\cdot\Vert\cdot)$ is defined in \eqref{eq-relative_entropy}. We make use of the fact that the relative entropy is a Bregman divergence~\cite{tsuda2005MMW,petz2007bregman,dhillon2008matrixnearnessBregman}. In particular, for $\rho,\sigma$ density operators, with $\sigma$ positive definite, we have that
    \begin{equation}
        D(\rho\Vert\sigma)=B_F(\rho\Vert\sigma)\coloneqq F(\rho)-F(\sigma)-\Tr[\nabla F(\sigma)(\rho-\sigma)],
    \end{equation}
    where
    \begin{align}
        F(P)&\coloneqq \Tr[P\log P-P],\label{eq-Bregman_F}\\
        \nabla F(Q)&\equiv \log Q,\label{eq-grad_F}
    \end{align}
    for $P$ positive semi-definite and $Q$ positive definite. It follows that the projection map in \eqref{eq-rel_ent_CPTP_proj} is a Bregman projection; consequently, we have the so-called \textit{Pythagorean inequality}~\cite{dhillon2008matrixnearnessBregman},
    \begin{equation}\label{eq-rel_ent_Pythagorean}
        D(\rho\Vert\sigma)\geq D(\rho\Vert\Pi(\sigma))+D(\Pi(\sigma)\Vert\sigma),
    \end{equation}
    for density operators $\rho$ and $\sigma$. This inequality essentially tells us that projection only get us closer to the set of CPTP maps, in the sense that
    \begin{equation}\label{eq-rel_ent_Pythagorean_2}
        D(\rho\Vert\Pi(\sigma))\leq D(\rho\Vert\sigma),
    \end{equation}
    which follows directly from \eqref{eq-rel_ent_Pythagorean}, due to the fact that $D(\rho\Vert\sigma)\geq 0$ for all density operators $\rho$ and $\sigma$. We also require the \textit{Pinsker inequality}~\cite{Wat18_book}: for all density operators $\rho$ and $\sigma$,
    \begin{equation}\label{eq-pinsker_inequality}
        D(\rho\Vert\sigma)\geq\frac{1}{2}\norm{\rho-\sigma}_1^2.
    \end{equation}

    Finally, we make the observation that the update step 4 in the MMW algorithm (Algorithm~\ref{alg:mmw} and Algorithm~\ref{alg:projected-mmw-choi-state})
    can be written as
    \begin{equation}
        W^{(t+1)}=\exp\!\left[\log(W^{(t)})-\eta L^{(t)}\right]=\exp\!\left[\nabla F(W^{(t)})-\eta L^{(t)}\right],
    \end{equation}
    where we recall the expression for $\nabla F$ in \eqref{eq-grad_F}. With this observation, we can equivalently formulate Algorithm~\ref{alg:projected-mmw-choi-state} as the lazy version of Algorithm~\ref{alg:projected_mmw_choi_state_2} below, which is a \textit{mirror descent} algorithm~\cite[Section~5.3]{OCO_hazan_book}.

    \begin{algorithm}
     \caption{Mirror descent algorithm for Choi states of quantum channels}\label{alg:projected_mmw_choi_state_2}
     \begin{algorithmic}[1]
     \Require {$\eta \in (0, 1)$; Initialize $W_{A,B}^{(1)} = \mathbbm{1}_A \otimes \mathbbm{1}_B$ and $\rho_{A,B}^{(1)}=\frac{1}{d_Ad_B}\mathbbm{1}_A\otimes\mathbbm{1}_B$}
     \For{\texttt{$t = 1, 2, \ldots, T$}}
        \State Output the decision/estimate $\rho_{A,B}^{(t)}$
        \State Receive the cost matrix $L_{A,B}^{(t)}$, $-\mathbbm{1}_{A,B}\leq L_{A,B}^{(t)}\leq\mathbbm{1}_{A,B}$
        \State Take the gradient step
            \begin{align*}
                & \text{(Lazy version)} \quad W_{A,B}^{(t)}\to Y_{A,B}^{(t)}\coloneqq\nabla F(W_{A,B}^{(t)})-\eta L_{A,B}^{(t)} \\[1ex]
                & \text{(Agile version)} \quad W_{A,B}^{(t)}\to Y_{A,B}^{(t)}\coloneqq\nabla F(\rho_{A,B}^{(t)})-\eta L_{A,B}^{(t)}
            \end{align*}
        \State Set $W_{A,B}^{(t+1)}=\exp[Y_{A,B}^{(t)}]$.
        \State Project: $\rho_{A,B}^{(t+1)}\coloneqq\Pi(\omega_{A,B}^{(t+1)})$, $\omega_{A,B}^{(t+1)}\equiv\frac{W_{A,B}^{(t+1)}}{\Tr[W_{A,B}^{(t+1)}]}$.
        \EndFor
     \end{algorithmic}
\end{algorithm}

\begin{proposition}[Regret bound for lazy mirror descent for Choi states]\label{prop-projected_MMW_Choi_state_regret_bound}
    Let $\sigma_{A,B}$ be an arbitrary Choi state. Let $T\in\mathbb{N}$ be the number of rounds of interaction, and consider a sequence $L_{A,B}^{(1)},L_{A,B}^{(2)},\dotsc,L_{A,B}^{(T)}$ of cost matrices along with the updates $\rho_{A,B}^{(t)}$ provided by the lazy version of Algorithm~\ref{alg:projected_mmw_choi_state_2}. Then,
    \begin{equation}\label{eq-projected_MMW_Choi_state_regret_bound}
        \sum_{t=1}^T\Tr[L_{A,B}^{(t)}\rho_{A,B}^{(t)}]\leq\Tr\!\left[\sigma_{A,B}\left(\sum_{t=1}^TL_{A,B}^{(t)}\right)\right]+2\eta\sum_{t=1}^T\norm{L_{A,B}^{(t)}}_{\infty}^2+\frac{\log(d_Ad_B)-H(\sigma_{A,B})}{\eta}.
    \end{equation}
\end{proposition}

\begin{proof}
    It turns out that the lazy version of Algorithm~\ref{alg:projected_mmw_choi_state_2} is equivalent to the \textit{regularized follow-the-leader} algorithm~\cite[Section~5.3.1]{OCO_hazan_book}. In particular, using Ref.~\cite[Lemma~5.5]{OCO_hazan_book}, it follows that the projection step for the lazy version of Algorithm~\ref{alg:projected_mmw_choi_state_2} is given by
    \begin{align}
         \nonumber
         \rho_{A,B}^{(t+1)}&\coloneqq\Pi(\omega_{A,B}^{(t+1)})\\
        &=\argmin_{\rho_{A,B}}\Big\{\eta\sum_{s=1}^{t}\Tr[L_{A,B}^{(s)}\rho_{A,B}]+F(\rho_{A,B}):\rho_{A,B}\geq 0,\,\Tr_B[\rho_{A,B}]=\frac{\mathbbm{1}_A}{d_A}\Big\},\label{eq-projected_MMW_Choi_state_regret_bound_pf1}
    \end{align}
    where the function $F$ is defined in \eqref{eq-Bregman_F}. Indeed, the gradient of the objective function on the right-hand side is equal to
    \begin{equation}\label{eq-projected_MMW_Choi_state_regret_bound_pf2}
        \nabla\!\left( \eta\sum_{s=1}^{t}\Tr[L_{A,B}^{(s)}P_{A,B}]+F(P_{A,B}) \right)=\eta\sum_{s=1}^T(L_{A,B}^{(s)})^{\t}+\nabla F(P_{A,B})
    \end{equation}
    where $P_{A,B}$ is an arbitrary positive semi-definite operator. At the same time, let us observe that the gradient step of the lazy version of Algorithm~\ref{alg:projected_mmw_choi_state_2} is given by
    \begin{align}
        \nabla F(W_{A,B}^{(t+1)})&=\nabla F(W_{A,B}^{(t)})-\eta L_{A,B}^{(t)}\\
         \nonumber
        &=\nabla F(W_{A,B}^{(t-1)})-\eta L_{A,B}^{(t-1)}-\eta L_{A,B}^{(t)}\\
         \nonumber
        &~~\vdots\\
         \nonumber
        &=\nabla F(W_{A,B}^{(1)})-\eta\sum_{s=1}^t L_{A,B}^{(s)}\\
         \nonumber
        &=-\eta\sum_{s=1}^t L_{A,B}^{(s)},
         \nonumber
    \end{align}
    where the last equality holds because $W_{A,B}^{(1)}=\mathbbm{1}_{A,B}$ and $\log(\mathbbm{1}_{A,B})=0$. This implies that
    \begin{align}
        \nabla F(\omega_{A,B}^{(t+1)})&=\log\omega_{A,B}^{(t+1)}\\
         \nonumber
        &=\log W_{A,B}^{(t+1)}-(\log\Tr[W_{A,B}^{(t+1)}])\mathbbm{1}_{A,B}\\
         \nonumber
        &=\nabla F(W_{A,B}^{(t+1)})-(\log\Tr[W_{A,B}^{(t+1)}])\mathbbm{1}_{A,B}\\
         \nonumber
        &=-\eta\sum_{s=1}^t (L_{A,B}^{(s)})^{\t}-(\log\Tr[W_{A,B}^{(t+1)}])\mathbbm{1}_{A,B}.
         \nonumber
    \end{align}
    Therefore,
    \begin{align}
        D(\rho_{A,B}\Vert\omega_{A,B}^{(t+1)})&=F(\rho_{A,B})-F(\omega_{A,B}^{(t+1)})-\Tr\!\left[\nabla F(\omega_{A,B}^{(t+1)})(\rho_{A,B}-\omega_{A,B}^{(t+1)})\right]\\
         \nonumber
        &=F(\rho_{A,B})-F(\omega_{A,B}^{(t+1)})+\eta\Tr\!\left[\left(\sum_{s=1}^t L_{A,B}^{(s)}\right)(\rho_{A,B}-\omega_{A,B}^{(t+1)})\right]\nonumber\\
         \nonumber
        &\qquad\qquad-\log(\Tr[W_{A,B}^{(t+1)}])\underbrace{\Tr[\rho_{A,B}-\omega_{A,B}^{(t+1)}]}_{=0}\\
         \nonumber
        &=F(\rho_{A,B})-F(\omega_{A,B}^{(t+1)})+\eta\Tr\!\left[\left(\sum_{s=1}^tL_{A,B}^{(s)}\right)(\rho_{A,B}-\omega_{A,B}^{(t+1)})\right],
         \nonumber
    \end{align}
    which implies that
    \begin{equation}\label{eq-projected_MMW_Choi_state_regret_bound_pf3}
        \nabla D(\rho_{A,B}\Vert\omega_{A,B}^{(t+1)})=-\nabla F(\rho_{A,B})+\eta\sum_{s=1}^t (L_{A,B}^{(s)})^{\t}.
    \end{equation}
    Combining \eqref{eq-projected_MMW_Choi_state_regret_bound_pf2} and \eqref{eq-projected_MMW_Choi_state_regret_bound_pf3}, and using the fact that the function $F$ is strictly convex, we can conclude that \eqref{eq-projected_MMW_Choi_state_regret_bound_pf1} holds. The desired regret bound then follows by Ref.~\cite[Theorem~3]{ACH+19}, which considers the regularized follow-the-leader algorithm for quantum states. 
\end{proof}

\begin{proposition}[Regret bound for agile mirror descent for Choi states]\label{prop-projected_MMW_Choi_state_regret_bound_agile}
    Let $\sigma_{A,B}$ be an arbitrary Choi state. Let $T\in\mathbb{N}$ be the number of rounds of interaction, and consider a sequence $L_{A,B}^{(1)},L_{A,B}^{(2)},\dotsc,L_{A,B}^{(T)}$ of cost matrices along with the updates $\rho_{A,B}^{(t)}$ provided by the agile version of Algorithm~\ref{alg:projected_mmw_choi_state_2}. Then,
    \begin{equation}\label{eq-projected_MMW_Choi_state_regret_bound_agile}
        \sum_{t=1}^T\Tr[L_{A,B}^{(t)}\rho_{A,B}^{(t)}]\leq\Tr\!\left[\sigma_{A,B}\left(\sum_{t=1}^TL_{A,B}^{(t)}\right)\right]+\frac{\eta}{2}\sum_{t=1}^T\norm{L_{A,B}^{(t)}}_{\infty}^2+\frac{\log(d_Ad_B)-H(\sigma_{A,B})}{\eta}.
    \end{equation}
\end{proposition}

\begin{proof}
    Similar to the proof of Proposition~\ref{prop-MMW_amortization_regret_bound} (see \eqref{eq-MMW_amortization_regret_bound_pf}, in particular), the idea of the proof is to bound $D(\sigma_{A,B}\Vert\rho_{A,B}^{(t)})-D(\sigma_{A,B}\Vert\rho_{A,B}^{(t+1)})$ for every time step $t\in\{1,2,\dotsc,T\}$. To this end, we start by noting that from the gradient step of the agile version of Algorithm~\ref{alg:projected_mmw_choi_state_2}, it holds that
    \begin{equation}
        L_{A,B}^{(t)}=\frac{1}{\eta}\left(\log\rho_{A,B}^{(t)}-\log W_{A,B}^{(t+1)}\right),
    \end{equation}
    for all $t\in\{1,2,\dotsc,T\}$. Using this, and with straightforward manipulations, we obtain the following:
    \begin{align}
        \Tr[L_{A,B}^{(t)}\rho_{A,B}^{(t)}]-\Tr[L_{A,B}^{(t)}\sigma_{A,B}]&=\Tr[L_{A,B}^{(t)}(\rho_{A,B}^{(t)}-\sigma_{A,B})]\\
         \nonumber
        &=\frac{1}{\eta}\Tr\!\left[\left(\log\rho_{A,B}^{(t)}-\log W_{A,B}^{(t+1)}\right)(\rho_{A,B}^{(t)}-\sigma_{A,B})\right]\\
         \nonumber
        &=\frac{1}{\eta}\Tr\!\left[\left(\log W_{A,B}^{(t+1)}-\log\rho_{A,B}^{(t)}\right)(\sigma_{A,B}-\rho_{A,B}^{(t)})\right]\\
         \nonumber
        &=\frac{1}{\eta}\left(D(\sigma_{A,B}\Vert\rho_{A,B}^{(t+1)})-D(\sigma_{A,B}\Vert W_{A,B}^{(t+1)})+D(\rho_{A,B}^{(t)}\Vert W_{A,B}^{(t+1)})\right)\\
         \nonumber
        &=\frac{1}{\eta}\left(D(\sigma_{A,B}\Vert\rho_{A,B}^{(t)})-D(\sigma_{A,B}\Vert\omega_{A,B}^{(t+1)})+D(\rho_{A,B}^{(t)}\Vert\omega_{A,B}^{(t+1)})\right)\\
         \nonumber
        &\leq\frac{1}{\eta}\left(D(\sigma_{A,B}\Vert\rho_{A,B}^{(t)})-D(\sigma_{A,B}\Vert\rho_{A,B}^{(t+1)})+D(\rho_{A,B}^{(t)}\Vert\omega_{A,B}^{(t+1)})\right),
         \nonumber
    \end{align}
    where for the inequality we have used \eqref{eq-rel_ent_Pythagorean_2}, and we also made use of the fact that
    \begin{equation}
        \log\omega_{A,B}^{(t+1)}=\log W_{A,B}^{(t+1)}-\log(\Tr[W_{A,B}^{(t+1)}])\mathbbm{1}_{A,B},
    \end{equation}
    which means that $D(\sigma_{A,B}\Vert\omega_{A,B}^{(t+1)})=D(\sigma_{A,B}\Vert W_{A,B}^{(t+1)})+\log(\Tr[W_{A,B}^{(t+1)}])$. Now, let us bound $D(\rho_{A,B}^{(t)}\Vert\omega_{A,B}^{(t+1)})$. Consider that
    \begin{align}
        &D(\rho_{A,B}^{(t)}\Vert\omega_{A,B}^{(t+1)})+D(\omega_{A,B}^{(t+1)}\Vert\rho_{A,B}^{(t)})\\
         \nonumber&\qquad=\Tr\!\left[(\rho_{A,B}^{(t)}-\omega_{A,B}^{(t+1)})\left(\log\rho_{A,B}^{(t)}-\log\omega_{A,B}^{(t+1)}\right)\right]\\
          \nonumber
        &\qquad=\Tr\!\left[(\rho_{A,B}^{(t)}-\omega_{A,B}^{(t+1)})\left(\log\rho_{A,B}^{(t)}-\log W_{A,B}^{(t+1)}+\log(\Tr[W_{A,B}^{(t+1)}])\mathbbm{1}_{A,B}\right)\right]\\
         \nonumber
        &\qquad=\Tr\!\left[\left(\log\rho_{A,B}^{(t)}-\log W_{A,B}^{(t+1)}\right)(\rho_{A,B}^{(t)}-\omega_{A,B}^{(t+1)})\right]+\log(\Tr[W_{A,B}^{(t+1)}])\underbrace{\Tr[\rho_{A,B}^{(t)}-\omega_{A,B}^{(t+1)}]}_{=0}\\
         \nonumber
        &\qquad=\eta\Tr[L_{A,B}^{(t)}(\rho_{A,B}^{(t)}-\omega_{A,B}^{(t+1)})]\\
         \nonumber
        &\qquad\leq\eta\norm{L_{A,B}^{(t)}}_{\infty}\norm{\rho_{A,B}^{(t)}-\omega_{A,B}^{(t+1)}}_{1},
         \nonumber
    \end{align}
    where we have used the H\"{o}lder inequality in the final line. Let us now use the fact that $(x-y)^2\geq 0\Rightarrow xy\leq\frac{1}{2}x^2+\frac{1}{2}y^2$ for all $x,y\in\mathbb{R}$. Letting $x\equiv \eta\norm{L_{A,B}^{(t)}}_{\infty}$ and $y\equiv\norm{\rho_{A,B}^{(t)}-\omega_{A,B}^{(t+1)}}_1$, and using the Pinsker inequality \eqref{eq-pinsker_inequality}, we obtain
    \begin{align}
        D(\rho_{A,B}^{(t)}\Vert\omega_{A,B}^{(t+1)})+D(\omega_{A,B}^{(t+1)}\Vert\rho_{A,B}^{(t)})&\leq\frac{\eta^2}{2}\norm{L_{A,B}^{(t)}}_{\infty}^2+\frac{1}{2}\norm{\rho_{A,B}^{(t)}-\omega_{A,B}^{(t+1)}}_1^2\\
         \nonumber
        &\leq\frac{\eta^2}{2}\norm{L_{A,B}^{(t)}}_{\infty}^2+D(\omega_{A,B}^{(t+1)}\Vert\rho_{A,B}^{(t)}),
    \end{align}
    which implies that
    \begin{equation}
        D(\rho_{A,B}^{(t)}\Vert\omega_{A,B}^{(t+1)})\leq\frac{\eta^2}{2}\norm{L_{A,B}^{(t)}}_{\infty}^2,
    \end{equation}
    for all $t\in\{1,2,\dotsc,T\}$. Altogether, we have
    \begin{equation}
        \Tr[L_{A,B}^{(t)}\rho_{A,B}^{(t)}]-\Tr[L_{A,B}^{(t)}\sigma_{A,B}]\leq\frac{1}{\eta}\left(D(\sigma_{A,B}\Vert\rho_{A,B}^{(t)})-D(\sigma_{A,B}\Vert\rho_{A,B}^{(t+1)})\right)+\frac{\eta}{2}\norm{L_{A,B}^{(t)}}_{\infty}^2,
    \end{equation}
    for all $t\in\{1,2,\dotsc,T\}$. Summing over all $t$, we obtain
    \begin{align}
        &\sum_{t=1}^T\Tr[L_{A,B}^{(t)}\rho_{A,B}^{(t)}]-\Tr\!\left[\sigma_{A,B}\left(\sum_{t=1}^TL_{A,B}^{(t)}\right)\right]\nonumber\\
         \nonumber
        &\qquad\leq\frac{\eta}{2}\sum_{t=1}^T\norm{L_{A,B}^{(t)}}_{\infty}^2+\frac{1}{\eta}\left(D(\sigma_{A,B}\Vert\rho_{A,B}^{(1)})-D(\sigma_{A,B}\Vert\rho_{A,B}^{(T+1)})\right)\\
        &\qquad\leq\frac{\eta}{2}\sum_{t=1}^T\norm{L_{A,B}^{(t)}}_{\infty}^2+\frac{1}{\eta}D(\sigma_{A,B}\Vert\rho_{A,B}^{(1)}),
    \end{align}
    where the second inequality is due to the fact that $D(\rho\Vert\sigma)\geq 0$ for all density operators $\rho$ and $\sigma$. After substituting $\rho_{A,B}^{(1)}=\frac{1}{d_Ad_B}\mathbbm{1}_{A,B}$, we obtain the desired result.
\end{proof}

\begin{remark}[Extending Proposition~\ref{prop-projected_MMW_Choi_state_regret_bound_agile} to multi-time processes]\label{rem:projected_MMW_choi_state_multi_time}
    The key elements of the proof of Proposition~\ref{prop-projected_MMW_Choi_state_regret_bound} are the fact that we project onto a convex set in \eqref{eq-rel_ent_CPTP_proj}, such that the inequality in \eqref{eq-rel_ent_Pythagorean_2} holds, and the Pinsker inequality in \eqref{eq-pinsker_inequality}. Consequently, it is straightforward to generalize Algorithm~\ref{alg:projected_mmw_choi_state_2}, and thus Proposition~\ref{prop-projected_MMW_Choi_state_regret_bound}, to the Choi states of multi-times processes. In particular, letting $\mathsf{COMB}_r\equiv\mathsf{COMB}_r(A_1,\dotsc,A_r;B_1,\dotsc,B_r)$ be the set of multi-time processes with $r\in\{1,2,\dotsc\}$ time steps, as given in Definition~\ref{def-multi_time_process}, we define the relative entropy projection onto this set as
    \begin{equation}\label{eq-rel_ent_comb_proj}
        \Pi(\sigma)\coloneqq\frac{1}{d_A}\argmin_{P}\Big\{D\Big(\frac{1}{d_A}P\Big\Vert\sigma\Big):P\in\mathsf{COMB}_r\Big\},
    \end{equation}
    for all $\sigma\in\Lin(\mathcal{H}_{A,B}^{(r)})$, $\sigma\geq 0$, where $d_A\equiv\prod_{k=1}^r d_{A_k}$. Then, by replacing step 6 in Algorithm~\ref{alg:projected_mmw_choi_state_2} with the projection in \eqref{eq-rel_ent_comb_proj}, we obtain a mirror descent algorithm for Choi states of multi-time quantum processes. Then, the analogue of Proposition~\ref{prop-projected_MMW_Choi_state_regret_bound} is as follows. If $\sigma=\frac{1}{d_A}Q$, with $Q\in\mathsf{COMB}_r$, is an arbitrary Choi state of a multi-time process with $r$ steps, $L^{(1)},L^{(2)},\dotsc,L^{(T)}\in\Lin(\mathcal{H}_{A,B}^{(r)})$ are cost matrices satisfying $-\mathbbm{1}_{d_Ad_B}\leq L^{(t)}\leq\mathbbm{1}_{d_Ad_B}$, with $d_B\equiv\prod_{k=1}^r d_{B_k}$, and $\rho^{(1)},\rho^{(2)},\dotsc,\rho^{(T)}$ are the projected Choi state updates resulting from the algorithm, then
    \begin{equation}\label{eq-projected_MMW_Choi_state_multi_time_regret_bound}
        \sum_{t=1}^T\Tr[L^{(t)}\rho^{(t)}]\leq\Tr\!\left[\sigma\left(\sum_{t=1}^TL^{(t)}\right)\right]+\frac{\eta}{2}\sum_{t=1}^T\norm{L^{(t)}}_{\infty}^2+\frac{\log(d_Ad_B)-H(\sigma)}{\eta},
    \end{equation}
    which is directly analogous to \eqref{eq-projected_MMW_Choi_state_regret_bound_agile}.~\qedgen
\end{remark}

\end{appendix}

\end{document}